\documentclass[11pt]{article}
\usepackage[letterpaper=true,colorlinks=true,pdfpagemode=none,linkcolor=blue,citecolor=blue,pdfstartview=FitH]{hyperref}

\usepackage[OT4]{fontenc}
\usepackage{fullpage}
\usepackage{amsfonts}
\usepackage{amssymb}
\usepackage{amsthm}
\usepackage{latexsym}
\usepackage{amsmath}
\usepackage{enumerate}
\usepackage{color}

\newtheorem{thm}{Theorem}

\newtheorem{obs}[thm]{Observation}
\newtheorem{lem}[thm]{Lemma}
\newtheorem{lemma}[thm]{Lemma}

\newtheorem{claim}[thm]{Claim}

\newtheorem{defn}[thm]{Definition}%

\newenvironment{Proof}{\begin{proof}}{\end{proof}}

\newcommand{\oeta}{{\ensuremath{{\overline{\eta}}}}}

\newcommand{\eps}{{\ensuremath{{\varepsilon}}}}

\newcommand{\tN}{{\ensuremath{{\tilde{N}}}}}

\newcommand{\oj}{{\ensuremath{{\overline{j}}}}}

\newcommand{\os}{{\ensuremath{{\overline{s}}}}}

\newcommand{\tS}{{\ensuremath{{\widetilde{S}}}}}
\newcommand{\ooi}{{\ensuremath{{i^t}}}}

\newcommand{\oi}{{\ensuremath{{\overline{i}}}}}

\newcommand{\DelPhi}{{\ensuremath{{\Delta\Phi}}}}

\newcommand{\HitAlg}{{\ensuremath{{\mathrm{H}}}}}

\newcommand{\MovAlg}{{\ensuremath{{\mathrm{M}}}}}
\newcommand{\HitOPT}{{\ensuremath{{\mathrm{H^*}}}}}

\newcommand{\MovOPT}{{\ensuremath{{\mathrm{M^*}}}}}
\newcommand{\opt}{{\mathrm{Opt}}}

\newcommand{\Alg}{{\ensuremath{{\mathrm{Alg}}}}}
\newcommand{\optcost}{{\mathrm{Optcost}}}
\newcommand{\hc}{{\mathrm{Hitcost}}}

\newcommand{\oy}{{\overline{y}}}

\newcommand{\ox}{{\overline{x}}}
\newcommand{\cB}{{\mathcal{D}}}
\newcommand{\Ex}{{\mathbb E}}

\begin{document}
\title{{\bf A Polylogarithmic-Competitive Algorithm for the $k$-Server Problem}}
\author{
Nikhil Bansal
\thanks{Technical University of Eindhoven, Netherlands.
E-mail: n.bansal@tue.nl.}
\and
 Niv Buchbinder
 \thanks{Computer Science Department, Open University of Israel. E-mail: niv.buchbinder@gmail.com. Supported by ISF grant 954/11 and BSF grant 2010426.}
\and Aleksander M\k{a}dry
\thanks{Microsoft Research, Cambridge, MA USA. E-mail: madry@mit.edu. Research done while at the Computer Science and Artificial Intelligence Laboratory, MIT, Cambridge, MA USA, and partially supported by NSF grant CCF-0829878 and by ONR grant N00014-11-1-0053.}
\and
Joseph (Seffi) Naor
\thanks{Computer Science Department, Technion, Haifa, Israel. E-mail: naor@cs.technion.ac.il. Supported by ISF grant 954/11 and BSF grant 2010426.}}

\date{}
\maketitle

\begin{abstract}
We give the first polylogarithmic-competitive randomized online algorithm
for the $k$-server problem on an arbitrary finite metric space. In
particular, our algorithm achieves a competitive ratio of
$\widetilde{O}(\log^3 n \log^2 k)$ for any metric space on $n$
points. Our algorithm improves upon the deterministic $(2k-1)$-competitive
algorithm of Koutsoupias and Papadimitriou \cite{KP95} whenever
$n$ is sub-exponential in $k$.
\end{abstract}

\section{Introduction}

The $k$-server problem is one of the most fundamental and
extensively studied problems in online computation. Suppose there
is an $n$-point metric space and $k$ servers are located at some
of the points of the metric space. At each time step, an online
algorithm is given a request at one of the points of the metric
space, and this request is served by moving a server to the
requested point (if there is no server there already). The cost of
serving a request is defined to be the distance traveled by the
server. Given a sequence of requests, the task is to devise an
online strategy minimizing the sum of the costs of serving the
requests.

The $k$-server problem was originally proposed by Manasse et
al.~\cite{MMS90} as a broad generalization of various online
problems. The most well studied problem among them is the paging
(also known as caching) problem, in which there is a cache that
can hold up to $k$ pages out of a universe of $n$ pages. At each
time step a page is requested; if the page is already in the cache
then no cost is incurred, otherwise it must be brought into the
cache (possibly causing an eviction of some other page) at a cost
of one unit. It is easy to see that the paging problem is
equivalent to the $k$-server problem on a uniform metric space,
and already in their seminal paper on competitive analysis,
Sleator and Tarjan \cite{ST85} gave several $k$-competitive
algorithms for paging, and showed that no deterministic algorithm
can do better. This prompted Manasse et al.~\cite{MMS90} to state
a far-reaching conjecture that a similar result holds for an
arbitrary metric. More precisely, they conjectured that there is a
$k$-competitive online algorithm for the $k$-server problem on any
metric space and for any value of $k$. This conjecture is known as
as the {\em $k$-server conjecture}.

At the time that the $k$-server conjecture was stated,  an online
algorithm with competitive ratio that depends only on $k$ was not
known. It was first obtained by Fiat et al.~\cite{FRR90}. Improved
bounds were obtained later on by \cite{G91,BG00}, though the ratio
still remained exponential in $k$. A major breakthrough was
achieved by Koutsoupias and Papadimitriou \cite{KP95}, who showed
that  so-called {\em work function algorithm} is
$(2k-1)$-competitive. This result is almost optimal, since we know
that any deterministic algorithm has to be at least
$k$-competitive. We note that a tight competitive factor of $k$ is
only known for special metrics such as the uniform metric, line
metric, and -- more generally -- trees \cite{C+,CL91}.

Even though the aforementioned results are all deterministic,
there is also a great deal of interest in randomized algorithms
for the $k$-server problem. This is motivated primarily by the
fact that randomized online algorithms (i.e., algorithms working
against an oblivious adversary) tend to have much better
performance than their deterministic counterparts. For example,
for the paging problem, several $O(\log k)$-competitive algorithms
are known ~\cite{F+,GS,ACN,BBN07}, as well as a lower bound of
$\Omega(\log k)$ on the competitive ratio.

Unfortunately, our understanding of the $k$-server problem when
randomization is allowed is much weaker than in the deterministic
case. Despite much work \cite{BKRS92,BBM01,BLMN}, no better lower
bound than $\Omega(\log k)$ is known on competitive factors in the
randomized setting. Conversely, no better upper bound, other than
the {\em deterministic} guarantee of $2k-1$ \cite{KP95} mentioned
above, is known for general metrics. Thus, an exponential gap
still remains between the known lower and upper bounds.

Given the lack of any lower bounds better than $\Omega(\log k)$,
it is widely believed that there is an $O(\log k)$-competitive
randomized algorithm for the $k$-server problem on every metric
space against an oblivious adversary - this belief is captured by
the so-called {\em randomized $k$-server conjecture}.
Unfortunately, besides the previously-mentioned $O(\log
k)$-competitive algorithm for the case of a uniform metric, even
when we allow the competitiveness to depend on other parameters of
the metric, such as the number of points $n$, or the diameter
$\Delta$, non-trivial guarantees are known only for very few
special cases. For example, the case of a well-separated metric
\cite{Seiden}, the case of a metric corresponding to a binary HST
with high separation \cite{CMP08}, the case of $n=k+O(1)$
\cite{BBBT}, as well as some other cases
\cite{CL06,BBN10a,BBN10b}. For the weighted paging
problem\footnote{In weighted paging, arbitrary weights are
associated with fetching the pages into the cache. This problem
corresponds to the $k$-server problem on a weighted star.},
\cite{BBN07} gave an $O(\log k)$-competitive algorithm (see also
\cite{BBN10a}) which is based on the online primal-dual approach.
However, no non-trivial guarantees are known  even for very simple
extensions of the uniform metric, e.g., two-level HSTs with high
separation.

For a more in-depth treatment of the extensive literature on both paging and the $k$-server problem, we suggest \cite{BE98}.

\subsection{Our Result}

We give the first polylogarithmic competitive algorithm for the
$k$-server problem on a general metric with a finite number of
points $n$. More precisely, our main result is the following.
\begin{thm}
\label{thm:main}
There is a randomized algorithm for the $k$-server problem that achieves a competitive ratio of
$O(\log^2 k \log^3 n \log \log n) = \widetilde{O}(\log^2 k \log^3 n)$
on any metric space on $n$ points.

\end{thm}

The starting point of our algorithm is the recent approach
proposed by Cot\'{e} et al.~\cite{CMP08} for solving the
$k$-server problem on hierarchically well-separated trees (HSTs).
It is well known that solving the problem on HSTs suffices, as any
metric space can be embedded into a probability distribution over
HSTs with low distortion \cite{FRT03}.

More precisely, Cot\'{e} et al. defined a problem on uniform
metrics which we call the {\em allocation problem}. They showed
that an online randomized algorithm for the allocation problem
that provides certain refined competitive guarantees can be used
as a building block to recursively solve the $k$-server problem on
an HST, provided the HST is sufficiently well-separated. Roughly
speaking, in their construction, each internal node of the HST
runs an instance of the allocation problem that determines how to
distribute the available servers among its children nodes.
Starting from the root, which has $k$ servers, the recursive calls
to the allocation problem determine the number of servers at each
leaf of the HST, giving a valid $k$-server solution. The guarantee
of this $k$-server solution depends on both the guarantees for the
allocation problem, as well as the depth of the HST (i.e., the
number of levels of recursion). The guarantees obtained by
Cot\'{e} et al.~\cite{CMP08} for the allocation problem on a
metric space with two points allowed them to obtain an algorithm
for the $k$-server problem on a sufficiently well-separated {\em
binary} HST having a competitive ratio that is polylogarithmic in
$k$, $n$, and the diameter $\Delta$ of the underlying metric
space. Unfortunately, the fact that the HST has to be binary as
well as have a sufficiently good separation severely restricts the
metric spaces to which this algorithm can be applied.

Given the result of Cot\'{e} et al.~\cite{CMP08}, a natural
approach to establishing our result is coming up with  a
randomized algorithm having the required refined guarantees for
the allocation problem on an arbitrary number of points. However,
it is unclear to us how to obtain such an algorithm. Instead, we
pursue a more refined approach to solving the $k$-server problem
via the allocation problem. By doing so we are able to bypass the
need for a ``true" randomized algorithm for the allocation problem
and are able to work with a (much) weaker formulation. More
precisely, our result consists of three main parts.

\begin{enumerate}
\item We show that instead of obtaining a randomized algorithm for
the allocation problem, it suffices to obtain an algorithm for a
certain {\em fractional} relaxation of it.
Employing this relaxation makes the task of designing such a
fractional allocation algorithm easier than designing the version
of the allocation problem that was considered earlier. Next,
building upon the arguments in Cot\'{e} et al. \cite{CMP08}, we
show that a sufficiently good online algorithm for this fractional
allocation problem can be used as a building block to obtain a
good {\em fractional} solution to the $k$-server problem on an
HST. Finally, by proving that such a fractional $k$-server
solution can be rounded in an online randomized manner, while
losing only an $O(1)$ factor in the competitive ratio, we get a
reduction of the $k$-server problem to our version of the
fractional allocation problem.

An interesting feature of this reduction is that our fractional
relaxation is too weak to give anything better than an $O(k)$
guarantee for the (integral) allocation problem, since there are
instances on which any integral solution must pay $\Omega(k)$
times the fractional cost. Therefore, it is somewhat surprising
that even though our relaxation is unable to provide any
reasonable algorithm for the (integral) allocation problem, it
suffices to give a good guarantee for the (integral) $k$-server
problem.

\item As the next step, we design an online algorithm for the
fractional allocation problem with the refined guarantees required
in the above reduction. Our techniques here are inspired by the
ideas developed recently in the context of the caching with costs
problem \cite{BBN10a} and weighted paging \cite{BBN07}.
However, while these previous algorithms were designed and
described using the online primal-dual framework, our algorithm is
combinatorial. To analyze the performance we employ a novel
potential function approach.

By plugging the algorithm for the fractional allocation problem
into the above reduction, we get a (roughly) $O(\ell
\log(k\ell))$-competitive algorithm for the $k$-server problem on
an HST of depth $\ell$, provided that the HST is sufficiently
well-separated.

\item Finally, we note that the competitive guarantee provided by
the above $k$-server algorithm depends on the depth $\ell$ of the
HST we are working with and, as $\ell$ can be $\Omega(\log
\Delta)$, this guarantee can be polylogarithmic in $\Delta$.
Therefore, as $\Delta$ can be $2^{\Omega(n)}$, this would lead to
competitiveness that is even polynomial in $n$.\footnote{To see an
example when this is the case, one could consider a metric space
corresponding to $n$ points on a line that are spaced at
geometrically increasing distances.} To deal with this issue, we
define a weighted version of an HST in which the edge lengths on
any root-to-leaf path still decrease at (at least) an exponential
rate, but the lengths of the edges from a given node to its
children could be non-uniform. We prove that any HST can be
transformed to a weighted HST of depth $\ell=O(\log n)$ while
incurring only an $O(1)$ distortion in leaf-to-leaf distances. We
then show that our previous ideas can be applied to
 weighted HSTs as well. In particular, our online
fractional allocation algorithm is actually developed for a
weighted star metric (instead of a uniform one), and, as we show,
it can be employed in our reduction to obtain a fractional
$k$-server algorithm on a weighted HST. The fractional $k$-server
algorithm can again be rounded to a randomized algorithm with only
an $O(1)$ factor loss. Since $\ell$ is now $O(\log n)$ and thus
does not depend on $\Delta$, it gives us an overall guarantee
which is polylogarithmic only in $n$ and $k$.

\end{enumerate}

In Section \ref{s:overview} we describe the above ideas more
formally and also give an overview of the paper.

\subsection{Preliminaries}

We provide definitions and concepts that will be needed in the paper.

\paragraph{Hierarchically well-separated trees.}
Hierarchical well-separated trees (HST-s), introduced by Bartal
\cite{Ba96,Ba98}, is a metric embedding technique in which a
general metric is embedded into a probability distribution defined
over a set of structured trees (the HST-s). The points of the
metric are mapped onto the leaves of the HST, while internal tree
nodes represent clusters. The distances along a root-leaf path
form a geometric sequence, and this factor is called the {\em
stretch} of the HST. An HST with stretch $\sigma$ is called a
$\sigma$-HST. The embedding guarantees that the distance between
any pair of vertices in the metric can only increase in an HST,
and the expected blowup of each distance, known as the {\em
distortion}, is bounded. It is well known that any  metric on $n$
points can be embedded into a distribution over $\sigma$-HSTs with
distortion $O(\sigma \log_{\sigma} n)$ \cite{FRT03}. This approach
of embedding into HSTs is particularly useful for problems (both
offline and online) which seem hard on a general metric, but can
be solved fairly easily on trees (or HSTs).

Due to the special structure of HSTs, the task of solving problems
on them can sometimes be reduced to the task of solving a more
general (and thus harder) problem, but on a uniform metric. For
example, this approach was used to obtain the first
polylogarithmic guarantees for the {\em metrical task systems}
problem (MTS) by \cite{BBBT} (later further refined by \cite{FM}).
More precisely, Blum et al.~\cite{BBBT} defined a refined version
of MTS on a uniform metric known as {\em unfair}-MTS and showed
how an algorithm with a certain refined guarantee for it can be
used recursively to obtain an algorithm for MTS on an HST. This
approach is especially appealing in the context of the $k$-server
problem, as this problem on a uniform metric (i.e.~paging) is
well- understood. This motivated Cot\'{e} et al.~\cite{CMP08} to
define a problem on a uniform metric, that we call the allocation
problem, and show how a good algorithm for it can be used to
design good $k$-server algorithms on HSTs. This problem is defined
as follows.

\paragraph{The allocation problem.}
Suppose that a metric on $d$ points is defined by a weighted star
in which the distance from the center to each point $i$, $1 \leq i
\leq d$, is $w_i$.\footnote{Even though Cot\'{e} et al.
\cite{CMP08} considered the allocation problem on a uniform
metric, we find it useful to work with the more general
weighted-star metric version of this problem.}  At time step $t$,
the total number of available servers, $\kappa(t) \leq k$, is
specified, and we call the vector
$\kappa=(\kappa(1),\kappa(2),\ldots)$ the {\em quota pattern}. A
request arrives at a point $\ooi$ and it is specified by a
$(k+1)$-dimensional vector $\vec{h}^t=(h^t(0), h^t(1), \ldots,
h^t(k))$, where $h^t(j)$ denotes the cost of serving the request
using $j$ servers. The cost vectors at any time are guaranteed to
satisfy the following {\em monotonicity} property: for any $0\leq
j \leq k-1$, the costs satisfy $h^t(j) \geq h^t(j+1)$. That is,
serving a request with more servers cannot increase the cost. Upon
receiving a request, the algorithm may choose to move additional
servers from other locations to the requested point and then serve
it. The cost is divided into two parts. The {\em movement cost}
incurred for moving the servers, and the {\em hit cost} determined
by the cost vector and the number of servers at location $\ooi$.

In this paper, we will be interested in designing algorithms for
(a fractional version of) this problem that provide a certain
refined competitive guarantee. Namely, we say that an online
algorithm for the allocation problem is {\em
$(\theta,\gamma)$-competitive} if it incurs:
\begin{itemize}
\item a hit cost of at most $\theta \cdot (\optcost+\Delta\cdot
g(\kappa))$; \item a movement cost of at most $\gamma \cdot
(\optcost+\Delta\cdot g(\kappa))$,
\end{itemize}
where $\optcost$ is the total cost (i.e., hit cost plus movement
cost) of an optimal solution to a given instance of the allocation
problem, $g(\kappa):=\sum_{t} |\kappa(t)-\kappa(t-1)|$ is the
total variation of the server quota pattern, and $\Delta$ is the
diameter of the underlying metric space.

\paragraph{From allocation to $k$-server.}
Cot\'{e} et al.~\cite{CMP08} showed that a
$(1+\varepsilon,\beta)$-competitive online algorithm for the
allocation problem on $d$ points -- provided $\eps$ is small
enough and $\beta = O_{\eps}(\textrm{polylog}(d,k))$ -- can be
used to obtain a polylogarithmic competitive algorithm for the
$k$-server problem on general metrics. In particular, the next
theorem follows from their work and is stated explicitly in
\cite{BBN10a}.
\begin{thm}\label{thm:mey}
Suppose there is a $(1+\eps,\beta)$-competitive algorithm for the
allocation problem on a uniform metric on $d$ points. Let $H$ be
an $\sigma$-HST with depth $\ell$. Then, for any $\eps \leq
1$, there is an $ O(\beta \gamma^{\ell+1}/(\gamma-1))$-competitive
algorithm for the $k$-server problem on $H$, where
$$ \gamma = (1+\varepsilon)\left(1+\frac{3}{\sigma}\right) + O\left(\frac{\beta}{\sigma}\right).$$
Setting $\eps=1/\ell$, this gives an $O(\beta \ell)$-competitive
algorithm on $\sigma$-HSTs, provided the HST separation parameter
$\sigma$ is at least $\beta \ell$.
\end{thm}

At a high level, the $k$-server algorithm in Theorem \ref{thm:mey}
is obtained as follows.  Each internal node $p$ in the HST runs an
instance of the allocation problem on the uniform metric formed by
its children. In this instance, the cost vectors appearing at a
child $i$ are guided by the evolution of the cost of the optimal
solution to the instance of the $k$-server problem restricted to
the leaves of the subtree that is rooted at $i$. Furthermore, the
quota patterns for each of the allocation problem instances is
determined recursively. The root of the tree has a fixed server
quota of $k$, and the quota corresponding to a non-root node $i$
is specified by the number of servers that are allocated to $i$ by
the instance of the allocation problem run at the parent of $i$.
The distribution of the servers on the leaves of the tree is
determined in this manner, thus leading to a solution to the
$k$-server problem. The overall guarantee in Theorem \ref{thm:mey}
follows roughly by showing that the hit cost guarantee of
($1+\eps$) multiplies at each level of the recursion, while
the movement cost guarantee of $\beta$ adds up.

\paragraph{Weighted HSTs.} Note that the guarantee in
Theorem \ref{thm:mey} depends on $\ell$, the depth of the
$\sigma$-HST, which in general is $\Theta(\log_\sigma \Delta)$. To
avoid the resulting dependence on the diameter $\Delta$ that can
be as large as $2^{\Omega(n)}$, we introduce the notion of a {\em
weighted} $\sigma$-HST. A weighted $\sigma$-HST is a tree having
the property that for any node $p$, which is not the root or a
leaf, the distance from $p$ to its parent is at least $\sigma$
times the distance from $p$ to any of its children. Thus, unlike
an HST, distances from $p$ to its children can be non-uniform. The
crucial property of weighted HSTs that we will show later is that
any $\sigma$-HST $T$  with $O(n)$ nodes can be embedded into a
weighted $\sigma$-HST with depth $O(\log n)$, such that the
distance between any pair of leaves of $T$ is distorted by a
factor of at most $2 \sigma/(\sigma-1)$ (which is $O(1)$ if, say,
$\sigma\geq 2$). Reducing the depth from $O(\log \Delta)$ to
$O(\log n)$ allows us to replace the factor of $\log \Delta$ by
$\log n$ in the bound on the competitive factor we get for the
$k$-server problem.

\paragraph{Fractional view of randomized algorithms.}
The relation between randomized algorithms and their corresponding
fractional views is an important theme in our paper. By
definition, a randomized algorithm is completely specified by the
probability distribution over the configurations (deterministic
states) at each time step of the algorithm. However, working
explicitly with such distributions is usually very cumbersome and
complex, and it is often simpler to work with a {\em fractional
view} of it. In a fractional view, the algorithm only keeps track
of the marginal distributions on certain quantities, and specifies
how these marginals evolve with time. Note that there are multiple
ways to define a fractional view (depending on which marginals are
tracked). For example, for the $k$-server problem on an HST, the
fractional view might simply correspond to specifying the
probability $p_i$ of having a server at leaf $i$ (instead of
specifying the entire distribution on the $k$-tuples of possible
server locations). Clearly, the fractional view is a lossy
representation of the actual randomized algorithm. However, in
many cases (though not always), a fractional view can be converted
back to a randomized algorithm with only a small loss. We now
describe the fractional views we employ for the two main problems
considered in this paper.

\paragraph{Fractional view of the $k$-server problem on an HST.}

Let $T$ be a $\sigma$-HST. For a node $j \in T$, let $T(j)$ be the
set of leaves in the subtree rooted at $j$. In the fractional
view, at each time step $t$, the probability of having a server at
leaf $i$, denoted by $p^t_i$, is specified. Upon getting a request
at leaf $i$ at time $t$, a fractional algorithm must ensure that
$p^t_i=1$. Let the expected number of servers at time $t$ at
leaves of $T(j)$ be denoted by $k^t(j) = \sum_{i\in T(j)}
p^t_{i}$. Clearly, the movement cost incurred at time $t$ is $
 \sum_{j \in T} W(j) |k^{t}(j) - k^{t-1}(j)|,$ where
$W(j)$ is the distance from $j$ to its parent in $T$.

It is easy to verify that the cost incurred by any randomized
algorithm is at least as large as the cost incurred by its induced
fractional view. Conversely, it turns out that the fractional view
is not too lossy for a $\sigma$-HST  (provided $\sigma>5$). In
particular, in Section \ref{s:rounding} we show that for a
$\sigma$-HST ($\sigma>5$), an online algorithm for the $k$-server
problem in the fractional view above can be converted into an
online randomized algorithm, while losing only an $O(1)$ factor in
the competitive ratio.

\paragraph{The fractional allocation problem.}
For the allocation problem we consider the following fractional
view. For each location $i \in [d]$, and all possible number of
servers $j \in \{0,\ldots,k\}$, there is a variable $x^t_{i,j}$
denoting the probability of having {\em exactly} $j$ servers at
location $i$ at time $t$. For each time $t$, the variables
$x^t_{i,j}$ must satisfy the following constraints.

\begin{enumerate}
\item For each location $i$, the variables $x^t_{i,j}$ specify a
probability distribution, i.e., $\sum_j x^t_{i,j}=1$ and each $x_{i,j}^t$ is non-negative.
 \item The
number of servers used is at most $\kappa(t)$, the number of available
servers. That is,
 $$\sum_i \sum_j  j \cdot x^t_{i,j} \leq  \kappa(t).$$
\end{enumerate}
At time step $t$, when cost vector $h^t$ arrives at location
$\ooi$, and possibly $\kappa(t)$ changes, the algorithm can change
its distribution from $\ox^{t-1}$ to $\ox^t$ incurring a hit cost
of $\sum_j h^t(j) x^t_{\ooi,j}$. The movement cost incurred is
defined to be
\begin{eqnarray}
\sum_{i} w_i \sum_{j=1}^{k}  \left| \sum_{j' < j} x^t_{i,j'} -
\sum_{j'<j} x^{t-1}_{i,j'} \right|. \label{eq:move_cost}
\end{eqnarray}

{\em Remark:} Note that when our configurations are integral, this
quantity is exactly the cost of moving the servers from
configuration $\ox^{t-1}$ to configuration $\ox^t$. %
In the fractional case, each term in the outermost sum can be seen as equal to the {\em earthmover distance} between the probability vectors
$(x_{i,0}^{t-1},\ldots,x_{i,k}^{t-1})$ and $(x_{i,0}^{t},\ldots,x_{i,k}^{t})$ with respect to a linear metric defined on
$\{0,1, \ldots, k \}$.
The earthmover distance is the optimal solution to a
transportation problem in which $x^{t-1}$ is the supply vector,
$x^t$ is the demand vector, and the cost of sending one unit of
flow between $x^{t-1}_{i,j}$ and $x^t_{i,j'}$ is $w_i\cdot
|j-j'|$, since $|j-j'|$ is the change in number of servers
resulting from sending this unit of flow. It is not hard to
see\footnote{Using uncrossing arguments on the optimal
transportation solution.} that in the case of a linear metric, the
optimal solution to the transportation problem (up to a factor of
2) is captured by (\ref{eq:move_cost}).

\paragraph{A gap instance for the fractional allocation problem.}
As mentioned earlier, unlike the fractional view of the $k$-server
problem presented above, the fractional view of the allocation
problem turns out to be too weak to yield a randomized algorithm
for its integral counterpart. We thus present an instance of the
allocation problem for which the ratio between the cost of any
integral solution and the cost of an optimal fractional solution
is $\Omega(k)$. However, we stress that even though the fractional
view fails to approximate the integral allocation problem, we are
still able to use it to design a fractional (and, in turn,
integral) solution to the $k$-server problem. In particular, we
show in Section \ref{s:alloc_to_kserver} that Theorem
\ref{thm:mey} holds even when we substitute the randomized
algorithm for the allocation problem with the fractional
algorithm.

Let us consider a uniform metric space over $d=2$ points, and
consider an instance of the allocation problem in which exactly
$\kappa(t)=k$ servers are available at each time. Furthermore, at
each odd time step $1,3,5,\ldots$, the cost vector
$h=(1,1,...,1,0)$ arrives at location $1$, and at each even time
step $2,4,6,\ldots$, the vector $h'=(1,0,0,...,0)$ arrives at
location $2$.

We show that any integral solution to this instance of the
allocation problem must incur a high cost, while there is an
$\Omega(k)$ times cheaper solution in the fractional view.

\begin{claim} Any solution to the instance above incurs a cost of $\Omega(T)$ over $T$ time steps.
\end{claim}
\begin{proof}
Observe that the hit cost can be avoided at location $1$ only if
it contains $k$ servers, and it can be avoided at location $2$
only if it contains at least one server. Thus, any algorithm that
does not pay a hit cost of at least $1$ during any two consecutive
time steps, must move at least one server between locations 1 and
2, incurring a movement cost of at least $1$, concluding that the
cost is $\Omega(T)$.
\end{proof}
\begin{claim} There is a solution in the fractional  view of cost $O(T/k)$ over $T$ time steps.
\end{claim}
\begin{proof}
Consider the following solution in the fractional view. At each
time step $t$, let:
$$ x^t_{1,0}= \frac{1}{k}, \quad \quad  x^t_{1,k}=1-\frac{1}{k}, \qquad \textrm {and} \qquad   x^t_{2,1}=1.$$
Note that this solution satisfies all the constraints in the
fractional view. Since location $2$ always has a server, it never
pays any hit cost. Moreover, location $1$ has fewer than $k$
servers with probability $1/k$, it thus incurs only a hit cost
$1\cdot x_{1,0} = 1/k$ at every odd time step. Also, as the
solution does not change over time, the movement cost is $0$.
\end{proof}

\section{Overview of Our Approach}
\label{s:overview}

In this section we give a formal description of our results,
outline how they are organized, and discuss how they fit together
so as to obtain our main result.

\paragraph{Fractional allocation algorithm.} In Section \ref{s:frac} we consider the
fractional allocation problem on a weighted star, and prove the
following theorem.

\begin{thm}\label{thm:frac_ov}
For any $\eps>0$, there exists a {\em fractional} $(1+\eps, O(\log (k/\eps)))$-competitive allocation algorithm on a weighted star metric.
\end{thm}

\paragraph{From allocation to $k$-server problem.}
In Section \ref{s:alloc_to_kserver} we show how the algorithm from
Theorem \ref{thm:frac_ov} can be used to obtain a fractional
$k$-server algorithm on a sufficiently well-separated weighted
HST. In particular, we show that:

\begin{thm}
\label{thm:alloc_to_kserver} Let $T$ be a weighted $\sigma$-HST of
depth $\ell$. If, for any $0\leq \eps \leq 1$, there exists a
$(1+\eps,\log (k/\eps))$-competitive algorithm for the fractional
allocation problem on a weighted star, then there is an $O(\ell
\log (k\ell))$-competitive  algorithm for the fractional
$k$-server problem on $T$, provided $\sigma = \Omega( \ell \log(k
\ell))$.
\end{thm}

\paragraph{Putting it all together.}
We now show how to use Theorems \ref{thm:frac_ov} and
\ref{thm:alloc_to_kserver} to prove our $k$-server guarantee for
general metrics, i.e., to prove Theorem \ref{thm:main}.

To this end, we need two more results that we prove in Section \ref{s:combining}. First,
\begin{thm}
\label{thm:rounding_ov}
Let $T$ be a $\sigma$-HST with $\sigma>5$. Then any online fractional $k$-server algorithm on $T$ can be converted into
a randomized $k$-server algorithm on $T$ with an $O(1)$ factor loss in the competitive ratio.
\end{thm}
Note that the above result gives a rounding procedure only for
HSTs (and not weighted HSTs). To relate HSTs to weighted HSTs, we
show the following.

\begin{thm}
\label{thm:transform_ov} Let $T$ be a $\sigma$-HST with $n$
leaves, but possibly arbitrary depth. Then $T$ can be transformed
into a weighted $\sigma$-HST $\widetilde{T}$ such that:
$\widetilde{T}$ has depth $O(\log n)$, the leaves of
$\widetilde{T}$ and $T$ are identical, and any leaf to leaf
distance in $T$ is distorted by a factor of at most $2
\sigma/(\sigma-1)$ in $\widetilde{T}$.
\end{thm}

Given the above results, we can present the proof of our main theorem.
\begin{Proof}[Proof of Theorem \ref{thm:main}]
Our algorithm proceeds as follows. First, we use the standard
technique \cite{FRT03} to embed the input (arbitrary) metric $M$
into a distribution $\mu$ over $\sigma$-HSTs with stretch $\sigma=
\Theta(\log n \log (k \log n))$. This incurs a distortion of
$O(\sigma \log_\sigma n)$ and the resulting HSTs have depth
$O(\log_\sigma \Delta)$, where $\Delta$ is the diameter of $M$.

Next, we pick a random HST $T$ according to the distribution
$\mu$, and transform $T$ into $\widetilde{T}$ using Theorem
\ref{thm:transform_ov}. As $\widetilde{T}$ has depth $\ell =
O(\log n)$, it holds that $\sigma =\Theta(\ell \log (k \ell))$ and
hence applying Theorem \ref{thm:alloc_to_kserver} to
$\widetilde{T}$ gives an $O(\ell \log (k\ell)) = O(\log n \log(k
\log n))$-competitive fractional $k$-server algorithm on
$\widetilde{T}$. Since the leaves of $T$ and $\widetilde{T}$ are
identical, and the distances only have $O(1)$ distortion, the
fractional $k$-server solution on $\widetilde{T}$ induces an
$O(\log n \log(k \log n))$-competitive fractional $k$-server
solution on $T$. By Theorem \ref{thm:rounding_ov}, this gives an
$O(\log n \log (k\log n))$-competitive randomized $k$-server
algorithm on $T$.

We now relate the optimum $k$-server cost on $M$ to the optimum on
$T$. Let $\opt_M^*$ denote the optimum $k$-server solution on $M$,
and let $c_T$ denote the cost of this solution on $T$. Since the
expected distortion of distances in our ensemble of HSTs is small,
we have:
\begin{equation}
\label{low:dist}
\Ex_{\mu}[c_T] = O( \sigma \log_\sigma n) \cdot \opt_M^*.
\end{equation}

Let $\Alg_T$ denote the cost of the solution produced by the
online algorithm on $T$, and let $\Alg_M$ denote the cost of this
solution on the metric $M$. As the pairwise distances in $T$ are
at least the distances in $M$,  $\Alg_M \leq \Alg_T $. Also, as
$\Alg_T$ is $O(\log n \log (k\log n))$-competitive, it follows
that:
$$\Alg_M \leq \Alg_T  = O(\log n \log (k\log n)) \cdot c_T^* \leq O(\log n \log (k\log n)) \cdot c_T $$
where $c_T^*$ is the optimum  $k$-server cost on $T$ (and hence
$c_T^* \leq c_T$). Taking expectation with respect to $\mu$ above
and using \eqref{low:dist}, the expected cost of our solution
$\Ex_{\mu}[\Alg_M]$ satisfies:
$$\Ex_{\mu}[\Alg_M] =
O(\log n \log (k\log n)) \cdot \Ex_{\mu}[ c_T]  =  O( \sigma
\log_\sigma n)  \cdot  O(\log n \log (k\log n))  \cdot \opt_M^*,$$
which implies that the  overall algorithm has a competitive ratio
of
$$ O\left(\sigma \left(\frac{\log n}{\log \sigma}\right)\right) \cdot
O\left( \log n \log (k \log n)\right) = O\left( \frac{\log^3 n (\log (k  \log n))^2}{\log \log n} \right) = O\left(\log^2 k \log^3 n \log \log n\right).$$
\end{Proof}

\section{The Fractional Allocation Problem}
\label{s:frac}

Consider a metric corresponding to a weighted star on $d$ leaves (also called {\em locations}) $1,\ldots,d$, where $w_i$ is the distance from leaf $i$ to the root.
Let us fix a sequence of cost vectors $h^0, h^1,\ldots$ and a server quota pattern $\kappa=(\kappa(1),\kappa(2),\ldots)$, where $\kappa(t)$ is the number of servers available at time $t$,
and $\kappa(t)\leq k$ for all times $t$.

Recall that in the fractional allocation problem the state at each time $t$ is described by non-negative variables $x^t_{i,j}$
denoting the  probability that there are {\em exactly} $j$ servers at location $i$.
At each time $t$, the variables $x^t_{i,j}$ satisfy: (1) $\sum_j x^t_{i,j}=1$, for each $i$; (2) $\sum_i \sum_j j x^t_{i,j} \leq \kappa(t)$.%

As we shall see, when describing and analyzing our algorithm for the fractional allocation problem, it will be easier to work with variables $y^t_{i,j}$, defined as
$$y^t_{i,j} = \sum_{j'=0}^{j-1} x^t_{i,j'}, \qquad \textrm{for } i \in \{1, \ldots, d\}, \quad j\in\{1,2, \ldots, k+1\}.$$
I.e.,  $y^t_{i,j}$ is the probability that at time $t$ we have {\em less} than $j$ servers at location $i$. Clearly, for every $i$,  as long as:
\begin{eqnarray}\label{eq:y_1_consistency}
y_{i,j}^t&\in &[0,1]\\\label{eq:y_2_consistency}
y_{i,j-1}^t&\leq& y_{i,j}^t, \qquad  y_{i,k+1}=1, \qquad \forall  i \in \{1, \ldots, d\}, \quad j\in\{2, \ldots, k+1\},
\end{eqnarray}
there is always a unique setting of the variables $x_{i,j}^t$s that corresponds to the $y_{i,j}^t$s. Therefore, in what follows we make sure that the variables $y_{i,j}^t$s generated by our algorithm satisfy the above two conditions.

The condition that at most $\kappa(t)$ servers are available at each time $t$ can be expressed in terms of $y_{i,j}^t$ as:
\begin{eqnarray}
\label{eq:volume_invariant}
 \sum_{i=1}^d \sum_{j=1}^k  y_{i,j}^t &= &\sum_{i=1}^d \sum_{j=0}^k (k-j) x_{i,j}^t
   =  k \sum_{i=1}^d \sum_{j=0}^k x_{i,j}^t  -  \sum_{i=1}^d \sum_{j=0}^k j x_{i,j}^t \nonumber \\
   & = & kd -  \left(\sum_{i=1}^d \sum_{j=0}^k j x_{i,j}^t\right)
    \geq  kd - \kappa(t).
 \end{eqnarray}
Let us now focus on a particular cost vector
$h^t=(h^t(0),h^t(1),\ldots, h^t(k))$ corresponding to time step
$t$. Recall that $h^t(j)$ is the hit cost incurred when serving
the request using {\em exactly} $j$ servers.
We can express $h^t$ as
$$\lambda^t_j = \left\{ \begin{tabular}{ll}$h^t(j-1)-h^t(j)$ & $j=1,2,\ldots, k$ \\
$h^t(k)$ & $j= k+1$ \end{tabular} \right.$$
The variables $\lambda^t_j$ are non-negative as the hit costs are non-increasing in $j$, i.e., $h^t(0)\geq h^t(1)\geq \ldots \geq h^t(k)$.
Intuitively, $\lambda_j^t$ captures the marginal cost of serving the request with strictly less than $j$ servers.\footnote{We note that we can assume that $\lambda_{k+1}^t$ is always $0$. Otherwise, as any valid algorithm (including the optimal one) always has at most $k$ servers at a given location, any competitive analysis established for the case  $\lambda_{k+1}^t=h^t(k)=0$ carries over to the general case. Thus, from now on we remove $\lambda_{k+1}^t$ and also $y_{i,k+1}$ (that is always $1$) from our considerations and notation.}
The hit cost incurred by a configuration $\oy^t=\{y_{i,j}^t\}_{i,j}$ now has a simple formulation. Let $\ooi$ denote the location on which the hit cost vector $h^t$ appears, then the hit cost $ \sum_{j=0}^{k-1} h^t(j) x^t_{\ooi,j}$
can be expressed as
$$  \sum_{j=1}^{k} \lambda^t_j \cdot \oy^t_{\ooi,j}.$$
Similarly, expression \eqref{eq:move_cost} for the movement cost
from a configuration $\oy^{t-1}$ to a configuration $\oy^{t}$ becomes
$$ \sum_{i=1}^d w_i \left( \sum_{j=1}^k  \left|y^{t}_{i,j}-y^{t-1}_{i,j}\right| \right). $$

\subsection{Description of the Algorithm}

In light of the above discussion it suffices to specify how state
$\{y_{i,j}^{t-1}\}_{i,j}$ evolves to $\{y_{i,j}^{t}\}_{i,j}$ at
time $t$ upon arrival of cost vector $h^t$ and server quota
$\kappa(t)$. Our algorithm performs this evolution in two stages.
First, it executes a {\em fix} stage in which the number of servers is
decreased so as to comply with a decrease of the quota
$\kappa(t)$. Then, it proceeds with a {\em hit} stage, during which the
(fractional) configuration of the servers is modified to react to
the cost vector $h^t$. We describe the dynamics of both stages as
a continuous process governed by a set of differential equations.
As it turns out, viewing the evolution of the server configuration
this way allows us to both simplify the description and the
analysis of the algorithm. The evolution of the fractional
solution during the fix stage is parametrized by a time index
$\tau$ that starts at 0 and grows until the number of servers is
no more than $\kappa(t)$. The hit stage is parametrized by a time
index $\eta$ that starts initially at 0 and ends at 1.

For the sake of simplicity, let us drop the index $t$ from our
notation since it does not play any role in our analysis. We
denote the configuration at time $t-1$ by $y^0$ and the
configuration at time $t$ by $y^1$. Let $\lambda$ denote the hit
cost vector $\lambda^t$ and let $\oi$ denote the location
$\ooi$ that $\lambda^t$ penalizes. The intermediate states of the
algorithm are denoted by $y^{\tau}$,  $\tau \geq 0$, during the
fix stage, and by $y^{\eta}$, $\eta \in [0,1]$, during the hit
stage. At each time $\eta \in [0,1]$ (respectively, $\tau \geq
0$), the algorithm specifies the derivative
$\frac{dy_{i,j}^\eta}{d \eta}$ of each variable $y_{i,j}^\eta$
(respectively, $\frac{dy_{i,j}^\tau}{d \tau}$ of each
$y_{i,j}^\tau$). Denote by $\tau_e$ the final value that $\tau$
reaches during the fix stage.
Eventually, each $y_{i,j}^t$ is defined as follows.

\begin{equation}
y_{i,j}^{t} = y_{i,j}^{t-1}+ \int_{\tau=0}^{\tau_e}\frac{dy_{i,j}^\tau}{d \tau}d\tau + \int_{\eta=0}^{1}\frac{dy_{i,j}^\eta}{d \eta}d\eta. \label{eq_update}
\end{equation}

An important issue that needs to be addressed is proving
that the differential equations specifying the derivatives at
each step have a (unique) solution and thus the algorithm is
well-defined. This proof turns out to be non-trivial in the case
of the hit stage, since the derivative during this stage might
change in a non-continuous manner. Nevertheless, as we will show,
the process is still well-defined.

Another technical detail is that during the hit stage, in
intermediate times $\eta \in [0,1]$, we will not work with the hit
cost vector $\lambda$ directly, but rather with a modified cost
vector $\lambda^\eta$ that can vary with $\eta$. (During the first
reading, the reader may assume that $\lambda^\eta = \lambda$ and
skip the part below about blocks and go directly to the description of the fractional algorithm.)

 We initialize $\lambda^0=\lambda$.
To define $\lambda^{\eta}$ for $\eta>0$, we need the notion of blocks.

\vspace{-0.15in}

\paragraph{Blocks:} During the hit stage, for each $\eta\in[0,1]$, we maintain a partition
of the index set $\{1,\ldots,k+1\}$ (for location $\bar{i}$) into
{\em blocks} $B^\eta_1,B^\eta_2\ldots, B^\eta_\ell$. The
collection of blocks is denoted by $\cB^\eta$ and it satisfies the
following properties.

\begin{enumerate}
\item $y^\eta_{\bar{i},j}$ is identical for all indices $j$ within any block $B \in \cB^{\eta}$.
For future reference, let us denote by $y^\eta(B)$ this common value for all $j \in B$.
\item
For any block $B =\{j,\ldots,j+s-1\}$ of length $s$ in $\cB^\eta$, it holds that for every $1 \leq r \leq s$,
\begin{equation}
\label{p:majorize}
\sum_{j'=j}^{j+r-1} \frac{1}{r} \lambda_{\oi,j'} \leq \sum_{j'=j}^{j+s-1}  \frac{1}{s} \lambda_{\oi,j'}.
\end{equation}
That is, the average value of the $\lambda$'s in any prefix of a
block is no more than the average of the entire block.
\end{enumerate}

We define $\lambda^\eta$ to be the cost vector obtained by averaging
$\lambda$ over the blocks in $\cB^{\eta}$. That is, for each  $B
\in \cB^{\eta}$, we set $\lambda(B) = (\sum_{j\in B} \lambda_{\oi,j})/|B|$, and then
\begin{eqnarray*}
\lambda^{\eta}_{i,j} & = &
\left\{\begin{array}{ll}
\lambda(B) & \mbox{if } i=\oi \mbox{ and } j\in B, \mbox{ for } B\in \cB^{\eta}, \\
0 & \mbox{otherwise.}\\
\end{array}\right.\label{eqn:def_block}
\end{eqnarray*}

Now, in our algorithm, the initial partitioning $\cB^0$ of blocks
is the trivial one, i.e., one in which each index $j$ forms its
own block. (Note that in this case we indeed have
$\lambda^0=\lambda$.) Next, blocks are updated as $\eta$
increases. For any $\eta\geq 0$, if two consecutive blocks $B_p,
B_{p+1} \in \cB^{\eta}$ satisfy:
\begin{equation}
\label{eq:mergeblocks}
y^{\eta}(B_p) = y^{\eta}(B_{p+1}) \qquad \textrm{ and } \qquad   \lambda(B_p) < \lambda(B_{p+1}),
\end{equation}
then $B_{p}$ and $B_{p+1}$ are merged and $\cB^{\eta}$ is modified
accordingly. Note that the condition  $\lambda(B_p) \leq
\lambda(B_{p+1})$ guarantees that \eqref{p:majorize} is satisfied
in the new block created by merging $B_p$ and $B_{p+1}$. As we
shall see later (Lemma \ref{obs:nosplit}), a crucial property of
the evolution of $\cB^\eta$ during the hit stage is that
$y_{\oi,j}^{\eta}$s are updated in a way that guarantees that a
block never splits once it is formed.

\paragraph{The algorithm.} We are now ready to state our algorithm. It is parameterized by a parameter $\eps>0$ that will be fixed later.

\vspace{0.35cm}
\noindent \framebox[1.05\width][l]{
\begin{minipage}{0.94\linewidth}
{\bf Fractional Allocation Algorithm:} \\
Set $\beta=\frac{\eps}{1+k}, \alpha=\ln(1+\frac{1}{\beta})= \ln(1+\frac{1+k}{\eps})$.\\
{\bf Fix stage:} Set $y^0=y^{t-1}$.\\
For any $\tau \in [0,\infty)$, while $\sum_{i,j}y^\tau_{i,j} < kd- \kappa(t)$ (i.e., while the total volume of servers exceeds the quota) we increase each variable $y_{i,j}^{\tau}$ at a rate:
\begin{eqnarray*}
\frac{d y_{i,j}^\tau}{d \tau} & = & \left\{\begin{array}{ll}\frac{1}{w_i}\left(y^\tau_{i,j}+\beta\right) & y^{\tau}_{i,j} <1\\
0 & y_{i,j}^{\tau}=1\end{array} \right.
\end{eqnarray*}
Denote by $\tau_e$ the termination time of the fix stage.

{\bf Hit stage:} Set $y^0$ to be the state obtained at the end of
the fix stage\footnote{Note that upon termination of the fix
stage, $\sum_{i,j}y^{\tau_e}_{i,j} \geq kd- \kappa(t)$.}. Define
the following update rule for any $\eta \in [0,1]$:
\begin{itemize}
\item If $\sum_{i,j}y^\eta_{i,j} = kd- \kappa(t)$, choose $N(\eta)\geq 0$ such that\footnote{As we show in Lemma \ref{lem:wd}, there is always a way of choosing $N(\eta)$ such that the desired conditions are satisfied.}:
\begin{eqnarray}
\frac{dy^{\eta}_{i,j}}{d \eta} & = &
\left\{\begin{array}{ll}
0 & \mbox{if either }\left(N(\eta)-\alpha \lambda^\eta_{i,j}\right)> 0 \mbox{ and } y^{\eta}_{i,j}=1, \\
& \mbox{or } \left(N(\eta)-\alpha \lambda^\eta_{i,j}\right)\leq 0 \mbox{ and } y^{\eta}_{i,j}=0\\
\frac{1}{w_i}\left(y^{\eta}_{i,j}+\beta\right)\cdot\left(N(\eta)-\alpha \lambda^\eta_{i,j}\right) & \mbox{otherwise}\\
\end{array}\right.\label{main-eqn}
\end{eqnarray}
and
$$\sum_{i,j}\frac{dy_{i,j}^{\eta}}{d \eta}=0.$$
\item Otherwise (i.e., if $\sum_{i,j}y^\eta_{i,j} > kd- \kappa(t)$), set $N(\eta)=0$, and define the derivatives of the variables as above.\\
\end{itemize}
{\bf Output:} For each $(i,j)$, return $y_{i,j}^{t} \triangleq
y_{i,j}^{t-1} + \int_{\tau=0}^{\tau_e}\frac{dy_{i,j}^\tau}{d
\tau}d\tau + \int_{\eta=0}^{1}\frac{dy^{\eta}_{i,j}}{d \eta} d
\eta$.
\end{minipage}}
\vspace{0.50cm}

\paragraph{High-level intuition.} Before proving correctness and analyzing the performance
of the above algorithm, we provide some intuition on the dynamics
underlying it.

{\em Dynamics of the fix stage:}
This is fairly straightforward. The algorithm simply increases all the variables
$y^{\tau}_{i,j}$ that are strictly less than 1 (which decreases
the total number of servers), until the quota $\kappa(t)$ on the
number of servers is met. We note that it may also be the case that
the total number of servers used is already strictly smaller than the quota $\kappa(t)$ to begin with,
e.g.,  if the server quota increases at time $t$. In this case, nothing is done during the fix stage. Notice that the
rate of change of a variable $y^{\tau}_{i,j}$ is proportional to
its value, which means that the change is governed by an
exponential function. This kind of update rule is in line with
previous algorithms for weighted paging \cite{BBN07,BBN10a}.

{\em Dynamics of the hit stage:} For simplicity, let us assume
that during this stage we have that $0<y_{i,j}^\eta<1$, for all $(i,j)$, and each $y_{i,j}^\eta$ is a
strictly increasing function of $j$. That is,
\begin{equation}\label{eq:simp_as}
0<y_{i,1}^\eta < y_{i,2}^{\eta} < \ldots < y_{i,k}^{\eta} <1,
\end{equation} for all locations $i$ and $\eta\in [0,1]$.

Note that under this assumption condition \eqref{eq:mergeblocks}
will never trigger. As a result, no blocks are merged and we have
$\lambda^{\eta}=\lambda$ for all $\eta\in[0,1]$. Furthermore, as
in this case each variable $y_{i,j}^\eta$ is strictly between 0
and 1, its rate of change during the hit stage simplifies to:
\begin{equation}\label{eq:simple_rate} \frac{dy_{i,j}^{\eta}}{d
\eta}
=\frac{1}{w_i}\left(y^{\eta}_{i,j}+\beta\right)\cdot\left(N(\eta)-\alpha
\lambda_{i,j}\right),
\end{equation}
with
\begin{eqnarray}
N(\eta)  &=& \left\{\begin{array}{l} 0 \qquad \mbox{if
}\sum_{i,j}y^\eta_{i,j} > kd- \kappa(t), \\
\\
\frac{\sum_{i,j}\frac{1}{w_i}\left(y^{\eta}_{i,j}+\beta\right)\cdot
\alpha
\lambda_{i,j}}{\sum_{i,j}\frac{1}{w_i}\left(y^{\eta}_{i,j}+\beta\right)}
\qquad \mbox{otherwise (i.e. if } \sum_{i,j}y^\eta_{i,j} = kd- \kappa(t)).\\
\end{array}\right.\label{eq:n_norm}
\end{eqnarray}
(Note that the value of $N(\eta)$ in the second case of
\eqref{eq:n_norm} is determined by the fact that
$\sum_{i,j}\frac{dy_{i,j}^{\eta}}{d \eta}$ has to be zero.)

Let us study the dynamics given by
\eqref{eq:simple_rate} more carefully. Recall that $\lambda_{i,j}=0$ for $i\neq
\oi$.
First, if the number of servers used is below the quota, i.e.~if $\sum_{i,j}y^\eta_{i,j} > kd- \kappa(t)$,
the algorithm responds to the cost vector $\lambda$ by simply
increasing the number of servers at location $\bar{i}$ by decreasing each
$y_{\bar{i},j}^\eta$ at a rate of $\frac{1}{w_{\bar{i}}}
(y^{\eta}_{\bar{i},j}+\beta)\cdot \alpha
\lambda^\eta_{\bar{i},j}$. To understand this better, it is instructive to consider the special case when
$\lambda^{\eta}_{\oi,\oj}=1$ for some particular index $\oj$ and is $0$ otherwise (this corresponds to the hit cost vector $h^\eta$ that incurs
cost $1$ if there are strictly fewer than $\oj$ servers at $\oi$ and cost $0$ otherwise).
In this case, the algorithm reduces $y^{\eta}_{\oi,\oj}$ and keeps other $y^{\eta}_{i,j}$'s unchanged (in particular $y^{\eta}_{\oi,\oj+1}$ and $y^{\eta}_{\oi,\oj-1}$ remain unchanged).
As $y^{\eta}_{\oi,\oj+1}$ and $y^{\eta}_{\oi,\oj-1}$ do not change while $y_{\oi,\oj}$ decreases, this has the effect of increasing the probability mass $x^\eta_{\oi,\oj}=y^\eta_{\oi,\oj+1}-y^\eta_{\oi,\oj}$ on $(\oi,\oj)$, and decreasing the probability mass $x^\eta_{\oi,\oj-1}=y^\eta_{\oi,\oj}-y^\eta_{\oi,\oj-1}$ on $(\oi,\oj-1)$\footnote{In the simplified discussion here we are implicitly assuming that $x^\eta_{\oi,\oj-1}>0$ by assuming that $y^\eta_{\oi,\oj}> y^\eta_{\oi,\oj-1}$.} Moreover, note that the decrease in $x^\eta_{\oi,\oj-1}$ is exactly equal to the increase in $x^\eta_{\oi,\oj}$.

Now, let us consider the case when the number of servers used is exactly equal to the quota. Here, we also need to ensure that the quota is maintained.
This is done by offsetting the increase in the number of servers at location $\oi$ (as described by the dynamics in the previous paragraph),
by decreasing the number of servers at all locations (including $\oi$). This is precisely the purpose of the term $N(\eta)$ in \eqref{eq:simple_rate}.
It increases $y^{\eta}_{i,j}$ (and hence decreases the number of servers) at a
rate proportional to $\frac{1}{w_i}(y^\eta_{i,j}+\beta)$ (as in the fix stage). Note that as $\lambda_{i,j}=0$ for
$i\neq \bar{i}$, this update can only decrease the
number of servers at locations $i \neq \oi$. The overall number of servers at location $\oi$ can only increase, but of course due to the redistribution of probability
mass at $\oi$, it may happen
that the probability mass at some $(\oi,j)$ goes down.

Unfortunately, when assumption \eqref{eq:simp_as} does not hold,
the simple dynamics described above may produce infeasible configurations.
First, increasing or decreasing variables according to
\eqref{eq:simple_rate} does not take into account that the variables need to stay in the range $[0,1]$,
and hence this may be violated.
This happens if (i) a variable is equal to 0 and has a negative
derivative, or (ii) when it is equal to 1 and has a positive derivative. To
avoid this problem we need to deactivate such variables (by
setting their derivative to be 0) when either one of these two
cases occurs. Moreover, the above dynamics may also violate the
monotonicity condition \eqref{eq:y_2_consistency}. To avoid
this issue, we need to merge blocks and modify
$\lambda^\eta$ accordingly, as was previously described.

Now, the resulting algorithm does not produce infeasible
configurations anymore. However, its dynamics is somewhat more
involved. Before we discuss it, let us first provide a formal definition of an {\em inactive} coordinate or
variable.
\begin{defn}\label{def-active}
During the
fix stage, a coordinate $(i,j)$ for which $y^{\tau}_{i,j}<1$ is said
to be {\em active} at time $\tau$. Otherwise it is said to be {\em inactive}.
During the hit stage, coordinate $(i,j)$ (or variable $y^\eta_{i,j}$) is said
to be {\em inactive} at time $\eta \in [0,1]$ if either
$\left(N(\eta)-\alpha \lambda^\eta_{i,j}\right)> 0 \mbox{ and }
y^{\eta}_{i,j}=1$, or $\left(N(\eta)-\alpha
\lambda^\eta_{i,j}\right)\leq 0 \mbox{ and } y^{\eta}_{i,j}=0$.
Otherwise, coordinate $(i,j)$ is said to be {\em active}.
Denote by $A^{\eta}$ (respectively
$A^{\tau}$) the set of active coordinates at time $\eta$
(respectively $\tau$).
\end{defn}

Now,  by definition, during the hit stage at time $\eta$ only the active variables might change. So, we can compactly rewrite the evolution of the variables during the hit stage given in \eqref{main-eqn} as
\begin{eqnarray}
\frac{dy^{\eta}_{i,j}}{d \eta} & = &
\left\{\begin{array}{ll}
0 & \mbox{if } (i,j)\notin A^{\eta}, \\
\frac{1}{w_i}\left(y^{\eta}_{i,j}+\beta\right)\cdot\left(N(\eta)-\alpha \lambda^\eta_{i,j}\right) & \mbox{otherwise}\\
\end{array}\right.\label{eq:y_evol_general}
\end{eqnarray}
Furthermore, as we still require that $\sum_{i,j}\frac{dy_{i,j}^{\eta}}{d \eta}=\sum_{(i,j)\in A^{\eta}}\frac{dy_{i,j}^{\eta}}{d \eta}=0$, if $\sum_{i,j}y^\eta_{i,j} > kd- \kappa(t)$, we have that $N(\eta)$ can be expressed as
\begin{eqnarray}
N(\eta)  &=& \left\{\begin{array}{l} 0 \qquad \mbox{if
}\sum_{i,j}y^\eta_{i,j} > kd- \kappa(t), \\
\frac{\sum_{(i,j)\in A^{\eta}}\frac{1}{w_i}\left(y^{\eta}_{i,j}+\beta\right)\cdot
\alpha
\lambda_{i,j}^{\eta}}{\sum_{(i,j)\in A^{\eta}}\frac{1}{w_i}\left(y^{\eta}_{i,j}+\beta\right)}
\qquad \mbox{otherwise (i.e. if } \sum_{i,j}y^\eta_{i,j} = kd- \kappa(t)).\\
\end{array}\right.\label{eq:n_general}
\end{eqnarray}

In light of the above, one can see that the simple evolution of the variables in the special case of \eqref{eq:simp_as}, as described by \eqref{eq:simple_rate}, is a special case of the general evolution in which all coordinates are being active and $\lambda^\eta=\lambda$ for all $\eta\in[0,1]$. The reason why the analysis of the general process is more complicated is that the set of active coordinates (and the hit cost $\lambda^{\eta}$) can, in principle, change very abruptly between two values of $\eta$. Moreover, as stated, equation \eqref{eq:n_general} and Definition \ref{def-active}, have a circular dependency. In particular, the value
of $N(\eta)$ depends on the set $A^\eta$, but in turn the definition of $A^{\eta}$ also depends on the value of $N(\eta)$.
As a result, a priori it is not even clear that our algorithm is well defined. That is, a unique trajectory consistent with our local evolutionary rules indeed exists. We proceed to proving this now.

\paragraph{Well-definiteness of the algorithm.} We start by addressing the above-mentioned issue of the circular dependency between the value of $N(\eta)$ and the set $A^{\eta}$. Note that it is not clear any more that there always exists a non-negative normalization factor $N(\eta)$ as required by our
algorithm. As we prove in the next lemma, however, one can use a simple continuity argument to prove the existence of the desired normalization factor.
\begin{lem}\label{lem:wd}
There
exists a $N(\eta)\geq 0$ for which
$\sum_{i,j}\frac{dy_{i,j}^{\eta}}{d \eta}= 0$, where
the derivatives $\frac{dy_{i,j}^{\eta}}{d \eta}$ are as defined in the algorithm.
Moreover, the set $A^\eta$ of active coordinates is never empty.
\end{lem}
\begin{proof}
Fix any  $\eta \in [0,1]$. Let us consider the function
$f(s) = \sum_{i,j}({dy_{i,j}^{\eta}}/{d \eta})|_{N(\eta)=s}$, i.e., $f(s)$ is the sum of all derivatives given by equation \eqref{main-eqn} for the case when $N(\eta)$ is equal to $s$.

Clearly, if $s=N(\eta)=0$, then $({dy_{i,j}^{\eta}}/{d \eta})|_{N(\eta)=s}\leq0$ for each $(i,j)$ and hence $f(s)\leq0$.
If $f(0)=0$ then $N(\eta)$ satisfies the requirements. Note that in this case all coordinates that are non-zero are active. This set is non-empty as at the beginning of the hit stage the sum over all coordinates is at least $kd-\kappa(t)>0$.

Thus, suppose that $f(0)<0$. Let $\lambda_{\max} = \max_{i,j} \lambda^\eta_{i,j} $ be the largest entry in $\lambda^{\eta}$. Then, at $s=N(\eta)=\alpha \lambda_{\max}$ we have $({dy_{i,j}^{\eta}}/{d \eta})|_{N(\eta)=s} \geq 0$ for each $(i,j)$ and hence $f(s)\geq 0$.

Next, we claim that each derivative $({dy_{i,j}^{\eta}}/{d \eta})|_{N(\eta)=s}$ is a continuous function of $s$.
To this end, note that if $0 < y_{i,j}^{\eta} <1$,
then the function $({dy_{i,j}^{\eta}}/{d \eta})|_{N(\eta)=s}$ is a linear function of $s$ (and thus is continuous).
For $y_{i,j}^{\eta}=0$, the
function $({dy_{i,j}^{\eta}}/{d \eta})|_{N(\eta)=s}$ is zero for $s \leq
\alpha \lambda^{\eta}_{i,j}$, and then increases linearly for $s\geq \alpha
\lambda^{\eta}_{i,j}$ -- so, again, it is continuous. Similarly, for $y^{\eta}_{i,j}=1$,
$(\frac{dy_{i,j}^{\eta}}{d \eta})|_{N(\eta)=s}$ is negative if $s \leq
\alpha \lambda^{\eta}_{i,j}$, and increases linearly until
$s=\alpha \cdot \lambda^{\eta}_{i,j}$, and then remains zero.

Now, as each derivative is continuous, so is $f$. Thus, we know that by the
intermediate value theorem the preimage $f^{-1}(0)$ in the interval $[0, \alpha \lambda_{\max}]$ is non-empty. Furthermore, as $f$ is continuous and $f(s)<0$ for $s=0$, there exists a minimal $s^*$ such that $f(s^*)=0$ and $0<s^* \leq \alpha \lambda_{\max}$.

We take $N(\eta)=s^*$ and claim that the corresponding set $A^{\eta}$ is non-empty (which would prove the lemma). To see why it is the case, note that if there exists a coordinate $0 < y_{i,j}^{\eta} < 1$ then it is always active and we are done. Otherwise, let us consider $A_0$ ($A_1$) to be the set of $(i,j)$ with $y_{i,j}^{\eta}=0$ (resp. with $y_{i,j}^{\eta}=1$). As the sum over all coordinates is at least $kd-\kappa(t)>0$, the set $A_1$ is non-empty.
Suppose for $s^*>0$ all the coordinates in $A_0, A_1$ are inactive, then by definition for all $(i,j)\in A_1$, $s^*>\alpha \lambda_{i,j}^{\eta}$ (note the strict inequality), and for all coordinates in $A_0$, $s^* \leq \alpha \lambda_{i,j}^{\eta}$. This is a contradiction to the minimality of $s^*$ as we could find a $0< s'<s^*$, such that $f(s')=0$.
\end{proof}

Now, as our algorithm is defined via a set of differential equations
indexed by $\tau$ and $\eta$, to prove that it is well-defined we need to show that there exists a
unique solution to this set, and furthermore  this solution is
feasible for the allocation problem.  To this end, we prove the following
lemma, whose proof is in Appendix \ref{app:nosplit}.

\begin{lemma}\label{obs:nosplit}
There exists a unique solution $y^{\tau}$ and $y^{\eta}$, defined
on the intervals $\tau \geq 0$, $\eta \in [0,1]$, to the set of
differential equations defined by the algorithm. Furthermore, the
solution satisfies the following properties:
\begin{itemize}
\item {\bf Boundary:} For each  $(i,j)$, and for all $0 \leq
\tau$, and $0 \leq \eta \leq 1$: $0 \leq
y^{\eta}_{i,j},y^{\tau}_{i,j} \leq 1$. \item {\bf Monotonicity:}
For each $(i,j)$, ($j\leq k$), $y^{\eta}_{i,j} \leq
y^{\eta}_{i,j+1}$ and $y^{\tau}_{i,j} \leq y^{\tau}_{i,j+1}$.
\item {\bf Quota:} The expected number (volume) of servers
at the end of the fix stage and at any $\eta \in [0,1]$ does not
exceed $\kappa(t)$. That is, $\sum_{i,j}y^{\eta}_{i,j} \geq kd-
\kappa(t)$ for all $\eta\in[0,1]$.
\item {\bf Blocks:} During the
hit stage, Blocks can only merge (and they never split). \item
{\bf Discontinuity:} The total number of times $\eta \in [0,1]$
that each location $(i,j)$ changes its status from active to
inactive, as well as the number of discontinuity points of
$N(\eta)$ as a function of $\eta$, is finite (in fact, polynomial in $k$ and $d$).
\end{itemize}
\end{lemma}

\subsection{Cost Accounting}

In this section we prove some helpful properties that allow us to
charge the algorithm and the optimal solution in a continuous
fashion. This will simplify the potential function based analysis
that we later perform. First, we deal with the charging of the hit
cost, and then with the accounting of the movement cost.

\paragraph{Charging the hit cost.}
The issue we want to address here is that at a given time $t$ the hit costs of the optimal solution and our algorithm  depend only on the final states of both solutions. More precisely, if $y^{*}$ is the optimal solution at time $t$, and $y=y^{1}$ is the final state of the algorithm at time $t$, then the hit cost of the optimal solution (respectively, of the algorithm) at time $t$ is equal to $\lambda \cdot y^*$ (respectively, $\lambda\cdot y$). However, as our algorithm is described in a continuous fashion, it would be simpler to also have a way of accounting for the hit costs in a continuous and local fashion.

In particular, we would like to account for the hit cost of the optimal solution as:
\begin{equation}
\label{eq:opt_hit_cost_account}
\int_{\eta=0}^{1} \lambda^\eta \cdot y^* \cdot d\eta,
\end{equation}
and for the hit cost of the algorithm as:
\begin{equation}\label{eq:hit_cost_account}
 \int_{\eta=0}^1 \lambda^\eta \cdot y^{\eta} d\eta.
\end{equation}
Note that the above expressions can be interpreted as charging locally at every time $\eta\in [0,1]$ an infinitesimally small hit cost of $\lambda^\eta \cdot y^* d\eta$ (respectively, $\lambda^\eta \cdot y^\eta d\eta$) to the optimal solution (respectively, to the algorithm). Now, to make this accounting valid, we need to show that the above expressions can only overestimate the  hit cost of our algorithm and underestimate the  hit cost of the optimal solution. We prove that this is indeed the case in the following lemma.

\begin{lem}\label{lem:charge1}
The following inequalities hold:
\begin{eqnarray}\int_{\eta=0}^{1} \lambda^{\eta} \cdot y^* \cdot d\eta & \leq & \lambda \cdot y^*, \label{charge-opt}
\end{eqnarray}

\begin{eqnarray}
\int_{\eta=0}^1 \lambda^\eta \cdot y^{\eta} d\eta & \geq & \lambda \cdot y.\label{charge-on}
\end{eqnarray}
\end{lem}
\begin{proof}
We first prove inequality (\ref{charge-opt}).
To this end, we show that for any non-decreasing vector
$v=(v_1,v_2,\ldots)$, and any $0 \leq \eta_1 < \eta_2 \leq 1$, it
holds that

\begin{equation}\label{ineq:montone_hit_cost}
\lambda^{\eta_1} \cdot v \geq \lambda^{\eta_2} \cdot v.
\end{equation}

Note that as $\lambda=\lambda^0$ and $y^*$ is feasible (and thus
satisfies property \eqref{eq:y_2_consistency}), taking $v$ equal
to $y^*$ immediately gives Inequality (\ref{charge-opt}).

By Lemma \ref{obs:nosplit}, we know that the only difference between $\lambda^{\eta_1}$ and $\lambda^{\eta_2}$ is that some of the blocks in $\cB^{\eta_1}$ can be merged in $\cB^{\eta_2}$.
Therefore, it suffices to show that whenever two consecutive blocks $B_1$ and $B_2$ merge to form another block $B$, it holds
 that $\lambda(B_1) \sum_{i \in B_1} v_i + \lambda(B_2) \sum_{i \in B_2} v_i \geq \lambda(B) \sum_{i\in B} v_i$ for any hit cost vector $\lambda$.

Let $\ell_1 = |B_1|, \ell_2 = |B_2|$, and let $a_1 = (\sum_{i \in
B_1} v_i)/\ell_1$, $a_2 = (\sum_{i \in B_2} v_i)/\ell_2$. Then, by
the definition of $\lambda(B)$, the inequality above is equivalent
to showing that
\begin{equation}
\label{eq:phew}
 \lambda(B_1) \ell_1 a_1 +   \lambda(B_2) \ell_2 a_2  \geq  \left(\frac{\lambda(B_1) \ell_1 + \lambda(B_2) \ell_2}{\ell_1 + \ell_2}\right) (a_1 \ell_1 + a_2 \ell_2).
 \end{equation}

As $v$ is increasing we have $a_1 \leq a_2$,  and since $B_1$ and $B_2$ were merged, by \eqref{eq:mergeblocks} it must be that $\lambda(B_1) < \lambda(B_2)$.
A direct calculation shows that \eqref{eq:phew} holds under these conditions.

 Now, to prove that inequality (\ref{charge-on}) also holds, we prove that whenever the derivative of $\lambda^{\eta} \cdot y^{\eta}$ is defined, i.e., whenever neither $A^\eta$ nor $\lambda^\eta$ change (which is the case except for possibly finitely many points, cf. Lemma \ref{obs:nosplit}) we have that

\begin{equation}\label{eq:decreas_loc_hit_cost}
\frac{d\left(\lambda^\eta\cdot y^\eta\right)}{d \eta} \leq 0.
\end{equation}
That is, the state $y^{\eta}$ evolves in a way that reduces the hit cost of the algorithm with respect to the corresponding hit cost vector.

To see how \eqref{charge-on} follows from \eqref{eq:decreas_loc_hit_cost} we first note that \eqref{eq:decreas_loc_hit_cost} implies that
$$
\lambda^\eta \cdot y^\eta \geq  \lambda^1 \cdot y^1,
$$
for any $\eta\in[0,1]$. Now, we have for any block $B\in \cB^{1}$,
$y^1_{\oi,j}=y^1(B)$ for all $j\in B$, and thus
$$
\lambda^0 \cdot y^1 = \sum_{B\in \cB^{1}} \sum_{j\in B} \lambda_{\oi,j}^0 y^{1}_{\oi,j} = \sum_{B\in \cB^{1}} y^1(B) \sum_{j\in B} \lambda_{\oi,j}^0 = \sum_{B\in \cB^{1}} y^{1}(B)\lambda^1(B) |B| = \lambda^1 \cdot y^1,
$$
where we recall that $\lambda^1(B) = (\sum_{j\in B} \lambda^\eta_{\oi,j})/|B| = (\sum_{j\in B} \lambda^0_{\oi,j})/|B|$.

So, we can conclude that
$$
\int_{\eta=0}^1 \lambda^\eta \cdot y^{\eta} d\eta \geq  \int_{\eta=0}^1 \lambda^1 \cdot y^1 d\eta =  \int_{\eta=0}^1 \lambda^0 \cdot y^1 d\eta= \lambda^0 \cdot y^1 = \lambda \cdot y.
$$
In light of the above, it remains to prove
\eqref{eq:decreas_loc_hit_cost}. To this end, recall that when
$A^\eta$ and $\lambda^\eta$ are fixed, the evolution of
$y^{\eta}_{i,j}$s is described by Equation
\eqref{eq:y_evol_general}, thus the statement we need to prove is
\begin{eqnarray}\label{eq:mid_hit_cost}
\sum_{(i,j) \in A^\eta}\frac{\lambda^\eta_{i,j}}{w_i} \left(y^\eta_{i,j}+\beta\right)\cdot \left(N(\eta) - \alpha  \lambda^\eta_{i,j}\right)  &\leq & 0.
\end{eqnarray}
Plugging in the expression for $N(\eta)$ given by
\eqref{eq:n_general} and canceling $\alpha$, we need to show that
\begin{eqnarray*}
 \left(\sum_{(i,j)\in A^\eta}\frac{\lambda^\eta_{i,j}}{w_i} \left(y^\eta_{i,j}+\beta\right)\right)^2 & \leq & \left(\sum_{(i,j) \in A^\eta}\frac{
(\lambda^\eta_{i,j})^2}{w_i}
\left(y^\eta_{i,j}+\beta\right)\right) \cdot \left(\sum_{(i,j) \in
A^\eta}\frac{1}{w_i}\left(y_{i,j}^{\eta}+\beta\right)\right).
\end{eqnarray*}
Now, this inequality follows from the Cauchy-Schwarz inequality
$(a \cdot b)^2 \leq |a|_2^2 |b|_2^2$, by taking $a$ to be the
vector with entries $a_{i,j}
=\sqrt{\frac{(\lambda^\eta_{i,j})^2}{w_i}
\left(y^\eta_{i,j}+\beta\right)}$, and $b$ to be the vector with
entries $b_{i,j}
=\sqrt{\frac{1}{w_i}\left(y_{i,j}^{\eta}+\beta\right)}$.
\end{proof}

\paragraph{Charging the movement cost.}

We turn our attention to the accounting of the movement cost of
our algorithm. Recall that the movement cost at time $t$ is
defined as

\begin{equation}\label{eq:expr_move_cost_org}
\sum_{i=1}^d w_i \left( \sum_{j=1}^k  |y^{t}_{i,j}-y^{t-1}_{i,j}| \right).
\end{equation}
We would like to approximate this expression by one which is
simpler and more convenient to work with. First, instead of
keeping track of both increases and decreases of the variables
$y^t_{i,j}$ as in the above, we will account for the movement cost
only through the increases of  the variables $y^t_{i,j}$. That is,
our bound for the movement cost is

\begin{equation}\label{eq:expr_move_cost_new}
\sum_{i=1}^d w_i \left( \sum_{j=1}^k  \max\{y^{t}_{i,j}-y^{t-1}_{i,j},0\} \right).
\end{equation}
Note that we have for any coordinate $(i,j)$ and $t\geq 1$

$$|y^{t}_{i,j}-y^{t-1}_{i,j}| \leq 2 \cdot \max\{y^{t}_{i,j}-y^{t-1}_{i,j},0\} + y^{t-1}_{i,j}-y^{t}_{i,j},$$
and thus on any input sequence consisting of $T$ requests, it is
the case that

$$
\sum_{t=1}^T \sum_{i=1}^d w_i \left( \sum_{j=1}^k  |y^{t}_{i,j}-y^{t-1}_{i,j}| \right) \leq 2\cdot \sum_{t=1}^T \sum_{i=1}^d w_i \left( \sum_{j=1}^k  \max\{y^{t}_{i,j}-y^{t-1}_{i,j},0\} \right) + \sum_{i=1}^d w_i \left( \sum_{j=1}^k  y^{0}_{i,j}-y^{T}_{i,j} \right).
$$
Thus, accounting for the movement cost via expression
\eqref{eq:expr_move_cost_new} approximates the true movement cost
(corresponding to expression \eqref{eq:expr_move_cost_org}) up to
a multiplicative factor of two and an additive factor of at most
$\sum_{i=1}^d kw_i$, which depends only on the starting and final
configuration, and is zero if the two configurations coincide. As
we will see, for the sake of our competitive analysis, such an
approximation suffices. %

Next, similarly to the way we accounted for the hit cost described
above, we wish to further simplify the charging of the movement
cost and perform it in a continuous and local fashion. Namely, it
is easy to see that the following quantity
$$
 \sum_{i=1}^d w_i \left( \sum_{j=1}^k  \int_{\tau=0}^1 \frac{d y_{i,j}^\tau}{d\tau} d\tau + \int_{\eta=0}^1 \max\{\frac{d y_{i,j}^\eta}{d\eta},0\} d\eta \right)
 $$
 can only overestimate the movement cost given by \eqref{eq:expr_move_cost_new}.
 (Note that $\frac{d y_{i,j}^\tau}{d\tau}$ is always non-negative.) Furthermore, as the derivatives $\frac{d y_{i,j}^\eta}{d\eta}$
 can only be positive if $(i,j)\in A^\eta$, we can write

 $$
\int_{\eta=0}^1 \max\left\{\frac{d y_{i,j}^\eta}{d\eta},0\right\}
d\eta \leq \int_{\eta=0}^1
\frac{1}{w_i}\left(y^{\eta}_{i,j}+\beta\right) N(\eta) \cdot
\mathbf{1}_{(i,j) \in A^{\eta}} d\eta,
 $$
 where we used \eqref{eq:y_evol_general} and the fact that, trivially, $N(\eta)-\lambda_{i,j}^\eta \leq N(\eta)$.
 Thus, we can use the above in our final version of the
estimate of the movement cost. The following claim summarizes the
discussion.
\begin{claim}
\label{clm:mvcost}
For any sequence of requests of length $T\geq 1$,
\begin{eqnarray*}
&& \sum_{t=1}^T \sum_{i=1}^d w_i \left( \sum_{j=1}^k
|y^{t}_{i,j}-y^{t-1}_{i,j}| \right)  \leq \\ && 2 \cdot
\sum_{t=1}^T \sum_{i=1}^d w_i \left( \sum_{j=1}^k  \int_{\tau=0}^1
\frac{d y_{i,j}^\tau}{d\tau} d\tau + \int_{\eta=0}^1
\frac{1}{w_i}\left(y^{\eta}_{i,j}+\beta\right) N(\eta) \cdot
\mathbf{1}_{(i,j) \in A^{\eta}} d\eta \right) +  C',
\end{eqnarray*}
where $C' \leq \sum_i k w_i$ depends only on the starting and
final configuration of the algorithm, and $C'=0$ if the two
configurations coincide.
\end{claim}
Thus, at time $t$, the movement cost of our algorithm is given by:
\begin{equation}\label{eq:movecost_final}
\sum_{i=1}^d w_i \left( \sum_{j=1}^k  \int_{\tau=0}^1 \frac{d y_{i,j}^\tau}{d\tau} d\tau + \int_{\eta=0}^1 \frac{1}{w_i}\left(y^{\eta}_{i,j}+\beta\right) N(\eta) \cdot
\mathbf{1}_{(i,j) \in A^{\eta}} d\eta \right).
\end{equation}

\subsection{Competitive Analysis}

We are finally ready to bound the competitiveness of our algorithm. To this end, we prove the following theorem.

\begin{thm}\label{thm:frac_alloc_variable}
Consider an arbitrary instance of the allocation problem with cost
vectors $h^1, h^2,\ldots$, a starting configuration $\oy^0$ and a
quota pattern $\kappa=(\kappa(1),\kappa(2),\ldots)$. For any
$0\leq \eps\leq 1$, we have the following bounds:
\begin{eqnarray*}
  \HitAlg & \leq &  (1+\eps)\left(\opt+ w_{\max} \cdot g(k)\right) + C,\\
 \MovAlg  & \leq & O(\log (k/\eps))\cdot \left(\opt + w_{\max} \cdot g(\kappa)\right).
 \end{eqnarray*}
Here, $\HitAlg$ and $\MovAlg$ denote the hit and movement costs of
our fractional algorithm, and $\opt$ denotes the sum of the total
hit and movement costs of a fixed {\em integral} optimum solution
to the allocation problem instance. Let $g(\kappa):=\sum_{t}
|\kappa(t)-\kappa(t-1)|$, and denote by $w_{\max}=\max_i w_i$  the
diameter of our metric space. Let $C$ be a quantity that depends
only on the start and final configurations of the online
algorithm, and $C=0$ if the two configurations coincide.
\end{thm}

It is easy to see that Theorem \ref{thm:frac_ov} immediately follows from the above theorem.

To prove Theorem \ref{thm:frac_alloc_variable} we employ a
potential function approach. Namely, we define potentials
$\Phi^h(\oy,t)$ and $\Phi^m(\oy,t)$ that depend on the state $\oy$
of the online algorithm and on the state of some arbitrary fixed
integral optimum solution at time $t$. Then we show that the
following inequalities are satisfied at each time step $t$.

\begin{eqnarray}\label{eq:conservation}
\MovAlg_t + \DelPhi^m_t & \leq & (1+ \eps) \cdot \alpha
\cdot\left(w_{\max} \cdot |\kappa(t)-\kappa(t-1)|+ \MovOPT_t+
\HitOPT_t\right),
\\
\label{eq:hit_conservation} \HitAlg_t  + \DelPhi^h_t
+\frac{1}{\alpha}\DelPhi^m_t & \leq & (1+\eps)\left(w_{\max}
\cdot |\kappa(t)-\kappa(t-1)|+ \MovOPT_t+  \HitOPT_t\right).
\end{eqnarray}

Here, $\HitAlg_t$ (respectively $\HitOPT_t$) and $\MovAlg_t$
(respectively $\MovOPT_t$) denote the hit and movement costs
incurred by the algorithm (respectively optimum) at time $t$. The
quantities
$$\DelPhi^h_t := \Phi^h(\oy^t, t)-\Phi^h(\oy^{t-1},t-1),%
\qquad \DelPhi^m_t := \Phi^m(\oy^t, t)-\Phi^m(\oy^{t-1},t-1)$$
respectively denote the change in the potentials $\Phi^h$  and
$\Phi^m$ at time step $t$.

As we shall see, it will be the case that $\Phi^m(\oy,t)\geq 0$
and $\Phi^m(\oy^0,0)=0$. Moreover, both $\Phi^m(\oy,t)$ and
$\Phi^h(\oy,t)$ will be bounded by some universal constant $C$,
independent of the length of the request sequence. Thus, Theorem
\ref{thm:frac_alloc_variable} will follow by summing up
\eqref{eq:conservation} and \eqref{eq:hit_conservation} over all
times $t$. (Note that in our algorithm $\alpha=\log (1+1/\beta)= O(\log (k/\eps))$.)

\paragraph{The potential functions.} The  potential function $\Phi^m$ is defined as follows.
\begin{eqnarray*}
\Phi^m(\oy, t) & := & (1+\eps)\cdot \sum_{i} w_i \left(
\sum_{j} y^{*t}_{i,j} \cdot \log\left(
\frac{1+\beta}{y^{\eta}_{i,j}+\beta}\right) \right).
\end{eqnarray*}
Here, $\oy^{*t}$ denotes the configuration of the optimum solution at time $t$. Note that if the configuration $\oy^{*t}$
is integral and the optimum has $k^*_i$ servers at location $i$ at time $t$, then the contribution of location $i$ to $\Phi^m(\oy^t,t)$ is
$$w_i \left( \sum_{j > k_i^*} \log \left(\frac{1+\beta}{y^t_{i,j}+\beta}\right) \right).$$

So, roughly speaking, $\Phi^m(\oy^t,t)$ accounts for the excess
servers in the online configuration $\oy^t$ at location $i$
compared to the optimum solution. For example, suppose that
$\oy^{t}$ has $k_i$ servers at location $i$, i.e., $y_{i,j}^t=0$
for $j\leq k_i$, for some $k_i>k_i^*$, and $y_{i,j}^t=1$
otherwise.  Then the contribution of location $i$ to
$\Phi^m(\oy^t,t)$ is $O( w_i (k_i-k_i^*) \log k)$. Intuitively,
the offline adversary can penalize the online algorithm for
``wasting" $k_i-k_i^*$ servers at $i$, by giving cost vectors at
the other locations (where the optimum has more servers), and
making it pay a larger hit cost. The potential $\Phi^m(\oy^t,t)$
will be used to offset this additional hit cost in such
situations.

Next, we define the potential $\Phi^h$ to be
\begin{eqnarray*}
\Phi^h(\oy, t) & := &  \frac{1}{\alpha} \left(\sum_{i,j} w_i
\cdot y^{t}_{i,j} \right).
\end{eqnarray*}
It is easily verified that both potentials are bounded. Moreover
$\Phi^m(\oy^0,0)=0$, as both offline and online are assumed to
start with the same initial configuration.

\paragraph{The proof plan.}
Our goal now is to show that Inequalities \eqref{eq:conservation}
and \eqref{eq:hit_conservation} always hold. For ease of analysis,
we will consider the events at time $t$ in three steps, and show
that Inequalities \eqref{eq:conservation} and
\eqref{eq:hit_conservation} hold at each of these steps. The steps
are the following:
\begin{enumerate}
\item The quota $\kappa(t)$ either increases, decreases, or stays
unchanged, compared to $\kappa(t-1)$, and the optimal solution
removes/adds servers accordingly.
\item The optimal solution moves some
servers and its state changes from $\oy^{*t-1}$ to $\oy^{*t}$.
\item The online algorithm changes its state from $\oy^{t-1}$ to $\oy^{t}$: first, we analyze the fix stage then we analyze the hit stage. Also, while analyzing the hit stage, the hit costs of both the online algorithm and the optimal solution are accounted for.
\end{enumerate}

\paragraph{The server quota increases/decreases and the optimum removes/adds servers.}

Note that the only quantities that can change in this step are the
movement cost of the optimum and the potential $\Phi^m$, thus it
suffices to prove that Inequality \eqref{eq:conservation} is
preserved (in this case, \eqref{eq:hit_conservation} is identical
to \eqref{eq:conservation} scaled by $1/\alpha$).
Now, there are two cases to consider, depending on how $\kappa(t)$ changes.
\begin{itemize}
\item Suppose $\kappa(t) <  \kappa(t-1)$, then the optimum has to
withdraw $|\kappa(t)-\kappa(t-1)|$ servers from some locations.
That is, for $|\kappa(t)-\kappa(t-1)|$ locations $(i,j)$, the
corresponding variables $y_{i,j}^{*(t-1)}$ are set to 1, and as a
result these locations start contributing to $\Phi^m$. Clearly,
each such location $(i,j)$ increases $\Phi^m$ by at most
$$ (1+\eps) w_{i} \cdot \log \left(\frac{1+\beta}{y_{i,j}^{t-1}+\beta}\right)\leq  (1+\eps)\ln\left(1+\frac{1}{\beta}\right)\cdot w_{\max}
=(1+\eps)\alpha \cdot w_{\max}.$$

As a result, the total increase of $\Phi^m$ is at most $w_{\max}
\cdot (1+\eps) \alpha  \cdot |\kappa(t)-\kappa(t-1)|$, and
hence \eqref{eq:conservation} holds even without accounting for
the movement cost of the optimum.

\item  Suppose $\kappa(t) \geq \kappa(t-1)$, then the optimum
can bring in $|\kappa(t)-\kappa(t-1)|$ servers, and as a result
$\Phi^m$ can only decrease, as some terms $y^{*(t-1)}_{i,j}$ may
change from $1$ to $0$, thus not contributing anymore to $\Phi^m$.
Hence, \eqref{eq:conservation} also holds in this case.
\end{itemize}

\paragraph{The optimum moves its servers.}

Without loss of generality, it suffices to analyze the case in
which the optimum moves exactly one server from  location $i$ to
$i'$. (If multiple servers are moved, we can consider these moves
one by one.) Also, as before, only $\Phi^m$ and the offline
movement cost change, and hence it suffices to just show that
\eqref{eq:conservation} holds, and in particular that
$$ \DelPhi^m_t \leq (1+\eps) \alpha \MovOPT_t.$$
Suppose that location $i$ had $j$ servers prior to the movementand
this number is reduced to $j-1$ (recall that by our convention we
account only for the movement cost corresponding to withdrawal of
servers). Then, the contribution of location $i$ to $\Phi^m$
increases by precisely
$$ w_i (1+\eps) \cdot  \log \left( \frac{1+\beta}{y^{t-1}_{i,j}+\beta}\right)  \leq  w_i (1+\eps) \cdot  \alpha =  (1+\eps) \alpha \MovOPT_t.$$

In contrast, increasing the number of servers at $i'$ can only
decrease $\Phi^m$. Thus, we get that the desired inequalities hold
in this case.

\paragraph{The online algorithm is executed.}

The case in which the online algorithm changes its distribution
$\oy^t$ is the most interesting, and we analyze Inequalities
\eqref{eq:conservation} and \eqref{eq:hit_conservation}
separately. Recall that our online algorithm works in two steps:
the fix stage and the hit stage, and hence we consider these
separately. Moreover, as the evolution of the algorithm is
described in a continuous manner, we will analyze Inequalities \eqref{eq:conservation} and
\eqref{eq:hit_conservation} in such a manner too.

\paragraph{The fix stage: proof of Inequality \eqref{eq:conservation}.}
To show that Inequality \eqref{eq:conservation} holds during the fix stage it suffices to prove that for any $\tau$,
\begin{equation}
\label{eq:fixstage4}
\frac{d M_t}{d \tau} + \frac{d \Phi^m}{d \tau} \leq 0.
\end{equation}
By definition of  $\frac{d y_{i,j}^\tau}{d\tau}$ during the fix
stage and our way of accounting for the movementcost (cf. Claim
\ref{clm:mvcost}) we have
$$
\frac{d M_t}{d \tau} =  \sum_{(i,j) \in A^{\tau}} w_i \cdot \frac{1}{w_i} \left(y^\tau_{i,j}+\beta\right) =
\sum_{(i,j) \in A^{\tau}}\left(y^{\tau}_{i,j}+\beta\right).$$
(Recall that $A^{\tau}$ is the set of active coordinates $(i,j)$,
i.e. those for which $y_{i,j} < 1$.)
Also, it is easy to see that the change in the potential function $\Phi^m$ is

$$
\frac{d \Phi^m}{d \tau} = -(1+\eps)\sum_{(i,j) \in A^{\tau}} y^{*t}_{i,j}.
$$
Next we need to prove the following claim

\begin{claim}\label{clm:mov_cost_inequality}
Consider a subset $A$ of the coordinates and two configurations $\oy$ and $\oy'$ with $\sum_{(i,j)\in A} y_{i,j}'\geq 1$, and $\sum_{i,j} y_{i,j}'\geq kd - k$, we have then that
$$
\sum_{(i,j)\in A} (y_{i,j}+\beta) - (1+\eps) \sum_{(i,j)\in A} y_{i,j}' \leq \sum_{(i,j)\in A} y_{i,j}-\sum_{(i,j)\in A} y_{i,j}'.
$$
\end{claim}
Before we prove this claim, let us describe how \eqref{eq:fixstage4} follows from it. If we set $\oy=\oy^\tau$, $\oy'=\oy^{*t}$, and $A=A^\tau$, then, clearly, $\sum_{i,j} y_{i,j}^{*t}=kd-\kappa(t)\geq kd - k$. Furthermore, since $(i,j)\notin A^\tau$ only if $y_{i,j}^\tau=1$, and as we apply the fix stage only if $\sum_{i,j} y_{i,j}^{\tau}<kd - \kappa(t)$, we need to have $|A^\tau|>\kappa(t)$, and thus
$$
\sum_{(i,j)\in A} y_{i,j}^{*t}\geq kd - \kappa(t) - |\overline{A^\tau}|\geq 1.
$$
So, both requirements of the claim are satisfied and it follows that
\begin{equation}\label{eq:fix_stage}
\frac{d M_t}{d \tau} + \frac{d \Phi^m}{d \tau} = \sum_{(i,j)\in A^\tau} (y_{i,j}^\tau+\beta) - (1+\eps) \sum_{(i,j)\in A^\tau} y_{i,j}^{*t} \leq \sum_{(i,j)\in A^\tau} y_{i,j}^\tau-\sum_{(i,j)\in A ^\tau} y_{i,j}^{*t}\leq 0,
\end{equation}
where the last inequality follows, as $\sum_{i,j} y_{i,j}^{\tau}<kd - \kappa(t)=\sum_{i,j}y_{i,j}^{*t}$, and  $(i,j)\notin A^\tau$ only if $y_{i,j}^\tau=1$ implies that $\sum_{(i,j)\in A^{\tau}} y_{i,j}^{\tau}\leq \sum_{(i,j)\in A^{\tau}} y_{i,j}^{*t}$.

\begin{proof}[Proof of Claim \ref{clm:mov_cost_inequality}]
As $\sum_{i,j} y_{i,j}' \geq kd - k$ and $\sum_{(i,j)\in A} y_{i,j}' \geq 1$, we have that
\begin{equation}\label{eq:claim_mov_cost_2}
\frac{|A|}{k+1} - \sum_{(i,j)\in A} y_{i,j}' \leq \frac{|A|}{k+1} - \max\{|A|-k,1\} \leq 0,
\end{equation}
where the last inequality follows as $\frac{|A|}{k+1}<1$ when $|A|\leq k$, and $\frac{|A|}{k+1} \leq |A|-k$ when $|A|\geq k+1$.

Now, the claim is proved by noticing that
\begin{eqnarray*}
\sum_{(i,j)\in A} (y_{i,j}+\beta) - (1+\eps) \sum_{(i,j)\in A} y_{i,j}' & = & \eps \left( \frac{|A|}{k+1} - \sum_{(i,j)\in A} y_{i,j}'\right) ~+ ~\left( \sum_{(i,j)\in A} y_{i,j}-\sum_{(i,j)\in A} y_{i,j}'\right) \\
&\leq& \sum_{(i,j)\in A} y_{i,j}-\sum_{(i,j)\in A} y_{i,j}',
\end{eqnarray*}
where the inequality follows by \eqref{eq:claim_mov_cost_2}.
\end{proof}

\paragraph{The fix stage, proof of Inequality \eqref{eq:hit_conservation}.}
To prove that Inequality \eqref{eq:hit_conservation} holds during the fix stage as well, we note that as we are accounting for
the hit cost during the hit stage, currently there is no hit cost incurred, and hence we just need to show that
\begin{equation}
\label{fixstage3}
 \frac{d \Phi^h}{d \tau } + \frac{1}{\alpha} \frac{ d \Phi^m}{d\tau}=0.
 \end{equation}
Observe that
$$\frac{d\Phi^h}{d \tau} = \frac{1}{\alpha} \sum_{(i,j) \in A^\tau} w_i \frac{d y_{i,j}}{d\tau}$$
which is exactly $1/\alpha$ times our accounting for the movement
cost in the fix stage (cf. Claim \ref{clm:mvcost}). Therefore, \eqref{fixstage3} is identical
to \eqref{eq:fix_stage} (up to a scaling by $1/\alpha$), and hence follows from
the above proof.

\paragraph{The hit stage, proof of Inequality \eqref{eq:hit_conservation}.}
Recall that both the optimal solution and the online algorithm incur a hit cost in this stage.
We start by proving that Inequality \eqref{eq:hit_conservation} holds during this stage -- later, we will analyze Inequality \eqref{eq:conservation}. We need to show that
 \begin{equation}
 \label{hitstage2}
 \frac{d H_t}{d \eta} +  \frac{d \Phi^h_t}{d \eta} + \frac{1}{\alpha}\frac{d \Phi^m_t}{d \eta} \leq (1+\eps) \frac{d H^*_t}{d \eta}.
 \end{equation}
First, we note that by our way of accounting for the hit cost of the algorithm (cf. \eqref{eq:hit_cost_account})
$$
\frac{d H_t}{d \eta}  =  \sum_{i,j} \lambda^\eta_{i,j} y^\eta_{i,j}
\qquad \textrm{and} \qquad \frac{d \Phi^h_t}{d \eta}  =
\frac{1}{\alpha}\sum_{(i,j) \in
A^\eta}\left(y^\eta_{i,j}+\beta\right)\cdot \left(N(\eta) - \alpha
\lambda^\eta_{i,j}\right).
$$
Now, we have that
\begin{eqnarray*}
\frac{d H_t}{d \eta} +  \frac{d \Phi^h_t}{d \eta} & =& \sum_{i,j} \lambda^\eta_{i,j} y^\eta_{i,j} + \frac{1}{\alpha}\sum_{(i,j) \in
A^\eta}\left(y^\eta_{i,j}+\beta\right)\cdot \left(N(\eta) - \alpha
\lambda^\eta_{i,j}\right)\\
& = &\sum_{(i,j)\notin A^{\eta}} \lambda^\eta_{i,j} y^\eta_{i,j}
+ \frac{1}{\alpha}\sum_{(i,j) \in A^\eta} (y^\eta_{i,j}+\beta) \cdot N(\eta) -  \sum_{(i,j) \in A^\eta} \beta  \lambda^{\eta}_{i,j}\\
& \leq & \sum_{(i,j)\notin A^{\eta}} \lambda^\eta_{i,j} y^\eta_{i,j}
+ \frac{1}{\alpha}\sum_{(i,j) \in A^\eta} (y^\eta_{i,j}+\beta) \cdot N(\eta)=  \sum_{(i,j)\notin A^{\eta}} \lambda^\eta_{i,j} y^\eta_{i,j}
+ \frac{1}{\alpha} \frac{d M_t}{d\eta},
\end{eqnarray*}
where in the last step we used the expression for our movement cost
accounting (cf. Claim \ref{clm:mvcost}),
\begin{eqnarray}
\frac{d M_t}{d \eta} &= & \sum_{(i,j) \in
A^{\eta}}\left(y^{\eta}_{i,j}+\beta\right)\cdot N(\eta).\label{eq:move_cost_frac}
\end{eqnarray}
Thus, to establish \eqref{hitstage2}, it suffices to show that
\begin{equation}
\label{hitstage5}
 \sum_{(i,j)\notin A^{\eta}} \lambda^\eta_{i,j} y^\eta_{i,j}
+ \frac{1}{\alpha} \frac{d M_t}{d\eta}   + \frac{1}{\alpha}\frac{d \Phi^m_t}{d \eta} \leq (1+\eps) \frac{d H^*_t}{d \eta}.
\end{equation}
By our update rule for $y^{\eta}_{i,j}$, the derivative of $\Phi^m_t$ is
\begin{eqnarray*}
 \frac{d \Phi^m_t}{d \eta} & = &
-(1+\eps)\sum_{(i,j) \in A^{\eta}} y^{*t}_{i,j} \cdot
\left(N(\eta)- \alpha  \lambda^\eta_{i,j}\right).
\end{eqnarray*}
Finally, by our way of accounting for the hit cost of the optimal solution (cf. \eqref{eq:opt_hit_cost_account}) we have
$$
\frac{d H^{*t}}{d \eta} = \sum_{i,j}  \lambda^\eta_{i,j} y^{*t}_{i,j}.
$$
Thus, after multiplying Inequality \eqref{hitstage5} by $\alpha$ and plugging in the above equalities, we need to prove that
\begin{equation}\label{alt:expression}
\alpha \sum_{(i,j)\notin A^{\eta}} \lambda^\eta_{i,j} y^\eta_{i,j} + \sum_{(i,j) \in
A^{\eta}}\left(y^{\eta}_{i,j}+\beta\right)\cdot N(\eta)-(1+\eps)\sum_{(i,j) \in A^{\eta}} y^{*t}_{i,j} \cdot
\left(N(\eta)- \alpha  \lambda^\eta_{i,j}\right)- (1+\eps) \alpha \sum_{i,j}  \lambda^\eta_{i,j} y^{*t}_{i,j} \leq 0.
\end{equation}
Let $A'$ be the set of all coordinates that are either active or have $y_{i,j}^\eta=0$.
As any inactive coordinate has $y_{i,j}^\eta\in \{0,1\}$, we observe that:

\begin{itemize}
\item If $(i,j)\in A'\setminus A^\eta$, then it must be that $y^\eta_{i,j}=0$ and $\alpha \lambda_{i,j}^\eta \geq N(\eta)$. This holds, as a coordinate for which $y_{i,j}^\eta=0$ can be inactive only if $\frac{d y_{i,j}^\eta}{d\eta}=(y_{i,j}^\eta + \beta)(N(\eta)-\alpha \lambda)\leq 0$.
\item If $(i,j)\notin A'$ it must be that $y_{i,j}^{\eta}=1$ and $\alpha \lambda_{i,j}^\eta < N(\eta)$. This holds, as a coordinate for which $y_{i,j}^\eta=1$ can be inactive only if $\frac{d y_{i,j}^\eta}{d\eta}=(y_{i,j}^\eta + \beta)(N(\eta)-\alpha \lambda) > 0$.
\end{itemize}
Thus, using the above observations, we may rewrite \eqref{alt:expression} and get that

\begin{eqnarray} \lefteqn{N(\eta)\cdot \left(\sum_{(i,j) \in
A^{\eta}}\left(y^{\eta}_{i,j}+\beta\right)-(1+\eps)\sum_{(i,j) \in A^{\eta}} y^{*t}_{i,j}\right)- (1+\eps) \alpha \sum_{(i,j)\notin A^\eta}  \lambda^\eta_{i,j} y^{*t}_{i,j} + \alpha \sum_{(i,j)\notin A^{\eta}} \lambda^\eta_{i,j} y^\eta_{i,j}} \nonumber \\
& \leq & N(\eta)\cdot \left(\sum_{(i,j) \in
A'}\left(y^{\eta}_{i,j}+\beta\right)-(1+\eps)\sum_{(i,j) \in A'} y^{*t}_{i,j}\right)- (1+\eps) \alpha \sum_{(i,j)\notin A'}  \lambda^\eta_{i,j} y^{*t}_{i,j} + \alpha \sum_{(i,j)\notin A^{\eta}} \lambda^\eta_{i,j} y^\eta_{i,j} \nonumber \\
& \leq & N(\eta)\cdot \left(\sum_{(i,j) \in
A'}\left(y^{\eta}_{i,j}+\beta\right)-(1+\eps)\sum_{(i,j) \in A'} y^{*t}_{i,j}\right) + \alpha \sum_{(i,j) \notin A'}\lambda^\eta_{i,j} \left(y^\eta_{i,j} - y^{*t}_{i,j}\right) \nonumber \\
& \leq & N(\eta)\cdot \left(\sum_{(i,j) \in
A'}\left(y^{\eta}_{i,j}+\beta\right)-(1+\eps)\sum_{(i,j) \in A'} y^{*t}_{i,j} + \sum_{(i,j) \notin A'}\left(y^\eta_{i,j} - y^{*t}_{i,j}\right)\right). \nonumber%
\end{eqnarray}
The first inequality
follows as $(i,j)\in A'\setminus A^\eta$, then it must be that $\alpha \lambda_{i,j}^\eta \geq N(\eta)$, and because $A^{\eta} \subseteq A'$.
The second inequality %
follows since for all $(i,j)\in A'\setminus A^\eta$ it must be that $y^\eta_{i,j}=0$.
The third inequality %
follows as $\alpha \lambda_{i,j}^\eta < N(\eta)$ for $(i,j)\notin A'$, and by the observation that for each $(i,j) \notin A'$, $y^\eta_{i,j}=1 \geq  y^{*t}_{i,j}$.

If $N(\eta)=0$, then the above expression equals 0, and we are done. Otherwise, if $N(\eta)>0$, we get that

\begin{eqnarray}
\lefteqn{\sum_{(i,j) \in
A'}\left(y^{\eta}_{i,j}+\beta\right)-(1+\eps)\sum_{(i,j) \in A'} y^{*t}_{i,j} + \sum_{(i,j) \notin A'}\left(y^\eta_{i,j} - y^{*t}_{i,j}\right)} \nonumber \\
& \leq & \sum_{(i,j) \in
A'}y^{\eta}_{i,j}-\sum_{(i,j) \in A'} y^{*t}_{i,j} + \sum_{(i,j) \notin A'}\left(y^\eta_{i,j} - y^{*t}_{i,j}\right) \label{ineq:hit3} \\
& = & \sum_{(i,j)}y^{\eta}_{i,j}-\sum_{(i,j)} y^{*t}_{i,j} =0.
\end{eqnarray}
Inequality \eqref{ineq:hit3} follows from Claim \ref{clm:mov_cost_inequality} by the following arguments.
>From Lemma \ref{lem:wd} we know that the set $A^{\eta}$ is non-empty. We claim that this implies that $\sum_{(i,j)\in A'} y_{i,j}^\eta$ has to be positive. Otherwise, all the active coordinates would have $y_{i,j}^\eta=0$, and thus could only increase, contradicting the fact that when $N(\eta)>0$, $\sum_{i,j} \frac{d y_{i,j}^\eta}{d\eta} = \sum_{(i,j)\in A^\eta} \frac{d y_{i,j}^\eta}{d\eta} = 0$.
Moreover, as $N(\eta)>0$, $\sum_{i,j} y_{i,j}^\eta$ is equal to $\sum_{i,j} y^{*t} = kd - \kappa(t)$ (and thus is integral). Now, as $y_{i,j}^{\eta}=1$ if $(i,j)\notin A'$, we have that $\sum_{(i,j)\in A'} y_{i,j}^{*t}\geq \sum_{(i,j)\in A'} y_{i,j}^{\eta}$ and that $\sum_{(i,j)\in A'} y_{i,j}^{\eta}$ is integral as well. As this quantity is positive by the argument above, we get that
\begin{equation}
\label{eq:Aprime1}
\sum_{(i,j)\in A'} y_{i,j}^{*t}\geq \sum_{(i,j)\in A'} y_{i,j}^{\eta}\geq 1.
\end{equation}
Thus, we can use Claim \ref{clm:mov_cost_inequality} with $\oy=\oy^\eta$, $\oy'=\oy^{*t}$ and $A=A'$ (as all the requirements of this claim are satisfied). The last equality follows, since for $N(\eta)>0$, we get from the algorithm, $\sum_{i,j} y_{i,j}^\eta = kd - \kappa(t) = \sum_{i,j} y^{*t}_{i,j}$.

\paragraph{The hit stage, proof of Inequality \eqref{eq:conservation}.}
To show that  \eqref{eq:conservation} holds,  we need to show that
\begin{equation}
\label{hitcase}
 \frac{d M_t}{d \eta} + \frac{d \Phi^m_t}{d \eta} \leq (1+\eps) \alpha \cdot \frac{d H^{*t}}{d \eta}.
 \end{equation}
However, this follows directly by noting that the above is simply Inequality \eqref{hitstage5} after removing the first (non-negative) term and scaling by $\alpha$.

\section{Fractional $k$-server on Weighted HSTs}
\label{s:alloc_to_kserver}
In this section, we show how the fractional allocation algorithm on a weighted star can be used as a building block to obtain a fractional $k$-server
solution on a weighted HST. In particular, we prove the following.

{
\renewcommand{\thethm}{\ref{thm:alloc_to_kserver}}
\begin{thm}
Let $T$ be a weighted $\sigma$-HST with depth $\ell$. If, for any
$0\leq \eps \leq 1$, there exists a $(1+\eps,\log
(k/\eps))$-competitive algorithm for the fractional allocation
problem on a weighted star, then there is an $O(\ell \log
(k\ell))$-competitive  algorithm for the fractional $k$-server
problem on $T$, provided $\sigma = \Omega( \ell \log(k \ell))$.
\end{thm}
\addtocounter{thm}{-1}
}

To this end, we focus on a particular weighted $\sigma$-HST $T$
and show how to construct a fractional $k$-server algorithm on it.
Roughly speaking, the construction works as follows. Each internal
node $p$ of $T$ will run a number of instances of the allocation
problem which differ with respect to their quota patterns, but
have the same hit cost vectors. These instances are maintained as
a convex combination. The fractional solutions to the different
instances, which we compute online using the fractional allocation
algorithm, determine in a recursive manner how the servers
available at each node are distributed among its children.

While this approach is similar to the approach of Cot\'{e} et
al.~\cite{CMP08}, the main difference here is that we can use the
(much weaker) fractional allocation problem instead of using a
randomized (integral) algorithm for the allocation problem.

Let us denote by $r$ the root of our $\sigma$-HST $T$. Recall that
for a node $p$ of $T$, $T(p)$ denotes the subtree
rooted at $p$, $W(p)$ is the length of the edge connecting
$p$ to its parent, and $w(p,i)$ denotes the length of the edge
connecting $p$ to its child $p_i$. By the definition of a weighted
$\sigma$-HST, we have  $W(p) \geq \sigma w(p,i)$ for all children
$i$ of $p$, unless $p$ is either a leaf or the root.

Recall that the input to the fractional allocation problem running at node
$p$ consists of
the quota pattern $\kappa=(\kappa(1),\kappa(2),\ldots)$ specifying
the number of servers $\kappa(t)$ available at each time $t$, and
the hit cost vectors $h^t$ that arrive at each time $t$ at
location $\ooi$.
The output of an algorithm for the fractional allocation problem
specifies a fractional solution $\ox^t$ that provides a
distribution on the number of servers at each location $p_i$,
subject to the aggregate bound of $\kappa(t)$ on the (expected)
number of servers.

Now, let us fix some instance $\rho$ of the $k$-server problem on
the  leaves of $T$. Let $\rho = (\rho(1),\rho(2),\ldots)$ be the
request sequence, where $\rho(t)$ denotes the leaf requested at
time $t$.

\begin{defn}
For a node $p$, integer $j$ and time $t$, let $\optcost(p,j \cdot
\vec{1},t)$ be the optimum cost for serving the $k$-server
instance $\{\rho(1),\ldots,\rho(t)\} \cap T(p)$, i.e. the request
sequence $\rho$ restricted to the leaves of $T(p)$ until time $t$,
subject to the constraint that exactly $j$ servers are available.
\end{defn}

{\em Remark:} $\optcost$ is well defined only with respect to an
initial configuration, which we will always assume to be the
initial starting positions of the servers at $t=0$. Also, we use
the notation  $\optcost(p,j \cdot \vec{1},t)$, instead of just
$\optcost(p,j,t)$, as we will later extend the definition of
$\optcost$ to the case in which $j$ can vary with time. For now,
we only consider fixed $j$.

\subsection{The algorithm}
\label{s:alloctokserver} In this section we define the ensemble of
allocation problems that will be running at each of the  internal
nodes of $T$. To do this, we have to define how the hit-cost
vectors and the quota patterns are generated. Consider some
internal node $p$. As mentioned earlier, each internal node $p$
will run several instances of the allocation problem that are
different with respect to their quota pattern. It will also hold a
convex combination over these instances. All instances will have
the same hit cost vector that will be defined later. We denote the
convex combination over allocation instances on node $p$ at time
$t$ by $\Lambda_p^t$.
$\Lambda_p^t$ is specified via the collection
$$  \Lambda_p^t = \{(\lambda_{p,s}^t,\kappa_{p,s}^t,H_{p}^t)\}_s, \qquad \forall t,p, \sum_{s} \lambda_{p,s}^t=1$$
Here $\lambda_{p,s}^t$ determines the fraction at time $t$ given to the instance with quota pattern $\kappa_{p,s}^t$ (until time $t$), and
$H_{p}^t=\{h^1_{p},h^2_{p},\ldots,h^t_{p}\}$ is the sequence of hit cost vectors that have appeared until time $t$. As we will see shortly,
the hit cost vector will be the same for all instances $s \in \Lambda_p^t$ and therefore there is no subscript $s$ to the hit cost.
We will use the notation $s\in \Lambda_p^t$ to index the triples $(\lambda_{p,s}^t,\kappa_{p,s}^t,H_{p}^t)$ in $\Lambda_p^t$.
As we shall see later, the convex combination $\Lambda_p^t$ will be a refinement of the convex combination $\Lambda_p^{t-1}$, for every $t$.

To complete the description of the fractional $k$-server algorithm
we need to define $\Lambda_p^t$ for each node $p$ and time $t$,
and show how the fractional number of servers at the leaves of $T$
is computed. We begin with defining how the hit costs $h_{p}^t$
are generated for each node $p$.

\paragraph{Hit costs:} Consider any internal node $p$. Let $p_1,\ldots, p_d$ be the children of $p$. For the allocation problems running at $p$, at time $t$ we give the hit cost vector
$$h^t_{p_i}(j) = \optcost(p_i,j \cdot \vec{1},t) - \optcost(p_i,j \cdot \vec{1},t-1).$$
As Cot\'{e} et al.~\cite{CMP08} prove, the cost vectors $h^t_{p_i}$ have the desired monotonicity property, i.e., $h^t_{p_i}(0)\geq h^t_{p_i}(1)\geq \ldots \geq h^t_{p_i}(k)$ for each $i$ and time $t$.
The following crucial observation follows directly from the definition of the $k$-server problem.

\begin{obs}\label{obs:infcost}
Consider subtree $T(p)$ and request $\rho(t)$. If $\rho(t) \in T(p_i)$, then
\begin{enumerate}
\item $h_p^t(i,0)=\infty$. (This follows since any $0$-server solution is infeasible for any instance with one or more requests, or equivalently incurs infinite cost.).
\item $h^t_p(i',j)=0$  for all $i'\neq i$ and for all $j$. (This follows simply since the request is not in the sub-tree of $p_{i'}$ for $i' \neq i$.).
\end{enumerate}
\end{obs}

This completes the description of the cost vectors of node $p$. We next define the quota patterns $\kappa_{p,s}(t)$ for the various allocation instances running at node $p$.

\paragraph{Quota patterns:}

The quota patterns are determined recursively in a top down manner over the tree $T$ (and inductively over time) by the fractional solutions of the allocation instances that are generated at each node. To specify how these patterns evolve, we describe below a procedure for updating both the quota patterns $\kappa^t_{p,s}$ and the convex combination $\lambda^t_{p,s}$, associated with the allocation instances maintained.

\paragraph{Base case:}

\begin{enumerate}
\item At the root $r$ of the tree $T$ there is a single allocation instance running  with a quota of $k$ at all times.
That is, $\Lambda_r^t$ consists of a single allocation instance (with fraction 1), hit costs as described above, and $\kappa = k \cdot \vec{1}.$
\item  For any internal node $p \in T$ and time $t=0$, $\Lambda^0_p$ consists of a single allocation instance (with fraction 1). The quota pattern $\kappa_{p,s}$ for this single instance $s$, until time $t=0$, is simply the number of servers present initially at the leaves of subtree $T(p)$. Moreover, there is no hit cost thus far.
\end{enumerate}

\paragraph{The inductive step:}
Consider time $t$. We describe the procedure to obtain $\Lambda_p^{t}$ from $\Lambda_p^{t-1}$ in a top
down manner on the tree as follows.
As the base case, recall that  $\Lambda_r^t$ has already been determined for all $t$.
Arguing inductively top down on the tree, suppose that $\Lambda_p^t$ has already been determined. Then, for the children $p_1,\ldots,p_d$ of $p$,
we determine $\Lambda_{p_i}^t$ as follows.

Consider the allocation instances that are executed at node $p$. Let $\{x^t_{i,j,s}\}_{i,j,s}$ be the fractional solutions generated (by the allocation instances) at time $t$. The algorithm will maintain the following consistency between the quota for servers available at $p_i$  and what the allocation problems running at the parent $p$ determine. In particular,

\begin{equation}\label{eq:lambda_consist}
\mbox{(Consistency)} \qquad  \sum_{s\in \Lambda_{p_i}^t  \ | \ \kappa^t_s(t)=j} \lambda^t_s = \sum_{s \in \Lambda^t_p}  \lambda^t_s x_{i,j,s} \triangleq x_{i,j}^t,
\end{equation}

Also, for each child $p_i$ it should maintain

\begin{equation}\label{eq:conv_lambda}
\mbox{(Convex combination)} \qquad \sum_{s \in \Lambda_{p_i}^t} \lambda^t_s =1.
\end{equation}

Suppose that $x_{i,j}^{t-1}$ changes to $x_{i,j}^{t}$ due to the execution of the allocation instances $s \in \Lambda_p^t$ at time $t$. We show how to update
$\Lambda_{p_i}^t$ from $\Lambda_{p_i}^{t-1}$ such that it remains consistent with (\ref{eq:lambda_consist}) and (\ref{eq:conv_lambda}). This update will be done in a natural (and cost-efficient) way.

Consider first the cost paid by the convex combination of the allocation instances running at node $p$.
The cost is
\begin{eqnarray}
 \sum_s \sum_i w(p,i) \lambda^t_s \sum_{j=1}^{k} \left|\sum_{\ell<j}\left(x_{i,\ell,s}^t - x_{i,\ell,s}^{t-1}\right)\right| \geq
  \sum_i w(p,i) \sum_{j=1}^{k} \left|\sum_s \lambda^t_s\sum_{\ell<j}\left(x_{i,\ell,s}^t - x_{i,\ell,s}^{t-1}\right)\right|,
\end{eqnarray}
where the inequality follows, since for any non-negative numbers
$p_i$, $\sum_{i}p_i|a_i| \geq |\sum_{i}p_i a_i|$.

We note that the change from  $x^{t-1}_{i,j}$ to $x^{t}_{i,j}$ can
be decomposed into a collection of elementary moves in which $\pm
\delta(i,j)$ units of mass are removed from $x_{i,j}$ and put on
$x_{i,j\pm 1}$, such that the total fractional movement cost
remains the same. Thus, we can assume without loss of generality
that $x^t_{i,j}$ and $x^{t-1}_{i,j}$ differ by an elementary move.

Consider an elementary move where $x^{t}_{i,j} = x^{t-1}_{i,j}-
\delta$ and $x^t_{i,j-1} = x^t_{i,j-1}+ \delta$ (all other types
of elementary moves are handled analogously). To implement this,
we choose an arbitrary $\delta$ measure of allocation problems
$s\in \Lambda_{p_i}^{t-1}$ with $\kappa_s(t-1) =j$ and set
$\kappa_s(t)$ to $j-1$. For all other $\kappa_s$, we set
$\kappa_s(t)=\kappa_s(t-1)$. After all  entries are updated by
applying the elementary moves, $\kappa_s(t)$ is determined.

It is clear that this update rule maintains both (\ref{eq:lambda_consist}) and  (\ref{eq:conv_lambda}).
This completes the procedure for obtaining $\Lambda_i^t$ from $\Lambda_i^{t-1}$.

\paragraph{Obtaining the fractional $k$-server solution:}

To complete the description of our algorithm we should describe how to determine the fractional number of servers at each leaf $q$ at each time $t$. This is determined in a natural way using the following observation.
Consider a leaf $q$ and let $p$ be its parent. Then,

$$z(q,t):=\sum_{s\in \Lambda_{p}^t} \lambda_s^t \sum_j  j \cdot x^t_{q,j,s}$$
is the number of servers at $q$ at time $t$. Here, $x^t_{q,j,s}$ is the probability of having $j$ servers at $q$  at time $t$, when the fractional allocation algorithm is applied to the allocation instance $s \in \Lambda_{p}^t$.

\subsection{Feasibility}

We first note that our fractional $k$-server solution is feasible since it satisfies the following.

\begin{lemma}
\label{l:feasibility}
Whenever there is a $k$-server request $\rho(t)$,
then there is at least one server unit at the location $\rho(t)$, i.e. $z_{\rho(t),t} \geq 1$. This holds provided the  total cost incurred by the allocation problems is finite.
\end{lemma}
\begin{proof}
The lemma follows by the way the hit costs are generated. Suppose leaf $q$ is requested at time $t$, and $q$ is the $i$-th child of its parent $p$.
Then, by observation \ref{obs:infcost} (part 1), the hit cost entry $h^t(i,0)$ for every allocation instance running at $p$ is $\infty$.
Thus, if the total cost of the allocation problems is finite, it must be that for each $s \in \Lambda_p^t$, the algorithm ensures that
$x^t_{q,0,s}=0$. Since, $\sum_j x^t_{q,j,s}=1$ for all $s$, and $\sum_{s \in \Lambda_p^t} \lambda_s =1$, it follows  that
$z(q,t) = \sum_{s\in \Lambda_{p}^t} \lambda_s^t \sum_j  j \cdot x^t_{q,j,s} \geq 1.$
\end{proof}

{\em Remark:} Lemma \ref{l:feasibility} assumes that the total cost of the allocation problems is finite. Later on we  show that the cost is in fact bounded by at most a polylogarithmic factor from the optimal $k$-server cost, and hence finite.

\subsection{Performance analysis}

We first show that the cost of the fractional $k$-server solution we generate (at the leaves of the tree) is at most the total convex combination cost of the allocation instances running on $T$.
For a node $p$ (not necessarily a leaf) in $T$, let
$z(p,t)$ denote the total (fractional) number of servers at time $t$ at the leaves of the subtree $T(p)$.
The cost of the $k$-server solution is
$$ \sum_t \sum_p W(p)  |z(p,t) - z(p,t-1)|. $$

\begin{lemma}
\label{l:easypart}
The movement cost incurred by the fractional $k$-server solution is no more than the total movement cost incurred by the convex combination of the allocation instances running on internal nodes of $T$.
\end{lemma}
\begin{proof}
First, we claim that $z(p,t) = \sum_{s \in \Lambda^t_p}
\lambda_s^t \kappa^t_s(t)$. This follows from the consistency
relation (\ref{eq:lambda_consist}) we maintain, and our procedure
for generating $\Lambda^t_p$ from $\Lambda^{t-1}_p$. In
particular, $\kappa^t_s(t)$ is the number of servers available for
the fractional allocation instance $s$ running at $p$. Since the
solution produced by the fractional allocation algorithm on this
instance satisfies $\sum_{i} \sum_{j} j \cdot x^t_{i,j,s} =
\kappa^t_s(t)$,\footnote{Note that when we designed the fractional
allocation algorithm in Section \ref{s:frac}, we allowed it to
deploy less servers than the current quota. As a result, when
applying this algorithm here we could sometimes have $\sum_{i}
\sum_{j} j \cdot x^t_{i,j,s} < \kappa^t_s(t)$. To see that our
analysis is still valid in this case, it suffices to consider
a modified version of the tree $T$. In this version, each non-leaf
node $p$ of $T$ would have a dummy leaf $i_p$ (that will never be
requested in our input  sequence) added as its child and set the
length $w(p,i_p)$ of the corresponding edge to $0$. Now, we would
just make each instance of the fractional allocation run at each
such $p$ deposit any unused quota of servers at the leaf $i_p$.
Note that as $w(p,i_p)=0$, this depositing would not incur any
additional movement cost and that the modified tree would still be
a weighted $\sigma$-HST.} this implies that
$$ \sum_{s \in \Lambda_{p}^t} \lambda_s^t \kappa^t_s(t) =  \sum_{s \in \Lambda_{p}^t}  \lambda^t_s \sum_i \sum_j j \cdot  x^t_{i,j,s}
  = \sum_i \sum _j \sum_{s' \in \Lambda_{p_i}^t | \kappa^t_{s'}(t) = j} j\lambda_{s'}^t =
   \sum_i  \sum_{s' \in \Lambda_{p_i}^t} \lambda_{s'}^t \kappa^t_{s'}(t).$$
The second equality above follows from  (\ref{eq:lambda_consist}).
Applying this iteratively, and noting that $z(q,t)$ for a leaf $q$ is simply $\sum_{s'\in \Lambda_q^t} \lambda_{s'}^t \kappa^t_{s'}(t)$, it follows that
$z(p,t) = \sum_{s \in \Lambda^t_p} \lambda_s^t \kappa^t_s(t)$.

Suppose that $\delta=z(p,t)-z(p,t-1)>0$ server units are
removed from the subtree $T(p)$ (the case when $\delta<0$ is
analogous). Let $p'$ be the parent of $p$ (note that $p \neq r$, since $z(r,t)=k$ for all $t$). As $z(p,t) = \sum_{s
\in \Lambda^t_{p'}} \lambda^t_s \sum_j j \cdot x^t_{p,j,s}$, it
follows that the allocation algorithm running on the instances
$s\in \Lambda^t_{p'}$ will incur a movement cost of at least
\begin{eqnarray*}
W(p) \sum_{s \in \Lambda^t_{p'}} \lambda^t_s \sum_{j=1}^{k} \left|\sum_{\ell<j}(x_{p,\ell,s}^t - x_{p,\ell,s}^{t-1})\right| & \geq & W(p) \cdot \left|\sum_{s \in \Lambda^t_{p'}} \lambda^t_s \sum_{j=1}^{k} \sum_{\ell<j}\left(x_{p,\ell,s}^t - x_{p,\ell,s}^{t-1}\right)\right|\\
& = &
W(p) \cdot
 \left|\sum_{s \in \Lambda^t_{p'}} \lambda^t_s  \sum_{j=0}^{k} (k-j)(x^t_{p,j,s} - x^{t-1}_{p,j,s})\right| \\
 & = &
 W(p) \cdot
  \left|\sum_{s \in \Lambda^t_{p'}} \lambda^t_s  \sum_{j=0}^{k} (-j)(x^t_{p,j,s} - x^{t-1}_{p,j,s})\right| \\
 &  = & W(p) \cdot |z(p,t)-z(p,t-1)|,
 \end{eqnarray*}
 where we used the fact that $k \cdot \sum_{j=0}^k x_{p,j,s}^t=k=k \cdot \sum_{j=0}^k x_{p,j,s}^{t-1}$.
\end{proof}

Given Lemma \ref{l:feasibility} and \ref{l:easypart} above, it suffices to consider
the total movement cost incurred by the allocations instances running on the tree $T$ and compare it with the optimum $k$-server cost.
This will be our goal in the following. We begin by defining a notion of optimum $k$-server cost on a weighted $\sigma$-HST $T$ when $k$  varies over time.

\begin{defn}
Let $T(p)$ be the subtree rooted at $p$, and let $\kappa$ be a quota pattern. We define  $\optcost(p,\kappa,t)$ as the optimum cost of serving the request sequence $\rho \cap T(p)$ until time step $t$ subject to the constraint that $\kappa(t')$ servers
are available at each time $t'$, for $1 \leq t' \leq t$.
\end{defn}
We should be precise about the meaning of a $k$-server solution on $T$ in the case $\kappa(t)$ can vary.
First, at any time $t'$ there should be one server unit at the requested location $\rho(t')$.
The cost of the solution is the total movement cost of the servers.
The servers are always located on the leaves of $T$.
At time $t$, when the number of servers changes from $\kappa(t-1)$ to $\kappa(t)$, we will require that $\kappa(t)-\kappa(t-1)$ servers enter (or leave, if $\kappa(t)<\kappa(t-1)$) from the root of $T$.

For a vector $\kappa$,  let us define  $g(\kappa,t) = \sum_{t'=1}^t  |\kappa(t') -\kappa(t'-1)|$.
The following is a simple but very useful fact about  $\optcost$, that we will need.

\begin{lemma} Let $p$ be an internal node in $T$ with children $p_1,\ldots,p_d$. For any $k$-server request sequence $\rho$ on the leaves of $T$ and any quota pattern vector $\kappa$, the following recurrence holds.
\begin{equation}
\label{opt-rel}
 \optcost(p,\kappa,t)
 =  \min_{\kappa_1,\ldots,\kappa_d : \sum_{i=1}^d \kappa_i = \kappa} \left( \sum_{i=1}^d \optcost (p_i,\kappa_i,t)
 +   w(p,i) \sum_{i=1}^d g(\kappa_i,t) \right).
 \end{equation}
Here, in the base case in which $p$ is a leaf,  define $\optcost(p,\kappa,t) = \infty$ if there is some time $t'\leq t$ such that $\rho(t')=p$ and $\kappa(t')=0$. Otherwise, if $\kappa(t')\geq 1$, whenever $\rho(t')=p$, define $\optcost(p,\kappa,t)=0$.
\end{lemma}
\begin{proof} The condition $\sum_{i=1}^d \kappa_i = \kappa$ ensures consistency between the number of servers in $T(p)$ and its  subtrees $T(p_i)$. The term $\optcost (p_i,\kappa_i,t)$ measures the cost of serving the requests within $T(p_i)$ and  $g(\kappa_i,t)$ measures the cost of servers leaving or entering subtree $T(p_i)$.
\end{proof}

Next, we need the following key lemma that relates $\optcost(p,\kappa,t)$ to our procedure for generating hit costs at node $p$.

\begin{lemma}
\label{l:optcost_to_hc}
Let $p$ be a non-root node of $T$. Given a quota pattern $\kappa$ for $T(p)$, let  $\optcost(p,\kappa,t)$ be as defined above. Then,
\begin{equation}\label{eq:hit-opt}
\left|\optcost(p,\kappa,t) - \sum_{t'=1}^t h^{t'}_{p}(\kappa(t')) \right|\leq 2 \cdot \frac{1}{\sigma-1} \cdot W(p) \cdot g(\kappa,t),
\end{equation}
where $h^{t'}_p(j) = \optcost(p,j \cdot 1,t') - \optcost(p,j\cdot 1,t'-1)$ denotes the incremental cost of the optimal $k$-server solution for $T(p)$ with exactly $j$ servers.
\end{lemma}
This lemma (for the case of HSTs) is implicit in the work of Cot\'{e} et al.~\cite{CMP08}.
For completeness, and since we need the extension to weighted $\sigma$-HSTs, we give a proof of Lemma \ref{l:optcost_to_hc} in Appendix \ref{app:mey}.

We are now ready to prove the following theorem.
\begin{thm}
\label{th:rec}
Let $T$ be a weighted $\sigma$-HST with $\sigma>9$, depth $\ell$, and diameter $\Delta$. Let $\rho$ be a $k$-server request sequence on $T$, and $\kappa$ be the quota pattern.
Consider the total movement cost incurred by the convex combination of the allocation instances running on nodes of $T$ (based on the algorithm described in Section \ref{s:alloctokserver}). This cost is no more than $$\beta_{\ell} (\optcost(r,\kappa,\infty)+ \Delta \cdot g(\kappa,\infty)), $$
 where  $\beta_\ell$ satisfies the recurrence  $ \beta_{\ell} = \gamma \beta_{\ell-1} + O\left(\log (k/\eps)\right)$, and $\beta_0=1$.
 Here, $\epsilon$ is any constant for which the fractional allocation algorithm is $(1+\eps,O(\log (k/\eps)))$-competitive, and
$$ \gamma = \left(1+\varepsilon\right)\left(1 + \frac{3}{\sigma}\right)+  O\left(\frac{1}{\sigma} \log (k/\eps)\right).$$
\end{thm}
We first show how Theorem \ref{th:rec} implies Theorem \ref{thm:alloc_to_kserver}.
The recurrence $
\beta_{\ell}  \leq \gamma \beta_{\ell-1} + O(\log (k/\eps))$ in Theorem \ref{th:rec}, together with  $\beta_0=1$, implies that
\[
 \beta_{\ell} =  O(\log (k/\eps)) \left(\frac{\gamma^{l+1} - 1}{\gamma-1} \right).
\]
Choosing $\eps = 1/(4\ell)$, and provided  $\sigma = \Omega(\eps^{-1} \log (k/\eps))=\Theta(\ell \log (k\ell))$, we get that $\gamma \leq  (1 + \frac{1}{2\ell})$, and hence $$\beta_{\ell} = O( \ell \log (k/\eps) )=O(\ell \log (k\ell)).$$
As $g(k\cdot \vec{1},\infty)=0$, this implies an $O(\ell \log (k\ell) )$ guarantee for a weighted $\sigma$-HST of depth $\ell$, provided $\sigma= \Omega(\ell \log (k\ell))$.
We now prove Theorem \ref{th:rec}.
\begin{proof}(Theorem \ref{th:rec}):
We prove by induction on the depth of the tree.
\paragraph{Base case:}
the theorem is clearly true for $\ell =0$ (i.e. a single point space).
\paragraph{Inductive step:}
suppose the theorem is true for weighted $\sigma$-HSTs of depth $\ell-1$, and let $T$ be a weighted $\sigma$-HST of depth $\ell$, rooted at $r$.
Let $w_i$ be the distance to the $i$-th child of $r$, and let $w = \max_{i} w_i$.
Given $\kappa$, consider some optimal solution for $T$ that achieves value  $\optcost(r,\kappa,\infty)$.
We also denote  the total cost $\optcost(r,\kappa,\infty)$ by $\optcost(r,\kappa)$, and $g(\kappa,\infty)$ by $g(\kappa)$.
Let $\kappa^*_i$ be
optimal vectors for the children $p_i$ of $r$ corresponding to this solution.
Since $\kappa^*_i$ determines $\optcost(r,\kappa)$,  by  (\ref{opt-rel})  we have
$$ \optcost(r,\kappa)   =   \sum_i \left( \optcost(p_i,\kappa^*_i) + w_i \cdot g(\kappa^*_i) \right). $$
By (\ref{eq:hit-opt}), for each child $i$, $ \optcost(p_i,\kappa^*_i) \geq  \hc(i,\kappa^*_i) - 2 w_i \cdot g(\kappa^*_i) /(\sigma-1)$,
 where we denote  $\hc(i,\kappa') = \sum_{t} h^{t}_{i}(\kappa'(t))$. Thus,
$$ \optcost(r,\kappa)  \geq  \sum_i \left(\hc(i,\kappa^*_i) + w_i \cdot \left(1-\frac{2}{\sigma-1}\right) g(\kappa^*_i) \right). $$
Multiplying throughout by $(\sigma-1)/(\sigma-3)$, which is at most $1+3/\sigma$  (as $\sigma\geq 9$), implies
\begin{equation}
\label{eqhc1}
 \sum_i (\hc(i,\kappa_i^*) + w_i \cdot g(\kappa_i^*))  \leq  \frac{\sigma-1}{\sigma-3}\cdot \optcost(r,\kappa) \leq   \left(1+\frac{3}{\sigma} \right)\cdot \optcost(r,\kappa).
 \end{equation}
Consider the fractional allocation instance $A$ running at $r$ in our algorithm and let $\{x_{i,j}^t\}_{i,j}$ be its solution at time step $t$.
(Since there is only a single quota pattern $\kappa(r)$ at the root, we assume that there is only one instance running, and keep it notationally convenient.)
By Theorem \ref{thm:frac_alloc_variable} and by \eqref{eqhc1}, the hit cost incurred by  $A$ satisfies
\begin{eqnarray}
 \sum_t \sum_{i=1}^{d} \sum_{j=0}^{k} x_{i,j}^t h_{i}^t(j)
& \leq & (1+\varepsilon) \left(\sum_{i=1}^{d}\hc(i,\kappa^*_i) + \sum_{i=1}^{d} w_i  \cdot g(\kappa^*_i) +  w \cdot g(\kappa)\right) \nonumber \\
&  \leq & (1+\varepsilon)\left(1 +
\frac{3}{\sigma}\right) \optcost(r,\kappa) + (1+\eps) w \cdot g(\kappa)\label{eqhc2},
\end{eqnarray}
and the movement cost satisfies
\begin{eqnarray}
\sum_t \sum_{i=1}^{d}  w_i \sum_{j} |y_{i,j}^t-y_{i,j}^{t-1}|
& \leq & O(\log (k/\eps)) \left(\sum_{i=1}^{d} \hc(i,\kappa^*_i) + \sum_{i=1}^{d} w_i \cdot g(\kappa^*_i) + w \cdot g(\kappa) \right) \nonumber \\
& \leq &   O(\log (k/\eps)) (\optcost(r,\kappa)+w \cdot g(\kappa)).
 \label{eqhc3}
\end{eqnarray}
Now, recall that each of the children $p_1,\ldots,p_d$ of $r$ is
running a convex combination on allocation instances
$\Lambda_i^t$, the quota pattern of which is determined by
$\{x^t_{i,j}\}_{i,j}$. So, the hit costs and movement costs of $A$
(i.e. left hand sides of (\ref{eqhc2}) and (\ref{eqhc3})) can be
expressed alternately as follows. Since the quota patterns at
$p_i$ maintain the invariant \eqref{eq:lambda_consist} throughout
the algorithm, the total hit cost accumulated by $A$ can be
expressed as

\begin{eqnarray}\label{eq:delta_to_hit}
\sum_t \sum_{i=1}^{d} \sum_{j=0}^{k} x_{i,j}^t h_{i}^t(j) &=& \sum_t \sum_{i=1}^{d} \sum_{j=0}^k  \sum_{s\in \Lambda_i^t, \kappa_s^t(t)=j} \lambda_s^t h_i^t(j) \nonumber \\
& = &   \sum_{i=1}^{d} \sum_{s\in \Lambda^\infty_i} \lambda^\infty_s \hc(i,\kappa^\infty_s).
\label{eq:star1}
\end{eqnarray}
Henceforth, we use $\Lambda_i$ and $\kappa_s$ to denote $\Lambda_i^\infty$ and $\kappa_s^\infty$.
By our cost preserving procedure for updating
$\kappa^t_s$ when $x^{t-1}_{i,j}$ changes to $x^{t}_{i,j}$,  the
movement cost incurred by  $A$ can be expressed as

\begin{eqnarray}\label{eq:delta_to_move}
\sum_t \sum_{i=1}^{d} w_i \sum_{j} |y_{i,j}^t-y_{i,j}^{t-1}|) &=&  \sum_{i=1}^{d} w_i \sum_t \sum_{s\in \Lambda_i^t} \lambda_s^t |\kappa_s^t(t)-\kappa_s^t(t-1)| \nonumber \\
 & = &   \sum_{i=1}^{d} w_i \sum_{s\in \Lambda_i} \lambda_s g(\kappa_s). \label{eq:star2}
\end{eqnarray}
Let us now consider the overall movement cost incurred by the convex combination of the allocation instances. This is equal to the movement cost for $A$ (running at the root) plus the sum movement costs incurred within $p_1,\ldots,p_d$.
By the induction hypothesis, the movement cost for each of these recursive algorithms that are run on subtrees $T(p_i)$ with quota pattern $\kappa_s$ is at most
$$\beta_{\ell-1} \left(\optcost(p_i,\kappa_s)+ \frac{w_i}{\sigma-1} \cdot g(\kappa_i)\right).$$
Thus, the total recursive cost is at most
\begin{eqnarray}
 & & \sum_{i=1}^{d} \sum_{s\in \Lambda_i} \lambda_s \beta_{\ell-1}\left(\optcost(p_i,\kappa_s)  + \frac{w_i}{\sigma-1} \cdot  g(\kappa_s)\right) \nonumber \\
 & \leq & \sum_{i=1}^{d} \sum_{s\in \Lambda_i} \lambda_s  \beta_{\ell-1} \left(\hc(i,\kappa_s)   + \frac{3 w_i}{\sigma-1} \cdot g(\kappa_s)\right) \label{rec:cost2} \\
 & \leq & (1+\eps) \beta_{\ell-1} \left(\left(1+\frac{3}{\sigma}\right) \optcost(r,\kappa) + w \cdot g(\kappa)\right)
 +  \sum_{i=1}^{d} \sum_{s\in \Lambda_i} \lambda_s  \beta_{\ell-1}  \frac{3 w_i}{\sigma-1} \cdot g(\kappa_s) \label{rec:cost3}.
 \end{eqnarray}
Here, \eqref{rec:cost2} follows as $\optcost(p_i,\kappa_s) \leq
\hc(i,\kappa_s) + 2w_i/(\sigma-1)$ by (\ref{eq:hit-opt}), and
\eqref{rec:cost3} follows from (\ref{eq:delta_to_hit}) and
(\ref{eqhc2}). Now, the total cost of movement of servers across
the subtrees $p_1,\ldots,p_d$ is
\begin{equation}
\label{eq:cost2}
  \sum_{i=1}^{d} w_i \sum_{s\in \Lambda_i} \lambda_s g(\kappa_s).
\end{equation}
Adding up the costs from (\ref{rec:cost3}) and (\ref{eq:cost2}),
the total cost incurred by the algorithm is at most
\begin{equation}
\label{eq:tot}
 (1+\eps) \beta_{\ell-1} \left(\left(1+\frac{3}{\sigma}\right) \optcost(r,\kappa) + w \cdot g(\kappa)\right)
+\sum_{i=1}^{d} \sum_{s\in \Lambda_i} \lambda_s  \left( 1+ \frac{3 \beta_{\ell-1}}{\sigma-1}\right) w_i \cdot g(\kappa_s).
\end{equation}
By (\ref{eq:delta_to_move}) and (\ref{eqhc3}) we have
\begin{equation}
\label{eq:bd2}
 \sum_{i=1}^{d}  w_i \sum_{s\in \Lambda_i} \lambda_s g(\kappa_s) \leq O(\log (k/\eps)) \left( \optcost(r,\kappa) + w \cdot g(\kappa)\right).
\end{equation}
Plugging (\ref{eq:bd2}) into (\ref{eq:tot}) implies that the total
cost is at most
$$
\Bigl(\beta_{\ell-1} \cdot \gamma+ O\left(\log (k/\eps)\right)
\Bigr) \Big(\optcost(r,\kappa)+ w \cdot g(\kappa)\Bigr),
$$
where $$\gamma = (1+\eps)\left(1+\frac{3}{\sigma}\right) + O\left(
\frac{\log (k/\eps)}{\sigma} \right).$$ Thus, the claimed result
follows.
\end{proof}

\section{Weighted HSTs and Online Rounding }
\label{s:combining} In this section, we show how one can embed a
$\sigma$-HST into a small depth weighted $\sigma$-HST with
constant distortion, i.e., we prove Theorem
\ref{thm:transform_ov}. Also, we present an online rounding
procedure for the fractional $k$-server problem on an
(un-weighted) $\sigma$-HST, that is, we establish Theorem
\ref{thm:rounding_ov}.

\subsection{Embedding $\sigma$-HSTs into Weighted $\sigma$-HSTs:}
\label{s:whst}

{
\renewcommand{\thethm}{\ref{thm:transform_ov}}
\begin{thm} Let $T$ be a $\sigma$-HST with $n$ leaves, but possibly arbitrary depth. Then, $T$ can be
transformed into a weighted $\sigma$-HST $\widetilde{T}$ such
that: the leaves of $\widetilde{T}$ and $T$ are identical,
$\widetilde{T}$ has depth $O(\log n)$,  and any leaf to leaf
distance in $T$ is distorted in $\widetilde{T}$ by a factor of at
most $2 \sigma/(\sigma-1)$.
\end{thm}
\addtocounter{thm}{-1}
}
\begin{proof}
For a given rooted tree $T'$, we say that it is {\em balanced} if:
(1) there is no child $p$ of the root such that the subtree
$T'(p)$ (rooted at $p$) contains more than half of the nodes of
$T'$, and (2) each subtree $T'(p')$, rooted at a child $p'$ of the
root, is balanced as well. It is easy to see that if a balanced
$T'$ has $n$ leaves then its depth is $O(\log n)$.

We now present a procedure that contracts some of the edges of $T$
and yields a weighted $\sigma$-HST $T'$ such that: (a) $T'$ is
balanced, and (b) for any leaf-to-leaf path in $T$, at least one
(out of two) of the longest edges on this path has not been
contracted in $T'$.

The procedure works as follows. Let $r$ be the root of $T$ and
$p_1,\ldots, p_d$ be its children. We first recursively
transform each of the trees $T(p_i)$ rooted at $p_i$. Next, we
check if there is a child $p_i$ such that $T(p_i)$ contains more
than half of the nodes of $T$. (Note that there can be at most one
such child.) If not, then $T$ (with modified $T(p_1),\ldots ,
T(p_d)$) is the desired $T'$. Otherwise, $T'$ is the tree $T$ with
the edge that connects $p_i$ to $r$ contracted.

It is easy to see that the tree $T'$ obtained by the above
procedure is indeed balanced and also the lengths of the edges on
any root-to-leaf part decrease at rate of at least $\sigma$. Thus,
$T'$ is a weighted $\sigma$-HST and (a) holds. Now, to prove (b),
let us inductively assume that it holds for all the transformed
subtrees $T(p_i)$. We prove it for the transformed $T$. Consider a
leaf-to-leaf path in $T$. If the path is contained entirely in one
of the subtrees $T(p_i)$, then we are done by our inductive
assumption. Otherwise, the path contains two edges incident to the
root $r$. As $T$ is a $\sigma$-HST, these two edges must be the
longest ones on this path. Thus, as the procedure could contract
only one of them, (b) follows as well.

Now, clearly, taking $\widetilde{T}$ to be $T'$ satisfies the
first two desired properties of $\widetilde{T}$, as stated in the
theorem. To prove that the last one holds too, we note that the
length of any leaf-to-leaf path can only decrease in
$\widetilde{T}$ (as we only contract edges). However, as we retain
at least one of the longest edges, we have that the worst case
distortion it incurs is at most:
$$
\frac{2}{\sigma^{\ell}} \cdot  \sum_{j=0}^{\ell} \sigma^{j} \leq 2
\cdot \sum_{j=0}^{\ell} \frac{1}{\sigma^j} \leq
\frac{2\sigma}{\sigma-1},
$$
where $\sigma^{\ell}$ is the length of the longest edge on the
path. The theorem thus follows.
\end{proof}

\subsection{Rounding the Fractional $k$-server Solution Online}
\label{s:rounding}

We now show how to obtain an online randomized (integral)
$k$-server algorithm from a {\em fractional} $k$-server algorithm,
when the underlying metric corresponds to a $\sigma$-HST $T$. The
competitiveness of the obtained algorithm will only be an $O(1)$
factor worse than the competitiveness of the fractional algorithm,
provided $\sigma>5$. The rounding procedure builds on ideas in
\cite{BBK} which were developed in the context of the finely
competitive paging problem, and extends those ideas from a uniform
metric to HSTs.

Let $1,\ldots,n$ denote the leaves of the $\sigma$-HST $T$. Recall
that the fractional solution to the $k$-server problem specifies
at each time $t$ the probability $x^t_i$ of having a server at
leaf $i$. The variables $x^t_i$ satisfy $\sum_i
x^t_i=k$.\footnote{Note that, in principle, the definition of a
fractional solution allows us to have $\sum_i x^t_i$ to be
strictly smaller than $k$. However, as the starting configuration
$\ox^0$ has exactly $k$ servers, it is easy to modify our
fractional solution so that it always has exactly $k$ servers,
while not increasing its movement cost.} When these probabilities
change at each time step $t$, the movement cost paid by the
fractional algorithm is equal to the earthmover distance between
the distribution $\ox^{t-1}$ and $\ox^{t}$. In contrast, an
execution of a randomized algorithm can be viewed as an evolution
of a distribution over $k$-tuples of leaves. (There is no point in
having more than one server at a leaf.) To make this more
precise, let us define a {\em configuration} to be a subset $C$ of
$\{1,\ldots, n\}$ of size exactly $k$. The state $S^t$ at a given
time $t$ of a randomized $k$-server algorithm  is specified by a
probability distribution $\mu_{S^t}$ on the configurations, where
$\mu_S(C)$ denotes the probability mass of configuration $C$ in
state $S$.

Now, we say that a state $S$ is {\em consistent} with  a fractional state $\ox$ if,
\begin{eqnarray}
\mbox{(consistency)} & \mbox{for each $i \in [n]$} \qquad \sum_{C:
i \in C} \mu_S(C) = x_i,
\end{eqnarray}
i.e., if the marginal probabilities of the state $S$ coincide with
$\ox$.

We therefore see that in order to round a fractional algorithm to
a randomized algorithm, we need to devise a way of maintaining (in
an online manner) a sequence of states $S^0,S^1,\ldots$ that are
always consistent with the corresponding states
$\ox^0,\ox^1,\ldots$ of the fractional algorithm, and such that
the cost of the maintenance is within an $O(1)$ factor of the
movement cost of the fractional algorithm. More precisely, our
goal is to establish the following result.

\begin{thm}
\label{thm:simulate} Let $T$ be a $\sigma$-HST with $n$ leaves,
$\sigma>5$, and let $\ox^0,\ox^1,\ldots$ be a sequence of states
of a fractional $k$-server algorithm. There is an online procedure
that maintains a sequence of randomized $k$-server states
$S^0,S^1,\ldots $ with the following properties:
\begin{enumerate}
\item At any time $t$, the state $S^t$ is consistent with the
fractional state $\ox^t$. \item If the fractional state changes
from $\ox^{t-1}$ to $\ox^{t}$ at time $t$, incurring a movement
cost of $c_t$, then the state $S^{t-1}$ can be modified to a state
$S^{t}$ while incurring a cost of $O(c_t)$.
\end{enumerate}
\end{thm}
The key idea in bounding the maintenance cost of our rounding in
Theorem \ref{thm:rounding_ov} is to require that the states $S^t$
that we produce are not only consistent with $\ox^t$, but also
each configuration in the support of $S^t$ does not deviate much
from the fractional state $x^t$. To this end, we introduce the
following additional property for $k$-server states.

For a node $p$ of $T$ and a fractional state $\ox$, let $x_p =
\sum_{i \in T(p)} x_i$ be the fractional amount of servers that
$\ox$ has on the leaves of the subtree $T(p)$ rooted at $p$. Also,
let $n_p(C) = C \cap T(p)$ denote the number of servers in
configuration $C$ on leaves of $T(p)$. We say that a configuration
$C$ is {\em balanced} with respect to $\ox$ if $n_p(C) \in
\{\lfloor x_p \rfloor,\lceil x_p \rceil\}$ for every node $p$.
Now, we say that a $k$-server state $S$ is balanced with respect
to $\ox$ if every configuration in its support (i.e., with
non-zero probability mass) is balanced with respect to $\ox$. That
is, for all $p$ and  $C$ for which $\mu_S(C)>0$,
\begin{eqnarray}
 \lfloor x_p\rfloor \leq  n_p(C)  \leq \lceil x_p \rceil.
\end{eqnarray}
Now, our approach to making the states we work with balanced is
facilitated by the following definition. Let $\ox$ be a fractional
state and let $S$ be some $k$-server state consistent with $\ox$
(but $S$ might be not balanced with respect to $\ox$). We define
the {\em balance gap}, $G(S,\ox)$ of $S$, (with respect to $\ox$)
to be:
\begin{eqnarray} G(S,\ox) =  \sum_p  W(p)   \sum_{C \in S}  \mu_S(C)  \min   \left(| n_p(C) - \lfloor x_p \rfloor |, |n_p(C) - \lceil x_p \rceil|\right).
\end{eqnarray}
Here, $W(p)$ denotes the length of the edge from  $p$ to its
parent. Clearly, when $S$ is balanced with respect to $\ox$, its
balance gap is zero. Intuitively, the balance gap measures the
distance of the state $S$ from being balanced. This intuition is
made concrete by the following lemma.

\begin{lemma}
\label{l:sim} Let $\ox$ be a fractional state and let $S$ be a
$k$-server state on the leaves of a $\sigma$-HST $T$, with
$\sigma>5$, which is consistent with $\ox$ (but not necessarily
balanced with respect to it). Then, $S$ can be converted to
another state $S'$ which is both consistent and balanced with
respect to $\ox$, while incurring a cost of  $O(G(S,\ox))$.
\end{lemma}
We will prove Lemma \ref{l:sim} later. First we show show how
Theorem \ref{thm:simulate} follows from it.

\begin{proof}[Proof of Theorem \ref{thm:simulate}]
Consider a fractional state $\ox$ that changes to some other
fractional state $\ox'$, and let $S$ be a $k$-server state which
is both consistent and balanced with respect to $\ox$. As $S^0$
(the state initially at time $t=0$) is consistent and balanced
with respect to $\ox^0$, it is easy to see that to establish the
theorem, it suffices to show that for any $\ox$, $\ox'$, and $S$
as above, there is a $k$-server state $S'$ which is consistent and
balanced with respect to $\ox'$, and such that the cost of
changing the state $S$ to state $S'$ is within $O(1)$ factor of
the cost of changing the state $\ox$ to $\ox'$. Furthermore, it is
enough to restrict oneself to the case in which $\ox'$ is obtained
from $\ox$ by applying an elementary move, i.e. $x_i$, for some
leaf $i$, is increased by $\delta$ and $x_{i'}$, for some other
leaf $i'$, is decreased by $\delta$, where $\delta$ can be chosen
to be infinitesimally small.

In light of this, we can focus on proving the existence of such
$S'$. To this end, let $p$ denote the least common ancestor of $i$
and $i'$. Note that in this case the fractional cost of changing
$\ox$ to $\ox'$ is at least $2 \delta w(p)$, where
$w(p)=W(p)/\sigma$ is the length of the edges from $p$ to its
children.

Consider now the following transformation of the state $S$. First,
we add the leaf $i$ to a probability mass of $\delta$ on arbitrarily
chosen configurations in $S$ that do not contain $i$ already.
Next, we remove $i'$ from some probability mass $\delta$ of
configurations containing $i'$. (Note that the existence of sufficient
mass of each type of configurations follows from the consistency property of
$S$ with $\ox$.) Let $\tS$ be the resulting state.

Before proceeding, we note that as $\delta$ can be taken to be
arbitrarily small, we can restrict our discussion to the case
in which $i$ is added to a mass of $\delta$ of a particular
configuration $C$, and $i'$ is removed from a mass of $\delta$ of
a particular configuration $C'$.

Now, to continue the proof, we observe that $\tS$ is consistent with $\ox'$. However, the
modified configurations $C$ and $C'$ that $\tS$ contains are not legal anymore
as they do not consist of exactly $k$ leaves. Also, $\tS$ might be unbalanced
with respect to $\ox'$.

We show how we can modify $\tS$ to fix these shortcomings. To this
end, we note that as $C$ contains $i$, and it was balanced with
respect to $\ox$ prior to adding $i$, it must hold now that $n_p(C)
\geq \lfloor x_p \rfloor+1$. Similarly, for $C'$, it must hold that
$n_p(C') \leq \lfloor x_p \rfloor<\lfloor x_p \rfloor+1=n_p(C)$.
Thus, by the pigeon hole principle, there must exist a leaf $j$ of
$T(p)$ which is contained in $C$, but not in $C'$. Let us modify
$\tS$ by removing $j$ from $C$ and adding it to $C'$. Clearly,
this makes all the configurations in $\tS$ legal again, keeps $\tS$
consistent with $\ox'$, and the total movement cost corresponding
to this modification (due to deleting $i$, adding $i'$, and
swapping $j$) is at most $4 \delta w(p) \sigma/(\sigma-1)=
O(\delta w(p))$, for $\sigma>5$, which is within an $O(1)$ factor
of the cost of changing $\ox$ to $\ox'$.

Unfortunately, $\tS$ might still be unbalanced with respect to
$\ox'$. To bound the imbalance, let us first consider the case
in which $\lfloor x_q\rfloor =\lfloor x_q'\rfloor$ and $\lceil x_q
\rceil=\lceil x_q' \rceil$ for all nodes $q$. This implies that
all the configurations in $\tS$ other than the modified
configurations $C$ and $C'$ are already balanced with respect to
$\ox'$ as they were balanced with respect to $\ox$. Now, we note
that $x_{q}\neq x_q'$ only for nodes $q$ that are on the path
between $i$ and $i'$ (but excluding $p$). Similarly, $n_q(C)$ and
$n_q(C')$ could change only for nodes on the path from $p$ to $i$,
$i'$, or $j$ (but, again, excluding $p$). Therefore, as both $C$
and $C'$ were initially balanced with respect to $\ox$, we can
conclude that the total imbalance gap $G(\tS,\ox)=G(\tS,\ox')$ of
$\tS$ after our modifications is at most:

$$3 \cdot 2\delta w(p)\left(1+ \frac{1
}{\sigma} + \frac{1}{\sigma^2} + \ldots \right) =O(\delta w(p)).$$
Thus, by applying Lemma \ref{l:sim} to $\tS$, we obtain a state $S'$
that is consistent and balanced with respect to $\ox'$ and the
cost of this procedure is $O(\delta w(p))$, as desired.

Now, it remains to deal with the case in which either
$\lfloor x_q\rfloor \neq \lfloor x_q'\rfloor$, or $\lceil x_q \rceil\neq \lceil x_q' \rceil$, for some nodes $q$. To this end, we note that by taking $\delta$
to be small enough (but non-zero), we can ensure that for each $q$ for which at least one of these two inequalities holds, it must be the case that either $x_q$ or $x_q'$ is an integer. In the former case, we have that for all the configurations $C''$ in $S$ that have non-zero mass, $n_q(C'')=x_q=\lfloor x_q\rfloor =\lceil x_q \rceil$ and thus $\lfloor x_q'\rfloor \leq n_q(C'') \leq \lceil x_q'\rceil$. In the latter case, as $|x_q-x_q'|\leq \delta$ and for every relevant configuration $C''$ in $S$, $\lfloor x_q\rfloor \leq n_q(C'') \leq \lceil x_q\rceil$, the total probability mass of configurations $C''$ in $S$, such that $n_q(C'')=\lfloor x_q\rfloor < \lfloor x_q'\rfloor=x_q'$ or $n_q(C'')=\lceil x_q\rceil>\lceil x_q'\rceil = x_q'$, can be at most $\delta$.

As a result, we see that the total probability mass of
configurations in $\tS$ that are not balanced with respect to
$\ox'$ is at most $3\delta$ (the contribution of $2\delta$ comes
from the modified configurations $C$ and $C'$). Thus, by
calculating the imbalance gap similarly to what we did before, we can show that
 $G(\tS,\ox')$ is $O(\delta w(p))$,
and once again use Lemma \ref{l:sim} to obtain the desired $S'$.
This concludes the proof of the theorem.
\end{proof}

It remains to prove Lemma \ref{l:sim}.

\begin{proof}[Proof of Lemma \ref{l:sim}]
Let us call a node $p$ in our $\sigma$-HST $T$ {\em imbalanced} if  $$\sum_{C \in S}  \mu_S(C) \min   (| n_p(C) - \lfloor x_p \rfloor |, |n_p(C) - \lceil x_p \rceil|) > 0.$$

If no node is imbalanced, then clearly $G(S,x)=0$, and we are already done, so we assume that this is not the case.
Let $p$ be an imbalanced nodes which is at the highest level of $T$ (breaking ties arbitrarily). We note that $p$ cannot be the root $r$ of $T$, as each configuration has exactly $k$ servers, and $x_r=k$.

Consider now a configuration $C$ for which $\mu_S(C)>0$ and $n_p(C) \notin \{\lfloor x_p \rfloor ,\lceil x_p \rceil\}$ and let us assume that $n_p(C) < \lfloor x_p \rfloor$ (the other case can be treated similarly). As $S$ is consistent with $\ox$, $\sum_{C''} \mu_S(C'') n_p(C'') = x_p$, and so there must be some other configuration $C'$ with $\mu_S(C')>0$ such that $n_p(C') \geq \lfloor x_p \rfloor+1$. So, in particular, we have $n_{p}(C')-n_p(C) \geq 2$.

Now, let $\widetilde{p}$ denote the parent of $p$ in $T$ (recall that $p$ is not the root). As we choose $p$ to be an imbalanced node at the highest possible level, $\widetilde{p}$ must be balanced, and hence $|n_{\widetilde{p}}(C') - n_{\widetilde{p}}(C)| \leq 1$. But, since $n_p(C')- n_p(C) \geq 2$, it implies the existence of some other  child $p'$, $p' \neq p$, of $\widetilde{p}$
such that  $n_{p'}(C') < n_{p'}(C)$.

Therefore, by the pigeon hole principle, there must exist a leaf $i$ in the subtree $T(p)$ rooted at $p$ which is contained in $C'$, but not in $C$. Similarly, $C$ must contain a leaf $i'$ in $T(p')$ which is not contained in $C'$.  Let $\delta = \min (\mu_S(C),\mu_S(C'))$ (note that $\delta>0$). Consider a modification of $S$ in which we take any arbitrary probability mass $\delta$  of configurations $C$ and replace $i'$ by $i$ in them. Next, we take any arbitrary  probability mass $\delta$ of configurations $C'$ and replace $i$ by $i'$ in them.

Let us summarize the properties satisfied by $S$ after this modification. First, the state remains consistent with the fractional solution $\ox$, because the marginals of the leaves $i$ and $i'$ have not changed.
Second, since neither $n_{\widetilde{p}}(C)$ nor $n_{\widetilde{p}}(C')$ have changed, $\widetilde{p}$ remains balanced.
Moreover, the only nodes for which the imbalance could have changed are on the path from $i$ to $p$ and $i'$ to $p'$.
Third, replacing $i'$ with $i$ (in a $\delta$ measure of $C$) increases $n_p(C)$ by $1$ for these configurations, and replacing $i$ with $i'$, leaves the quantity $n_p(C')$ to be of value at least $\lfloor x_p \rfloor$. Together, this implies that the imbalance of $p$ decreases by at least $\delta$. Finally, as $n_{p'}(C)>n_{p'}(C')$ before the modification, the imbalance $\sum_{C'' \in S}  \mu_S(C'') \min   (| n_{p'}(C'') - \lfloor x_{p'} \rfloor |, |n_{p'}(C'') - \lceil x_{p'} \rceil|)$ of $p'$ can only decrease.

Now, if we analyze the change in the imbalance gap of $S$, in the worst case, the imbalance of every node from $p$ to $i$ (excluding $p$) could have increased by $2\delta$ due to the addition of $i$ in $C$ or removal of $i$ from $C'$. Similarly, the imbalance of every node from $p'$ to $i'$ (excluding $p'$) could have increased by up to $2 \delta$.
Together with the above observations, this implies that the imbalance gap of $S$ decreases by at least
$$ W(p) \delta -  4 \delta w(p) \left(1+\frac{1}{\sigma} + \frac{1}{\sigma^2} + \ldots \right) = W(p) \delta \left( \frac{\sigma-5}{\sigma-1}\right) = \Omega (W(p) \delta),$$
where the last inequality uses the fact that $\sigma>5$.

On the other hand, as both $i$ and $i'$ lie in $T(\widetilde{p})$,
the movement cost incurred in the above procedure is at most $4 \delta w(\widetilde{p})(\sigma/(\sigma-1)) =   4 \delta W(p)\frac{\sigma}{(\sigma-1)}$, which is within $O(1)$ factor of the reduction in the imbalance gap.
The lemma follows by applying the above steps repeatedly until the imbalance gap reaches zero.
\end{proof}

\bibliographystyle{plain}
\bibliography{thesis}

\begin{thebibliography}{10}

\bibitem{ACN}
Dimitris Achlioptas, Marek Chrobak, and John Noga.
\newblock Competitive analysis of randomized paging algorithms.
\newblock {\em Theoretical Computer Science}, 234(1-2):203--218, 2000.

\bibitem{BBN07}
Nikhil Bansal, Niv Buchbinder, and Joseph~(Seffi) Naor.
\newblock A primal-dual randomized algorithm for weighted paging.
\newblock In {\em FOCS'07: Proceedings of the 48th Annual IEEE Symposium on
  Foundations of Computer Science}, pages 507--517, 2007.

\bibitem{BBN10a}
Nikhil Bansal, Niv Buchbinder, and Joseph~(Seffi) Naor.
\newblock Towards the randomized {\it k}-server conjecture: A primal-dual
  approach.
\newblock In {\em SODA'10: Proceedings of the 21st Annual ACM-SIAM Symposium on
  Discrete Algorithms}, 2010.

\bibitem{BBN10b}
Nikhil Bansal, Niv Buchbinder, and Joseph~(Seffi) Naor.
\newblock Unfair metrical task systems on hsts and applications.
\newblock In {\em ICALP'10: Proceedings of the 37th International Colloquium on
  Automata, Languages and Programming}, 2010.

\bibitem{Ba96}
Yair Bartal.
\newblock Probabilistic approximations of metric spaces and its algorithmic
  applications.
\newblock In {\em FOCS'96: Proceedings of the 37th Annual IEEE Symposium on
  Foundations of Computer Science}, pages 184--193, 1996.

\bibitem{Ba98}
Yair Bartal.
\newblock On approximating arbitrary metrices by tree metrics.
\newblock In {\em STOC'98: Proceedings of the 30th Annual ACM Symposium on
  Theory of Computing}, pages 161--168, 1998.

\bibitem{BBBT}
Yair Bartal, Avrim Blum, Carl Burch, and Andrew Tomkins.
\newblock A polylog($n$)-competitive algorithm for metrical task systems.
\newblock In {\em STOC'97: Proceedings of the 29th Annual ACM Symposium on
  Theory of Computing}, pages 711--719, 1997.

\bibitem{BBM01}
Yair Bartal, B{\'e}la Bollob{\'a}s, and Manor Mendel.
\newblock A ramsy-type theorem for metric spaces and its applications for
  metrical task systems and related problems.
\newblock In {\em FOCS'01: Proceedings of the 42nd Annual IEEE Symposium on
  Foundations of Computer Science}, pages 396--405, 2001.

\bibitem{BG00}
Yair Bartal and Eddie Grove.
\newblock The harmonic {\it k}-server algorithm is competitive.
\newblock {\em Journal of the ACM}, 47(1):1--15, 2000.

\bibitem{BLMN}
Yair Bartal, Nathan Linial, Manor Mendel, and Assaf Naor.
\newblock On metric ramsey-type phenomena.
\newblock In {\em STOC'03: Proceedings of the 35th Annual ACM Symposium on
  Theory of Computing}, pages 463--472, 2003.

\bibitem{BBK}
Avrim Blum, Carl Burch, and Adam Kalai.
\newblock Finely-competitive paging.
\newblock In {\em FOCS'99: Proceedings of the 40th Annual Symposium on
  Foundations of Computer Science}, page 450, 1999.

\bibitem{BKRS92}
Avrim Blum, Howard~J. Karloff, Yuval Rabani, and Michael~E. Saks.
\newblock A decomposition theorem and bounds for randomized server problems.
\newblock In {\em FOCS'92: Proceedings of the 31st Annual IEEE Symposium on
  Foundations of Computer Science}, pages 197--207, 1992.

\bibitem{BE98}
Allan Borodin and Ran El-Yaniv.
\newblock {\em Online computation and competitive analysis}.
\newblock Cambridge University Press, 1998.

\bibitem{C+}
M.~Chrobak, H.~Karloff, T.~Payne, and S.~Vishwanathan.
\newblock New results on server problems.
\newblock {\em SIAM Journal on Discrete Mathematics}, 4(2):172--181, 1991.

\bibitem{CL91}
M.~Chrobak and L.~Larmore.
\newblock An optimal on-line algorithm for {\it k}-servers on trees.
\newblock {\em SIAM Journal on Computing}, 20(1):144--148, 1991.

\bibitem{CMP08}
A.~Cot\'{e}, A.~Meyerson, and L.~Poplawski.
\newblock Randomized an optimal on-line algorithm for {\it k}-server on
  hierarchical binary trees.
\newblock In {\em STOC'08: Proceedings of the 40th Annual ACM Symposium on
  Theory of Computing}, pages 227--234, 2008.

\bibitem{CL06}
B.~Csaba and S.~Lodha.
\newblock A randomized on-line algorithm for the {\it k}-server problem on a
  line.
\newblock {\em Random Structures and Algorithms}, 29(1):82--104, 2006.

\bibitem{FRT03}
Jittat Fakcharoenphol, Satish Rao, and Kunal Talwar.
\newblock A tight bound on approximating arbitrary metrics by tree metrics.
\newblock In {\em STOC'03: Proceedings of the 35th Annual ACM Symposium on
  Theory of Computing}, pages 448--455, 2003.

\bibitem{FRR90}
A.~Fiat, Y.~Rabani, and Y.~Ravid.
\newblock Competitive {\it k}-server algorithms.
\newblock {\em Journal of Computer and System Sciences}, 48(3):410--428, 1994.

\bibitem{F+}
Amos Fiat, Richard~M. Karp, Michael Luby, Lyle~A. McGeoch, Daniel~Dominic
  Sleator, and Neal~E. Young.
\newblock Competitive paging algorithms.
\newblock {\em Journal of Algorithms}, 12(4):685--699, 1991.

\bibitem{FM}
Amos Fiat and Manor Mendel.
\newblock Better algorithms for unfair metrical task systems and applications.
\newblock {\em SIAM Journal on Computing}, 32(6):1403--1422, 2003.

\bibitem{G91}
Edward~F. Grove.
\newblock The harmonic online {\it k}-server algorithm is competitive.
\newblock In {\em STOC'91: Proceedings of the 23rd Annual ACM Symposium on
  Theory of Computing}, pages 260--266, 1991.

\bibitem{KP95}
Elias Koutsoupias and Christos~H. Papadimitriou.
\newblock On the {\it k}-server conjecture.
\newblock {\em Journal of the ACM}, 42(5):971--983, 1995.

\bibitem{MMS90}
M.~Manasse, L.A. McGeoch, and D.~Sleator.
\newblock Competitive algorithms for server problems.
\newblock {\em Journal of Algorithms}, 11:208--230, 1990.

\bibitem{GS}
Lyle~A. McGeoch and Daniel~D. Sleator.
\newblock A strongly competitive randomized paging algorithm.
\newblock {\em Algorithmica}, 6(6):816--825, 1991.

\bibitem{Seiden}
Steven~S. Seiden.
\newblock A general decomposition theorem for the {\it k}-server problem.
\newblock In {\em ESA'01: Proceedings of the 9th Annual European Symposium on
  Algorithms}, pages 86--97, 2001.

\bibitem{ST85}
Daniel~D. Sleator and Robert~E. Tarjan.
\newblock Amortized efficiency of list update and paging rules.
\newblock {\em Communications of the ACM}, 28(2):202--208, 1985.

\end{thebibliography}

\appendix

\section{Proof of Lemma \ref{obs:nosplit}}\label{app:nosplit}

Consider first the simpler case of the fix stage. Here the
variables evolve according to

\begin{eqnarray*}\frac{d y_{i,j}^\tau}{d \tau} & = & \left\{\begin{array}{ll}\frac{1}{w_i}\left(y^\tau_{i,j}+\beta\right) & y^{\tau}_{i,j} <1\\
0 & y_{i,j}^{\tau}=1.\end{array} \right.\end{eqnarray*}
As the derivative of $y_{i,j}^\tau$ only depends on $y^\tau_{i,j}$
and is continuous, the function $y^\tau_{i,j}$ is well defined. As
the derivative is 0 when $y^{\tau}_{i,j}=1$, and non-negative
otherwise, it ensures that $y^{\tau}_{i,j}$ always stays in the
range $[0,1]$. Finally, the monotonicity property holds here, as
it holds initially when $\tau=0$, and  whenever $y_{i,j} \leq
y_{i,j'}$ for some $j<j'$,  we have that $d y_{i,j}^{\tau}/d\tau
\leq d y_{i,j'}^{\tau}/d\tau$, unless $y_{i,j'}^{\tau}=1$, in
which case monotonicity holds trivially.

We now consider the hit stage.
Recall that configuration $y_{i,j}^{\eta}$, for each $(i,j)$, evolves according to Equation \eqref{eq:y_evol_general}, which we reproduce here for convenience:

\begin{eqnarray}
\frac{dy^{\eta}_{i,j}}{d \eta} & = &
\left\{\begin{array}{ll}
0 & \mbox{if } (i,j)\notin A^{\eta}, \\
\frac{1}{w_i}\left(y^{\eta}_{i,j}+\beta\right)\cdot\left(N(\eta)-\alpha \lambda^\eta_{i,j}\right) & \mbox{otherwise}.\\
\end{array}\right.\label{eq:y_evol_general2}
\end{eqnarray}
In the above, the set $A^{\eta}$ denotes the active coordinates at time $\eta$ (cf. Definition \ref{def-active}) and the normalization factor $N(\eta)$ can be expressed as (cf. \eqref{eq:n_general})
\begin{eqnarray}
N(\eta)  &=& \left\{\begin{array}{l} 0 \qquad  \qquad  \qquad \qquad  \qquad \quad \mbox{if
}\sum_{i,j}y^\eta_{i,j} > kd- \kappa(t), \\
\frac{\sum_{(i,j)\in A^{\eta}}\frac{1}{w_i}\left(y^{\eta}_{i,j}+\beta\right)\cdot
\alpha
\lambda_{i,j}^\eta}{\sum_{(i,j)\in A^{\eta}}\frac{1}{w_i}\left(y^{\eta}_{i,j}+\beta\right)}
\qquad \mbox{otherwise (i.e. if } \sum_{i,j}y^\eta_{i,j} = kd- \kappa(t)).\\
\end{array}\right.\label{eq:n_general2}
\end{eqnarray}
First, we show that during the hit stage blocks never split. To this end, we note that when two blocks merge, their $y$-values are identical, and since we also modify $\lambda^{\eta}$ to be
identical for these blocks, all the variables contained in the merged block evolve in the same way from that
point on, as desired. (Note that we do not assume here that the trajectory $\oy^\eta$ is well defined and unique, we just argue that any trajectory compatible with our definition of derivatives cannot split blocks.)

Now, we proceed to analyzing the properties of the evolution described by Equations \eqref{eq:y_evol_general2} and \eqref{eq:n_general2}. As a first step, let us prove the following claim that will be helpful later.

\begin{claim}\label{cl:n_nonincr}
Consider a feasible configuration $\oy^{\eta'}$, for $\eta'\in[0,1]$, a subset of coordinates $A$, and a hit cost vector $\lambda$. Define for any $\eta\geq \eta'$,
\[
N_A(\eta)=\frac{\sum_{(i,j)\in A}u_{i,j}^\eta\cdot
\alpha
\lambda_{i,j}}{\sum_{(i,j)\in A} u_{i,j}^\eta},
\]
where $u_{i,j}^\eta=\frac{1}{w_i}\left(y^{\eta}_{i,j}+\beta\right)$.
Now, if we make the configuration $\oy^{\eta'}$ evolve according to
\begin{eqnarray*}
\frac{dy^{\eta}_{i,j}}{d \eta} & = &
\left\{\begin{array}{ll}
0 & \mbox{if } (i,j)\notin A, \\
\frac{1}{w_i}\left(y^{\eta}_{i,j}+\beta\right)\cdot\left(N_A(\eta)-\alpha \lambda_{i,j}\right) & \mbox{otherwise}\\
\end{array}\right.,
\end{eqnarray*}
then we have that $N_A(\eta)$ does not increase, i.e., $N_A(\eta)\leq N_A(\eta')$ for any $\eta\geq\eta'$.
\end{claim}
\begin{proof}
Let us fix a $\eta\geq \eta'$ and denote $\eta^+=\eta+d \eta$. We will prove that $N_A(\eta^+)\leq N_A(\eta)$, which, in turn, implies our claim.
To this end, note that one can view $N_A(\eta)$ as a weighted average, over all coordinates in $A$, of the value of $\alpha \lambda_{i,j}$, where $u^{\eta}_{i,j}$ is the weight which we attribute to coordinate $(i,j)$ at time $\eta$.

Now, the key observation is that the way the $y_{i,j}$s evolve implies that
$y_{i,j}^{\eta^+}\geq y_{i,j}^{\eta}$ if $\alpha\lambda_{i,j}\leq N_A(\eta)$, and $y_{i,j}^{\eta^+}\leq y_{i,j}^{\eta}$ if $\alpha\lambda_{i,j}\geq N_A(\eta)$. So, as the weights $u^{\eta}_{i,j}$ are directly proportional to $y_{i,j}^\eta$, we can conclude that during our evolution, the weights of coordinates that have a value of $\alpha \lambda_{i,j}$ above the average value $N_A(\eta)$ (i.e. $u^{\eta^+}_{i,j}\leq u^{\eta}_{i,j}$ in this case) can only decrease, and the weights of coordinates that have a value of $\alpha \lambda_{i,j}$ which is at most the average (i.e., we have $u^{\eta^+}_{i,j}\geq u^{\eta}_{i,j}$ for such $(i,j)$) can only increase.
As a result, we can express $N_A(\eta^+)$ as
\begin{equation}\label{eq:lll3}
N_A(\eta^+)=\frac{\sum_{(i,j)\in A}u_{i,j}^{\eta^+}\cdot
\alpha
\lambda_{i,j}}{\sum_{(i,j)\in A} u_{i,j}^{\eta^+} }=\frac{\sum_{(i,j)\in A}u_{i,j}^{\eta}\cdot
\alpha
\lambda_{i,j} + \sum_{(i,j)\in A} \alpha
\lambda_{i,j} \Delta_{i,j} }{\sum_{(i,j)\in A} u_{i,j}^{\eta} + \sum_{(i,j)\in A} \Delta_{i,j}},
\end{equation}
where $\Delta_{i,j}=u^{\eta^+}_{i,j}-u^{\eta}_{i,j}$ and, by our discussion above, we have that $\Delta_{i,j}\geq 0$ if $\alpha \lambda_{i,j}\leq N_A(\eta)$, and $\Delta_{i,j}\leq 0$ otherwise.

Now, if the right hand side of \eqref{eq:lll3} is at most $N_A(\eta)$, then we are done. Otherwise, we must have that
\begin{equation}\label{eq:ll_n}
\frac{\sum_{(i,j)\in A}u_{i,j}^{\eta}\cdot
\alpha
\lambda_{i,j}}{\sum_{(i,j)\in A} u_{i,j}^{\eta}}>N_A(\eta),
\end{equation}
as it is easy to check that if $\frac{a+d\cdot t}{b+t}>c$, then also $\frac{a}{b}>c$, as long as $t$ and $d$ are such that $t\geq 0$ if $d\leq c$, and $t\leq 0$ otherwise.

But, the left hand side of \eqref{eq:ll_n} is by definition equal to $N_A(\eta)$. So, the obtained contradiction implies that indeed $N_A(\eta^+)\leq N_A(\eta)$, and the claim follows.
\end{proof}

Now, observe that the set $A^\eta$, as well as $N(\eta)$, depend only on the state $\oy^{\eta}$ and the hit cost vector $\lambda^{\eta}$. As a result, both $A^{\eta}$ and $N(\eta)$ can, in principle, vary drastically between points of time. However, as we show in the following claim, ``locally'' they tend to behave in a regular manner.

\begin{claim}\label{cl:nbd}
For any $\eta'\in [0,1)$ and feasible configuration $\oy^{\eta'}$, there exists an $\eta''$, $\eta'<\eta''\leq 1$, such that:
\begin{enumerate}[(a)]
\item $A^{\eta}=A^{\eta'}$ for each $\eta\in [\eta',\eta'')$, i.e., the set $A^{\eta}$ of active coordinates does not change for $\eta<\eta''$;\label{cond:newproc1}
\item $\lambda^{\eta}=\lambda^{\eta'}$ for any $\eta\in [\eta',\eta'')$, i.e., there are no block merges until time $\eta''$;\label{cond:newproc2}
\item the sum $\sum_{i,j} y_{i,j}^\eta$ is bigger than $kd-\kappa(t)$ for all $\eta<\eta''$, unless it was already equal to $kd-\kappa(t)$  at time $\eta'$;\label{cond:newproc3}
\item the configuration $\oy^{\eta''}$ is well defined and feasible for $\eta\in [\eta',\eta'']$ (note that this interval contains $\eta''$).\label{cond:newproc4}
\end{enumerate}
Furthermore, one can assume $\eta''$ is maximal, that is, unless $\eta''$ is already equal to $1$, there exists no larger value of $\eta''$ for which the above holds.
\end{claim}

For future reference, we call $\eta''$, defined as in the above claim, the {\em horizon} of $\oy^{\eta'}$.

\begin{proof}[Proof of Claim \ref{cl:nbd}]
Let us define a new process for the evolution of $y_{i,j}^\eta$ in a suitably small neighborhood of
$\eta'$ as follows:

\begin{eqnarray}\frac{dy^{\eta}_{i,j}}{d \eta} = & = & \left\{\begin{array}{ll}\frac{1}{w_i}\left(y^{\eta}_{i,j}+\beta\right)\cdot\left(\tilde{N}(\eta)-\alpha
\lambda^{\eta'}_{i,j}\right) & \mbox{if } (i,j) \in
A^{\eta'}\\
0 & \mbox{otherwise}\end{array} \right.\label{eq:evol_new}
\end{eqnarray}
where
\begin{eqnarray}
\tilde{N}(\eta)  &=& \left\{\begin{array}{l} 0 \qquad  \qquad  \qquad \qquad  \qquad \qquad \mbox{if
}\sum_{i,j}y^{\eta'}_{i,j} > kd- \kappa(t), \\
 \frac{\sum_{(i,j)\in
A^{\eta'}}\frac{1}{w_i}\left(y_{i,j}^{\eta}+\beta\right) \alpha
\lambda^{\eta'}_{i,j}}{\sum_{(i,j)\in
A^{\eta'}}\frac{1}{w_i}\left(y_{i,j}^{\eta}+\beta\right)}
\qquad \mbox{otherwise (i.e. if } \sum_{i,j}y^{\eta'}_{i,j} = kd- \kappa(t)).\\
\end{array}\right.\label{eq:n_tilde}
\end{eqnarray}
While the new process looks similar to the original one, there are
some crucial differences. First, we do not care if a variable
$y^\eta_{i,j}$ either exceeds 1 or becomes negative, i.e., the active set is not updated as $\eta$ progresses and remains the set $A^{\eta'}$. Second, the cost vector is not updated and it
remains $\lambda^{\eta'}$, i.e.,
blocks are not merged and the monotonicity requirement \eqref{eq:y_2_consistency} are ignored. Third, the value of $\tilde{N}(\eta)$ is not changed once the value of $\sum_{i,j} y_{i,j}^{\eta}$ hits the quota. If $\sum_{i,j} y_{i,j}^{\eta'}>kd - \kappa(t)$ at time $\eta'$, $\tN(\eta)$ will be always $0$. As a result, the latter sum might become less than $kd - \kappa(t)$ at some point.

Now, the key observation that makes this new process
useful to us is that it is identical to our original process, as long as the evolving configuration is still feasible, and no coordinate becomes active and changes in the original process. More precisely, the trajectories of the two processes coincide as long as in the configuration $\oy^{\eta}$ evolving with respect to this new process:
\begin{enumerate}[(1)]
\item $y_{i,j}^{\eta}$ is in $[0,1]$ for each $(i,j)\in A^{\eta'}$. (Otherwise, a coordinate $(i,j)$ violating this condition would have already become inactive in the original process.), or \label{cond:good1}
\item $\tN(\eta)-\alpha \lambda_{i,j}^{\eta'}$ is non-negative (respectively non-positive) for $(i,j)\notin A^{\eta'}$ with $y_{i,j}^{\eta'}=1$ (respectively $y_{i,j}^{\eta'}=0$). (Otherwise, a coordinate $(i,j)$ violating this condition would have already become active and would have changed in the original process.), or \label{cond:good2}
\item the monotonicity condition holds, i.e., $y_{i,j}^{\eta}\leq y_{i,j+1}^\eta$ for each $(i,j)$. (Otherwise, a block merge would have already happened in the original process.), or \label{cond:good3}
\item the server quota is obeyed, i.e., $\sum_{i,j} y_{i,j}^{\eta}\geq kd - \kappa(t)$. (Otherwise, $N(\eta)$ would have already changed so as  to guarantee that the number of servers stays at the quota.) \label{cond:good4}
\end{enumerate}
Let $\eta^\ell$ for $\ell\in \{1,2,3,4\}$ be the last $\eta\geq
\eta'$ for which the $\ell$-th condition is satisfied in the
interval $[\eta',\eta]$.  Also, let $\oeta^0$ (respectively
$\oeta^1$) be the first time $\eta\geq \eta'$ in which
$y_{i,j}^\eta=0$ (respectively, $y_{i,j}^\eta=1$) for $(i,j)\in
A^{\eta'}$ (respectively $(i,j)\notin A^{\eta'}$), and
$\tN(\eta)-\alpha \lambda_{i,j}^{\eta'}=0$. As in the new process
$\tN(\eta)$ -- and thus all the derivatives and the trajectory of
$\oy^{\eta}$ -- are continuous and bounded in the our interval of
interest\footnote{More precisely, $\tN(\eta)$ is continuous and
bounded, as long as each $y_{i,j}^{\eta}$ is bounded away from
$-\beta$. However, we are interested in only analyzing the process
 as long as all the variables stay non-negative, so for the
sake of our analysis, $\tN(\eta)$ is indeed continuous and
bounded.}, such maximal $\eta^\ell$s and $\oeta^0$, $\oeta^1$
always exists. Let us take $\eta''=\min\{\min_\ell \eta^\ell,
\oeta^0,\oeta^1,1\}$. Note that this implies that  $\oy^{\eta}$ is
well defined and feasible for any $\eta\in [\eta',\eta'']$ -- so,
condition \eqref{cond:newproc4} is satisfied. Furthermore, it is
not hard to verify that all of the conditions
\eqref{cond:newproc1}-\eqref{cond:newproc3} hold for such
$\eta''$, and that this $\eta''\leq 1$ is indeed maximal with
respect to satisfying these conditions. So, if we manage to prove
that also $\eta''>\eta'$, our claim will follow.

To this end, observe that by applying Claim \ref{cl:n_nonincr} with $A=A^{\eta'}$ and $\lambda=\lambda^{\eta'}$, we see that $\tN(\eta)$ can only decrease for $\eta\geq \eta'$. Therefore, all the derivatives \eqref{eq:evol_new} in this process can only decrease too. This implies, in particular, that if some variable stops increasing (i.e., its derivative becomes non-positive) at some point, it will never increase again.

Now, note that if $(i,j)\in A^{\eta'}$ then
\begin{itemize}
\item if $y_{i,j}^{\eta'}\in (0,1)$, then, as the derivatives are bounded, it must be the case that $y_{i,j}^{\eta}\in (0,1)$ until some time $\oeta>\eta'$;
\item if $y_{i,j}^{\eta'}=0$, then its derivative $\tN(\eta')-\alpha \lambda_{i,j}^{\eta'}$ at time $\eta'$ has to be strictly positive (otherwise, $(i,j)$ would be inactive at time $\eta'$) and thus $y_{i,j}^{\eta}$ (respectively its derivative) has to stay non-negative (respectively positive) until some time $\oeta>\eta'$ too;
\item finally, if $y_{i,j}^{\eta'}=1$, then its derivative $\tN(\eta')-\alpha \lambda_{i,j}^{\eta'}$ at time $\eta'$ has to be non-positive (as otherwise, $(i,j)\notin A^{\eta'}$). However, by our discussion above, it means that neither this derivative, nor $y_{i,j}^{\eta}$, will ever increase again.
\end{itemize}
Thus, we can infer from the above that both $\eta^1$ and $\oeta^0$ are strictly larger than $\eta'$.

 Next, let us focus on some $(i,j)\notin A^{\eta'}$. We have that
\begin{itemize}
\item if $y_{i,j}^{\eta'}=0$, then $\tN(\eta')-\alpha \lambda_{i,j}^{\eta'}$ has to be non-positive (otherwise, $(i,j)$ would be active at time $\eta'$) and as $\tN(\eta)$ never increases, $\tN(\eta')-\alpha \lambda_{i,j}^{\eta'}$ will never become positive;
\item if $y_{i,j}^{\eta'}=1$, then $\tN(\eta')-\alpha \lambda_{i,j}^{\eta'}$ at time $\eta'$ has to be strictly positive (as otherwise, $(i,j)\in A^{\eta'}$). So, this quantity has to remain positive for some time $\oeta>\eta'$.
\end{itemize}
Thus, we see that both $\eta^2$ and $\oeta^1$ are also strictly larger than $\eta'$.

Observe now that if $y_{i,j}^{\eta'}=y_{i,j+1}^{\eta'}$ for some $(i,j),(i,j+1)\in A^{\eta'}$, we must have that $\lambda_{i,j}^{\eta'}\geq \lambda_{i,j+1}^{\eta'}$. Otherwise, the blocks to which $(i,j)$ and $(i,j+1)$ belong  would have been merged at time $\eta'$. This means that the derivative of $y_{i,j}^\eta$ is always bounded from above by the derivative of $y_{i,j+1}^{\eta}$. Thus, such a pair of coordinates will never violate the monotonicity property. On the other hand, if $y_{i,j}^{\eta}<y_{i,j+1}^\eta$, then, as the derivatives are bounded, there always exists $\oeta>\eta'$ for which this strict inequality still holds (and thus monotonicity is not violated). Hence, we get that $\eta^3>\eta'$.

Finally, to see that $\eta^{4}>\eta'$ as well, we note that if $\sum_{i,j}y_{i,j}^{\eta'}=kd-\kappa(t)$, then, by design, it will remain equal henceforth. If, however, the latter sum is larger than $kd-\kappa(t)$ at time $\eta'$, then it would still remain so for some $\oeta>\eta'$.

Thus, indeed we have $\eta''=\min\{\min_\ell \eta^\ell, \oeta^0,\oeta^1,1\}>\eta'$, concluding the proof of the claim.
\end{proof}

In light of the above claim, one can consider obtaining a feasible configuration $\oy^1$ from the starting (feasible) configuration $\oy^0$ by simply gluing together the trajectories corresponding to the horizons. More precisely, one could start with $\eta_0=0$, and for each $\eta_s$, with $s\geq 0$ and $\eta_s<1$, define $\eta_{s+1}$ to be the horizon of $\oy^{\eta_s}$. Note that, as we start with the feasible configuration $\oy^{0}$, Claim \ref{cl:nbd} implies that all $\oy^{\eta_s}$ are well defined and feasible too.

Now, the only reason why the above approach might not end up giving us the desired feasible configuration $\oy^{1}$ is that, a priori, it is not clear whether the sequence $\{\eta_s\}_{s}$ ever reaches $1$. That is, even though we know that $\eta_0=0$ and $\eta_{s+1}>\eta_{s}$, it might still be possible that this sequence converges without ever reaching $1$, and thus there is no $s$ with $\eta_s=1$.

In order to rule out this possibility, we will prove that the total number of horizons is always finite. Observe that each horizon  can be associated with at least one of the following {\em events}: \eqref{cond:newproc1} the set of active coordinates changes, or \eqref{cond:newproc2} a block merge occurs, or \eqref{cond:newproc3} the number of servers hits the quota. Thus, it suffices to show that the total number of such events is bounded.

To this end, let us note that in our evolution, once $\sum_{i,j} y_{i,j}^\eta$ becomes equal to $kd-\kappa(t)$, $N(\eta)$ is chosen so that
$\sum_{i,j} {d y_{i,j}^\eta}/{d\eta}=0$. %
So, once we hit the quota we stay there throughout the rest of the hit stage. Hence, there can be at most one event of type \eqref{cond:newproc3}. Also, as we have already argued, during our evolution blocks never split once they are formed, and thus the total number of block merges (i.e., events of type \eqref{cond:newproc2}) can be at most $k$.

It remains to bound the number of events of type \eqref{cond:newproc1}, i.e., the ones corresponding to variables becoming active/inactive. For notational convenience, let us say that a coordinate $(i,j)$ {\em $0$-inactivates} (respectively, {\em $1$-inactivates}) at time $\eta_s$, for some $s\geq 1$, if $N(\eta_s)\leq\alpha \lambda_{i,j}^{\eta_s}$ (respectively, $N(\eta_s)>\alpha \lambda_{i,j}^{\eta_s}$) and $y^{\eta_s}_{i,j}=0$ (respectively, $y^{\eta_s}_{i,j}=1$), but $(i,j)$ was active at time $\eta_{s-1}$. We prove the following claim.

\begin{claim}\label{cl:n}
$N(\eta)$ can increase only at a horizon, i.e., for any $s\geq 0$ with $\eta_s<1$, $N(\eta_s)\geq N(\eta)$ for $\eta\in [\eta_s,\eta_{s+1})$.
Furthermore, if $N(\eta)$ indeed increases at time $\eta_{s+1}$, then at  $\eta_{s+1}$ we have an occurrence of either a block merge, or the quota is hit, or a $1$-inactivation of some $(i,j)$ with $y_{i,j}^{\eta_s}<1$.
\end{claim}

\begin{proof}
First, consider the case where at time $\eta_s$ the number of servers is still below the quota. By \eqref{eq:n_general2}, it means that $N(\eta)=0=N(\eta_s)$, for $\eta\in (\eta_s,\eta_{s+1})$, and $N(\eta_{s+1})=0$ unless the quota is hit at time $\eta_{s+1}$. So, the claim follows in this case, and in the rest of the proof we can assume that the number of servers is already  at the quota at time $\eta_{s}$.

First, we prove that $N(\eta_s)\geq N(\eta)$ for $\eta\in [\eta_s,\eta_{s+1})$, i.e., the first part of the claim. Let us fix some $\eta\in [\eta_s, \eta_{s+1})$. Note that by the definition of the horizon, we have that $A^{\eta}=A^{\eta_s}$ and $\lambda^{\eta}=\lambda^{\eta_s}$. So, for our purposes, it suffices to show that whenever $A^{\eta}=A^{\eta^+}=A$, and $\lambda^\eta=\lambda^{\eta^+}=\lambda$, for $\eta^+=\eta+d\eta$, we have that $N(\eta)\geq N(\eta^+)$. This follows immediately from Claim \ref{cl:n_nonincr}.

Now, to prove the second part of the claim, let us assume that none of the events mentioned in the statement of the claim occurred at time $\eta_{s+1}$, otherwise we are already done. So, we have, in particular, that $\lambda_{i,j}^{\eta_s}=\lambda_{i,j}^{\eta_{s+1}}=\lambda_{i,j}$ for each $(i,j)$. This implies that if there is an $(i,j)$ with $y_{i,j}^{\eta_{s+1}}=y_{i,j}^{\eta_s}=1$ that becomes active at time $\eta_{s+1}$,  then by Definition \ref{def-active} it must be the case that
\[
N(\eta_{s})>\alpha \lambda_{i,j}^{\eta_s}=\alpha \lambda_{i,j}^{\eta_{s+1}}\geq N(\eta_{s+1}).
\]
So, in this case $N(\eta_s)\geq N(\eta_{s+1})$, and thus we can restrict ourselves to the scenario in which the only coordinates $(i,j)$ that become active at time $\eta_{s+1}$ have $y_{i,j}^{\eta_{s+1}}=0$. As a result, we have
\[
A^{\eta_{s+1}}=(A^{\eta_s}\setminus (A_-^0\cup A_-^1)) \cup A_+,
\] where $A_{-}^0$ is the set of coordinates $(i,j)$ such that $(i,j)$ $0$-inactivates at time $\eta_{s+1}$, $A_{-}^1$ contains $(i,j)$-s which $1$-inactivate at that time and $y_{i,j}^{\eta_s}=1$, and $A_{+}$ is the set of $(i,j)$-s that become active at time $\eta_{s+1}$ with $y_{i,j}^{\eta_{s+1}}=0$.

Now, observe that if some $(i,j)\in A_{-}^1$, then we need to have $\alpha \lambda_{i,j}\geq N(\eta_{s})$. Otherwise, $(i,j)$ would have already been inactive at time $\eta_{s}$. Furthermore, we actually need to have $\alpha \lambda_{i,j}=N(\eta_s)$, as otherwise the derivative of $y_{i,j}^\eta$ would be negative in the interval $[\eta_s,\eta_{s+1})$, contradicting the fact that $y_{i,j}^{\eta_{s+1}}=1$. (Recall that we have already proved that $N(\eta)$ -- and thus all the derivatives -- do not increase in the interval $[\eta_s,\eta_{s+1})$.)

So, by the above, and Definition \ref{def-active}, we can conclude that
\begin{eqnarray}\label{eq:lll1}
\alpha \lambda_{i,j} \geq N(\eta_{s+1}) & \qquad & \mathrm{for\ each} \quad (i,j)\in A_{-}^0,\\
\alpha \lambda_{i,j} = N(\eta_s) & \qquad & \mathrm{for\ each} \quad (i,j)\in A_{-}^1,\\\label{eq:lll2}
\alpha \lambda_{i,j} < N(\eta_{s+1}) & \qquad & \mathrm{for\ each} \quad (i,j)\in A_{+}.\label{eq:llll3}
\end{eqnarray}

On the other hand, we can express $N(\eta_{s+1})$ as the weighted average of $\alpha \lambda_{i,j}$s over the set $A^{\eta_{s+1}}$
(cf. Claim \ref{cl:n_nonincr}), i.e. we have
\begin{equation*}
N(\eta_{s+1})=\frac{\sum_{(i,j)\in A^{\eta_{s+1}}}u_{i,j}^{\eta_{s+1}}\cdot
\alpha
\lambda_{i,j}}{\sum_{(i,j)\in A^{\eta_{s+1}}} u_{i,j}^{\eta_{s+1}}},
\end{equation*}
where $u_{i,j}^\eta=\frac{1}{w_i}\left(y^{\eta}_{i,j}+\beta\right)$.

Now, as $A^{\eta_{s+1}}=(A^{\eta_s}\setminus (A_-^0\cup A_-^1)) \cup A_+$, we can utilize conditions \eqref{eq:lll1} and \eqref{eq:llll3} to bound $N(\eta_{p+1})$ from above by a corresponding weighted average of $\alpha \lambda_{i,j}$-s over the set $A^{\eta_{s}}\setminus A_-^1$. In particular, we have
\begin{eqnarray}\label{eq:n_non_1}
N(\eta_{s+1}) & = & \frac{\sum_{(i,j)\in A^{\eta_{s+1}}}u_{i,j}^{\eta_{s+1}}\cdot
\alpha
\lambda_{i,j}}{\sum_{(i,j)\in A^{\eta_{s+1}}} u_{i,j}^{\eta_{s+1}}}=\frac{\sum_{(i,j)\in ((A^{\eta_s}\setminus (A_-^0\cup A_-^1)) \cup A_+)}u_{i,j}^{\eta_{s+1}}\cdot
\alpha
\lambda_{i,j}}{\sum_{(i,j)\in ((A^{\eta_s}\setminus (A_-^0\cup A_-^1)) \cup A_+)} u_{i,j}^{\eta_{s+1}}}\nonumber \\ & \leq & \frac{\sum_{(i,j)\in (A^{\eta_s}\setminus A_{-}^1)}u_{i,j}^{\eta_{s+1}}\cdot
\alpha
\lambda_{i,j}}{\sum_{(i,j)\in (A^{\eta_s}\setminus A_{-}^1)} u_{i,j}^{\eta_{s+1}}},
\end{eqnarray}
where the last inequality follows as for any $a,b>0$,
$\frac{a}{b}\leq \frac{a+c_1t_1 - c_2t_2}{b+t_1-t_2}$,
whenever $c_1\geq \frac{a}{b}$, $c_2 \leq \frac{a}{b}$, and  $t_1\geq 0$, $b>t_2\geq 0$.

If the last expression in \eqref{eq:n_non_1} is at most $N(\eta_s)$, then we are already done. So, let us assume, for the sake of contradiction, that it is strictly larger than $N(\eta_s)$. In this case, by \eqref{eq:lll2}, we need to have that also
\begin{equation}\label{eq:n_non_2}
\frac{\sum_{(i,j)\in A^{\eta_s}}u_{i,j}^{\eta_{s+1}}\cdot
\alpha
\lambda_{i,j}}{\sum_{(i,j)\in A^{\eta_s}} u_{i,j}^{\eta_{s+1}}}> N(\eta_s),
\end{equation}
as for any $a,b,c,t>0$, if $\frac{a}{b}>c$, then also $\frac{a+ct}{b+t}>c$.
However, by applying Claim \ref{cl:n_nonincr} with $A=A^{\eta_s}$, we have that
\[
\frac{\sum_{(i,j)\in A^{\eta_s}}u_{i,j}^{\eta_{s+1}}\cdot
\alpha
\lambda_{i,j}}{\sum_{(i,j)\in A^{\eta_s}} u_{i,j}^{\eta_{s+1}}}\leq \frac{\sum_{(i,j)\in A^{\eta_s}}u_{i,j}^{\eta_{s}}\cdot
\alpha
\lambda_{i,j}}{\sum_{(i,j)\in A^{\eta_s}} u_{i,j}^{\eta_{s}}}=N(\eta_s),
\]
contradicting \eqref{eq:n_non_2}, and thus proving that indeed $N(\eta_s)\geq N(\eta_{s+1})$. The proof of the claim is now concluded.
\end{proof}
Now, we are ready to bound the number of events of type \eqref{cond:newproc1}. To show that the number of these events is finite, it suffices to show that the number of $0$-inactivations and $1$-inactivations is finite. Also, observe that during the period in which the number of servers is below the quota, by definition, we have $N(\eta)=0$, and thus  variables can only decrease. As a result, coordinates can only $0$-inactivate in that period, and once they become inactive they stay that way. Hence, we have at most $kd$ such events.

In light of the above, we can focus on analyzing the events after reaching the quota. Note that in this case we can assume that $N(\eta)>0$. (If $N(\eta)=0$, then all derivatives are equal to zero, and the desired bounds trivially follow.) As we have that $\lambda_{i,j}^{\eta}$ is always zero when $i\neq \oi$, $N(\eta)>0$ implies that coordinates $(i,j)$ with $i\neq \oi$ can only increase, and once they $1$-inactivate they stay inactive. As a consequence, it suffices to show that the number of $0$-inactivations and $1$-inactivations is finite for all coordinates $(i,j)$ with $i=\oi$. In order to do so, we prove the following claim.

\begin{claim}
\label{cl:1-inactive} The total number of $1$-inactivations of coordinates $(\oi,j)$ is finite.
\end{claim}
\begin{proof}
We will prove the claim first for $j=k$ and then consider consecutive $j$-s in decreasing order. As a result, our task is to prove for a given $j$, that $(\oi,j)$ $1$-inactivates a finite number of times, provided that the number of $1$-inactivations is finite for all coordinates $(\oi,j')$ with $j<j'\leq k$.

To this end, we argue that whenever there are two consecutive $1$-inactivations of some coordinate $(\oi,j)$ -- the first one at time $\eta_{s'}$, and the second one at time $\eta_{s''}$ -- then in the interval $[\eta_{s'}, \eta_{s''}]$ we have either a block merge, or the quota is hit, or a $1$-activation of a coordinate $(i',j')$ that has either $i'\neq \oi$ or $j'>j$. As we know that the number of occurrences of each of these events is finite, we get the desired proof.

To establish the above, let $\eta_{\os}$ for $s'<\os<s''$ be the time in which $(\oi,j)$ is activated between the two $1$-inactivations. Observe that as $N(\eta_{\os})\leq \alpha \lambda_{\oi,j}$ and $N(\eta_{s''})>\alpha \lambda_{\oi,j}$, there is a time $\eta_{s^*}$ with $\os<s^*\leq s''$ in which $N(\eta)$ increases above $\alpha \lambda_{\oi,j}$. (Recall that by Claim \ref{cl:n} we know that $N(\eta)$ can increase only at horizons.) Without loss of generality we take $s^*$ to be the first $s>\os$ corresponding to such an increase.

Now, in light of Claim \ref{cl:n}, we know that $N(\eta)$ increases at time $\eta_{s^*}$. Thus, to conclude our proof it suffices to show that if we have a coordinate $(i',j')$ that $1$-inactivates at time $\eta_{s^*}$ and $y_{i',j'}^{\eta_{s^*-1}}<1$, then we cannot have $i'=\oi$ and $j'\leq j$.

We consider two cases here. The first one corresponds to $s^*< s''$. In this case we have $y_{\oi,j}^{\eta_{s^*}}<1$, as otherwise $(\oi,j)$ would be $1$-inactivated already at time $s^*$, instead of $s''$. However, by the monotonicity property \eqref{eq:y_2_consistency}, we have that $y_{\oi,j''}^{\eta_{s^*}}\leq y_{\oi,j}^{\eta_{s^*}}$ for all $j''\leq j$. So, if $i'=\oi$, then $(i',j')$ cannot $1$-inactivate at time $\eta_{s^*}$ if $j'\leq j$, and the claim follows.

Consider now the remaining case of $s^*=s''$. If we have $i'=\oi$ and $j'\leq j$, then we must have  $\lambda_{\oi,j'}\geq \lambda_{\oi,j}$. Otherwise, condition \eqref{eq:mergeblocks} for block merge would trigger at time $\eta_{s''}$. As a result, by \eqref{eq:y_evol_general2}, we know that in the interval of our interest the derivatives of $y_{\oi,j'}^{\eta}$ are bounded from above by the derivatives of $y_{\oi,j}^{\eta}$.

Furthermore, we have that the derivative of $y_{\oi,j}^{\eta}$ is always non-positive for $\eta\in [\eta_{\os},\eta_{s''})$. This is so, as by the definition of $s^*$, $N(\eta)\leq \alpha \lambda_{\oi,j}$ for $\eta\in [\eta_{\os},\eta_{s^*})=[\eta_{\os},\eta_{s''})$. As a consequence, we must have both $y_{\oi,j}^{\eta}$ and $y_{\oi,j'}^{\eta}$ to be equal to $1$ for all $\eta\in [\eta_{\os},\eta_{s''})$, as otherwise these variables would not be able to reach $1$ at time $\eta_{s''}$. This, however, contradicts the fact that $y_{i',j'}^{\eta_{s^*-1}}$ has to be strictly smaller than $1$, as $s^*-1=s''-1\geq \os$. Thus, we cannot have $i'=\oi$ and $j'\leq j$ and our claim follows.
\end{proof}
Finally, it remains to bound the number of $0$-activations of coordinates $(\oi,j)$. We do this by simply noting that if there are two consecutive $0$-activations of some coordinate $(\oi,j)$, then $N(\eta)$ has to increase at least once between these two events. But, by Claim \ref{cl:nbd}, it means that one of the events, (whose total number is already bounded), would also occur in this period. Therefore, the number of $0$-activations is also finite and we can conclude the proof of Lemma \ref{obs:nosplit}.

\section{Proof of Lemma \ref{l:optcost_to_hc}}\label{app:mey}
Here we prove Lemma \ref{l:optcost_to_hc}. As mentioned earlier, the proof is implicit in the work of \cite{CMP08}, and we make it explicit here for completeness. We begin with some notation, and state another result that we need.

Let $M$ be an arbitrary metric space. Let $C[0]$ denote the configuration specifying the initial location of the $k$-servers. We assume that the servers are labeled, so for every $k' \leq k$, the first $k'$ entries of $C[0]$ specify the location of the first $k'$ servers. Let $\rho$ be some fixed $k$-server request sequence. Let
$\opt(k',X)$ denote the optimum cost of serving $\rho$ with $k'$ servers on $M$, starting in $C[0]$ and ending in configuration $X$
(for notational ease, we are suppressing the dependence on $\rho,M,C[0]$ here).
Let $\opt(k') = \min_X \opt(k',X)$,  denote the minimum cost of server $\rho$ starting in $C[0]$.

\begin{lemma}[\cite{CMP08}, Corollary 2]
\label{l:qconvex}
Let $\rho$ be some fixed request sequence and $C[0]$ be some fixed initial configuration.
For any $k_1,k_2 \in [k]$, given any state $X$ on $k_1$ locations, there exists another state $Y$ such that
\begin{enumerate}
\item $ |X \cap Y| = \min(|X|,|Y|)$, i.e. $Y$ overlaps with $X$ as much as possible, and
\item  $\opt(k_2,Y) \leq   \opt(k_2) + \opt(k_1,X) - \opt(k_1)$. That is, the excess cost incurred for an optimum $k_2$-server solution to end in $Y$, is no more than the excess cost incurred for the optimum $k_1$-server to end in $X$.
\end{enumerate}
\end{lemma}
This lemma and its proof can be found in \cite{CMP08} (Corollary 2).

Let $T$ be a weighted $\sigma$-HST. Again, for notational convenience, let us drop $\rho$, $C[0]$, and the underlying metric $T$ from the notation (these remain the same, and dropping them will not cause any confusion). Given a quota pattern $\kappa$, recall the definition of  $\optcost(\kappa,t)$ as the optimum cost of serving $\rho$ until time $t$ with quota pattern $\kappa$. We also use $\optcost(\kappa) = \optcost(\kappa,\infty)$ to denote the optimum cost of serving the entire sequence $\rho$.
As previously, let us define $h^t(\kappa) = \optcost(\kappa(t) \cdot \vec{1},t) - \optcost(\kappa(t) \cdot \vec{1},t-1)$ and $g(\kappa) = \sum_{t \geq 1} |\kappa(t)- \kappa(t-1)|$.
Let $D$ denote the diameter of $T$.

We will prove the following, which is the same as Lemma \ref{l:optcost_to_hc}. (In that notation, note that $\Delta \leq W(p)/(\sigma-1)$ for $T(p)$.)
\begin{thm}
$$ \sum_t h^t (\kappa) -  \Delta \cdot g(\kappa) \leq \optcost(\kappa)  \leq \sum_t h^t (\kappa) +  \Delta \cdot g(\kappa).$$
\end{thm}
\begin{proof}
We do an induction on the value of $g(\kappa)$. In the base case, when $g(\kappa)=0$,
the vector $\kappa$ is constant throughout, say $\kappa =  k \cdot \vec{1}$.
In this case, the claimed result holds trivially as the sum over $h^t$ telescopes and we obtain
$$ \sum_t h^t(\kappa) = \sum_t (\optcost(k \cdot \vec{1},t) - \optcost(k \cdot \vec{1},t-1)) =  \optcost(\kappa).$$

So, let $\kappa$ be such that $g(\kappa)>0$.
Let $\tau$ be the earliest time when $\kappa(\tau) \neq \kappa(\tau+1)$. Define a new quota pattern $\kappa'$ as
\begin{displaymath}
\kappa'(t) = \left\{ \begin{array}{ll}
            \kappa(\tau+1) & \textrm{if $t \leq \tau$}\\
            \kappa(t)      &   \textrm{if $t > \tau$.}
\end{array} \right.
\end{displaymath}
Note that both $\kappa$ and $\kappa'$ are constant for $t \leq \tau$. Also,
 $g(\kappa') = g(\kappa) - |\kappa(\tau+1) - \kappa(\tau)| < g(\kappa)$, and hence we can inductively
assume that the claimed result holds for $\kappa'$.

We first show that $ \optcost(\kappa)  \geq \sum_t h^t (\kappa) -  D \cdot g(\kappa)$.
Fix some solution that attains cost $\optcost(\kappa)$ and let $C[t]$ denote its configuration at time $t$.
Applying Lemma \ref{l:qconvex} with $X=C(\tau)$ and $k_1=\kappa(\tau)$ and $k_2= \kappa'(\tau)$, there is some configuration $Y$ satisfying
\begin{eqnarray} \optcost(k_2 \cdot \vec{1},\tau,Y) & \leq &  \optcost(k_2 \cdot \vec{1}, \tau) + \optcost(k_1 \cdot \vec{1},\tau,X) - \optcost(k_1 \cdot \vec{1},\tau)
\label{convexineq} \\
\label{maxovlap}
 |X \cap Y| & = &  \min(k_1,k_2).
 \end{eqnarray}

We construct a solution $S'$ corresponding to $\kappa'$ as follows:
Until time $\tau$, $S'$ follows the solution $\optcost(k_2 \cdot \vec{1},\tau,Y)$. Then, after serving the request at $t=\tau$, it switches to state $C[\tau]=X$, and henceforth for $t > \tau$ sets its the configurations $C'[t] = C[t]$.
Now $$\textrm{cost}(S') = \optcost(k_2\cdot \vec{1},\tau,Y) + c(Y,X)  + Q,$$
where $c(Y,X)$ is the cost of moving from state $Y$ to $X$, and $Q$ is the contribution of the solution $\optcost(\kappa)$ starting from time $\tau$ and state $X$ (recall that $C[\tau]=X$).

It is easily checked that the solution $S'$ constructed above is feasible for quota pattern $\kappa'$.
As the optimum solution for $\kappa'$ can only be better,  $\textrm{cost}(S') \geq \optcost(\kappa')$ and since $\optcost(\kappa') \geq
\sum_t h^t (\kappa') -  D \cdot g(\kappa')$ by the inductive hypothesis, it follows that
\begin{equation}
\label{eq:c'1} \sum_t h^t (\kappa') -  D \cdot g(\kappa') \leq  \textrm{cost}(S')  = \optcost(k_2\cdot \vec{1},\tau,Y) + c(Y,X)  + Q.
\end{equation}
By \eqref{maxovlap}, $|Y\cap X| = \min(k_1,k_2)$ and hence  $c(Y,X) \leq D |k_2 - k_1| = D | \kappa(\tau+1) - \kappa(\tau)|$.
Thus $c(Y,X) + D\cdot g(\kappa') \leq D \cdot g(\kappa)$, and hence \eqref{eq:c'1}
implies that
\begin{equation}
\label{eq:c'2} Q \geq \sum_t h^t (\kappa') -  D \cdot g(\kappa) - \optcost(k_2\cdot \vec{1},\tau,Y).
\end{equation}
On the other hand, as $X= C[\tau]$ we have
$$\optcost(\kappa) = \optcost(\kappa,\tau,X) + Q = \optcost(k_1\cdot \vec{1},\tau,X) + Q.$$
Thus by \eqref{eq:c'2},
\begin{eqnarray}
  \optcost(\kappa)  & \geq &  \optcost(k_1\cdot \vec{1},\tau,X)- \optcost(k_2\cdot \vec{1},\tau,Y) + \sum_t h^t (\kappa') -  D \cdot g(\kappa) \nonumber \\
  & \geq &  \optcost(k_1 \cdot \vec{1},\tau) - \optcost(k_2 \cdot \vec{1}, \tau)  + \sum_t h^t (\kappa') -  D \cdot g(\kappa)
  \label{eq:hitopt1}\\
  & = & \sum_t h^t (\kappa) -  D \cdot g(\kappa) \label{eq:hitopt2},
   \end{eqnarray}
implying the desired lower bound. Here \eqref{eq:hitopt1} follows from \eqref{convexineq}, and \eqref{eq:hitopt2} follows from \eqref{eq:hitopt1} since
\begin{eqnarray*}\sum_{t=1}^\tau h^t(\kappa)& = &\optcost(k_1 \cdot \vec{1},\tau), \\
\sum_{t=1}^\tau h^t(\kappa') &=& \optcost(k_2 \cdot \vec{1},\tau),\\
\end{eqnarray*}
and $\sum_t h^t (\kappa)  - \sum_t h^t (\kappa') = \sum_{t=1}^\tau h^t(\kappa) - \sum_{t=1}^\tau h^t(\kappa')$, since $\kappa$, and $\kappa'$ are the same for any $t>\tau$.

We now show show the upper bound on $\optcost(\kappa)$. The proof is similar to the one above.
Let $\kappa'$ be defined as previously. Let $\{C'[t]\}_t$ denote the configurations for some fixed solution that has value $\optcost(\kappa')$.
Applying Lemma \ref{l:qconvex} with $X=C'[\tau]$,$k_1=\kappa'(\tau)$ and $k_2= \kappa(\tau)$, we obtain a configuration $Y$ satisfying
\begin{eqnarray} \optcost(k_2 \cdot \vec{1},\tau,Y) & \leq &  \optcost(k_2 \cdot \vec{1}, \tau) + \optcost(k_1 \cdot \vec{1},\tau,X) - \optcost(k_1 \cdot \vec{1},\tau)
\label{convexineqb} \\
 |X \cap Y| & = &  \min(k_1,k_2).
 \end{eqnarray}
Consider the following solution $S$ corresponding to $\kappa$:
Until time $\tau$, $S$ mimics the solution $\optcost(k_2 \cdot \vec{1},\tau,Y)$. Then, after serving the request at $t=\tau$, it switches to state $C'[\tau]=X$, and henceforth for $t > \tau$ sets its the configurations $C[t] = C'[t]$.
Now $$\textrm{cost}(S) = \optcost(k_2\cdot \vec{1},\tau,Y) + c(Y,X)  + Q,$$
where $Q$ is the cost of solution $\optcost(\kappa')$ incurred from time $\tau$ starting at state $Y$.
Again $S$ is feasible for quota $\kappa$, and hence $\optcost(\kappa) \leq \textrm{cost}(S)$. By definition of $X$,
$$\optcost(\kappa') =  \optcost(k_1 \cdot \vec{1},\tau,X) + Q.$$
Thus,
 $$\optcost(\kappa) \leq  \optcost(k_2\cdot \vec{1},\tau,Y) + c(Y,X) + \optcost(\kappa') - \optcost(k_1 \cdot \vec{1},\tau,X).$$
As $\optcost(\kappa') \leq \sum_t h^t (\kappa') +  D \cdot g(\kappa')$ by the inductive hypothesis
and  $c(Y,X) + D \cdot g(\kappa') \leq D \cdot g(\kappa)$, we obtain
\begin{eqnarray}
\optcost(\kappa) & \leq &     \optcost(k_2\cdot \vec{1},\tau,Y) + \sum_t h^t (\kappa') +  D \cdot g(\kappa)
 - \optcost(k_1 \cdot \vec{1},\tau,X) \\
&  \leq  &  \optcost(k_2 \cdot \vec{1},\tau) - \optcost(k_1 \cdot \vec{1}, \tau)  + \sum_t h^t (\kappa') -  D \cdot g(\kappa)
  \label{eq:hitopt1b}\\
  & = & \sum_t h^t (\kappa) -  D \cdot g(\kappa) \label{eq:hitopt2b},
   \end{eqnarray}
implying the desired inequality.
Here \eqref{eq:hitopt1b} follows from \eqref{convexineqb}, and \eqref{eq:hitopt2b} follows by noting that
$$\sum_t h^t (\kappa')  - \sum_t h^t (\kappa) = \sum_{t=1}^\tau h^t(\kappa') - \sum_{t=1}^\tau h^t(\kappa) = \optcost(k_1 \cdot \vec{1},\tau) - \optcost(k_2 \cdot \vec{1},\tau).$$

\end{proof}

\end{document}